\documentclass{article}
\usepackage[utf8]{inputenc}
\usepackage{float} 
\usepackage{amsmath}
\usepackage{amssymb}
\usepackage{array}
\usepackage{bbm}
\usepackage{makeidx}
\usepackage{hyperref}
\usepackage{cleveref}
\usepackage{caption}
\usepackage{subcaption}
\usepackage{amsthm}
\usepackage{color}
\usepackage{mathtools}
\usepackage{graphicx} 
\usepackage[LGR,T1]{fontenc}
\usepackage{siunitx}
\usepackage{mathrsfs}
\usepackage{scalerel}
\usepackage{forest}
\usepackage{algorithm,algpseudocode}
\usepackage{lipsum}
\usepackage{booktabs}
\usepackage{siunitx}
\usepackage{xcolor}
\usepackage[title]{appendix}
\hypersetup{
    colorlinks,
    linkcolor={blue!80!black},
    citecolor={green!50!black},
}
\usepackage{accents}
\usepackage{apptools}
\usepackage{tikz}
\usetikzlibrary{arrows, positioning, automata}
\usepackage{comment}
\usepackage[shortlabels]{enumitem}

\usepackage[mathscr]{euscript}
\DeclareSymbolFont{rsfs}{U}{rsfs}{m}{n}
\DeclareSymbolFontAlphabet{\mathscrsfs}{rsfs}

\def\Ons{\textsf{\textsf{OC}}}
\def\<{\langle}
\def\>{\rangle}

\AtAppendix{\counterwithin{lemma}{section}}

\usepackage{mathtools}

\newtheoremstyle{myremark} 
    {\topsep}                    
    {\topsep}                    
    {\rm}                        
    {}                           
    {\bf}                        
    {.}                          
    {.5em}                       
    {}  

\newtheorem{claim}{Claim}[section]
\newtheorem{lemma}[claim]{Lemma}

\newtheorem{setting}{Setting}

\newtheorem{theorem}{Theorem}

\newtheorem{corollary}[claim]{Corollary}

\theoremstyle{myremark}
\newtheorem{remark}{Remark}[section]

\allowdisplaybreaks
\usepackage[top=1in,bottom=1in,left=1in,right=1in]{geometry}
\usepackage[utf8]{inputenc}
\def\mmse{\mathsf{mmse}}
\def\dd{\mathrm{d}}
\def\ed{\stackrel{{\mathrm d}}{=}}

\def\top{\intercal}

\def\bA{{\boldsymbol A}}
\def\bB{{\boldsymbol B}}
\def\bnu{\boldsymbol{\nu}}

\def\cS{{\mathcal S}}

\def\htheta{\hat{\Theta}}

\def\cT{\mathcal{T}}

\def\cF{\mathcal{F}}

\def\cS{\mathcal{S}}

\def\cF{\mathcal{F}}

\def\oW{\overline{W}}
\def\oX{\overline{X}}

\def\NN{\mathbb{N}}

\def\PP{\mathbb{P}}

\def\RR{\mathbb{R}}

\def\calL{\mathcal{L}}

\def\calS{\mathcal{S}}

\def\bA{\mathbf{A}}
\def\bB{\mathbf{B}}

\def\bD{\mathbf{D}}

\def\bG{\mathbf{G}}

\def\bS{\mathbf{S}}

\def\bV{\mathbf{V}}
\def\bW{\mathbf{W}}
\def\bX{\mathbf{X}}

\def\bZ{\mathbf{Z}}

\def\ba{\boldsymbol{a}}
\def\bb{\boldsymbol{b}}

\def\bg{\boldsymbol{g}}
\def\bh{\boldsymbol{h}}

\def\bk{\boldsymbol{k}}

\def\bq{\boldsymbol{q}}
\def\br{\boldsymbol{r}}
\def\bs{\boldsymbol{s}}

\def\bu{\boldsymbol{u}}
\def\bv{\boldsymbol{v}}
\def\bw{\boldsymbol{w}}
\def\bx{\boldsymbol{x}}
\def\by{\boldsymbol{y}}
\def\bz{\boldsymbol{z}}

\def\bA{\boldsymbol{A}}
\def\bB{\boldsymbol{B}}

\def\bD{\boldsymbol{D}}

\def\bG{\boldsymbol{G}}

\def\bS{\boldsymbol{S}}

\def\bV{\boldsymbol{V}}
\def\bW{\boldsymbol{W}}
\def\bX{\boldsymbol{X}}

\def\bZ{\boldsymbol{Z}}
\def\hbtheta{\hat{\boldsymbol \theta}}
\def\normal{{\mathsf{N}}}

\def\bmu{\boldsymbol{\mu}}
\def\btheta{\boldsymbol{\theta}}
\def\bvtheta{\boldsymbol{\vartheta}}

\def\bSigma{\boldsymbol{\Sigma}}

\def\bzero{\boldsymbol{0}}
\def\bfone{\boldsymbol{1}}
\def\balpha{\boldsymbol{\alpha}}

\def\bTheta{\boldsymbol{\Theta}}

\def\reals{{\mathbb R}}
\def\sT{{\sf T}}

\def\hTheta{\hat{\Theta}}
\def\id{{\boldsymbol I}}

\DeclareMathOperator*{\plim}{p-lim}

\DeclareMathOperator*{\pliminf}{p-liminf}

\renewcommand{\P}{\mathbb{P}}
\newcommand{\E}{\mathbb{E}}
\newcommand{\R}{\mathbb{R}}

\newcommand{\eps}{\varepsilon}
\newcommand{\Var}{\operatorname{Var}}

\newcommand{\argmin}{\operatorname{argmin}}

\newcommand{\Cov}{\operatorname{Cov}}
\newcommand{\sign}{\operatorname{sign}}

\newcommand{\diag}{\operatorname{diag}}

\newcommand{\RN}[1]{%
  \textup{\uppercase\expandafter{\romannumeral#1}}%
}

\newcommand\iidsim{\stackrel{\mathclap{iid}}{\sim}}

\newcommand{\RNum}[1]{\uppercase\expandafter{\romannumeral #1\relax}}



\makeatletter
\newcommand*{\rom}[1]{\expandafter\@slowromancap\romannumeral #1@}
\makeatother
\title{Statistically Optimal First Order Algorithms:\\ 
A Proof via Orthogonalization}
\author{
	Andrea Montanari\footnotemark[2]
	\thanks{Department of Electrical Engineering, Stanford University} 
	\and 
	Yuchen Wu\thanks{Department of Statistics, Stanford University}
}
\date{}
\begin{document}
\maketitle

\begin{abstract}
	We consider a class of statistical estimation problems in which we are given 
	a random data matrix $\bX\in\R^{n\times d}$ (and possibly some labels $\by\in\R^n$)
	and would like to estimate a coefficient vector $\btheta\in\R^d$ (or possibly
	a constant number of such vectors). Special cases include low-rank matrix estimation
	and regularized estimation in generalized linear models (e.g., sparse regression).
	First order methods proceed by iteratively multiplying current estimates by $\bX$ or 
	its transpose. Examples include gradient descent or its accelerated variants.
	
	Celentano, Montanari, Wu  \cite{celentano2020estimation} proved that for any constant 
	number of iterations (matrix vector multiplications), the optimal first order algorithm is 
	a specific approximate message passing algorithm (known as `Bayes AMP'). The error of this estimator
	can be characterized in the high-dimensional asymptotics $n,d\to\infty$, 
	$n/d\to\delta$, and provides a lower bound to the estimation error of any first order algorithm.
	Here we present  a simpler proof of the same result, and generalize it to broader classes 
	of data distributions and of first order algorithms, including algorithms with non-separable
	nonlinearities. Most importantly, the new proof technique does not require to construct an
	equivalent tree-structured estimation problem, and is therefore susceptible of a 
	broader range of applications.
\end{abstract}

\section{Introduction}

In this note we study high-dimensional estimation in a class of problems in which 
the data consists of a high dimensional matrix $\bX\in\reals^{n\times n}$ (symmetric)
or $\bX\in\reals^{n\times d}$ (asymmetric), and, possibly, a vector of labels $\by\in\reals^n$. 
More precisely, we consider two cases: $(i)$~Low-rank matrix estimation, whereby 
$\bX = \frac{1}{n} \btheta \btheta^{\top} + \bW$ with $\bW$ a noise matrix, and we would like to estimate
$\btheta \in\reals^{n}$; $(ii)$~Generalized linear models, whereby 
$y_i = h(\btheta^{\sT}\bx_i;w_i)$ with $\bx_i$ the $i$-th row of $\bX$ and $w_i$
a noise variable, and we would like to estimate $\btheta\in\reals^{d}$. 

The recent paper \cite{celentano2020estimation} introduced a class of `generalized first order methods'
(GFOM) to perform estimation efficiently. Informally, GFOMs proceed iteratively.
At time $t$, the state of the algorithm is given by order $t$ vectors of dimension $n$
or $d$ (which we can think of as estimates of $\btheta$). A new vector is computed
by applying a nonlinear function to these vectors (independent of the data)
and then multiplying the result by $\bX$ or $\bX^{\sT}$. This class of algorithm is broad enough to
include classical first order methods from optimization theory \cite{nesterov2003introductory}, 
such as gradient descent, accelerated gradient descent,
and mirror descent with respect to a broad class of objective functions (both convex and nonconvex).

Given this setting, a natural question is: 
\begin{quote}
\begin{center}
\emph{What is the optimal estimation algorithm among all GFOMs?}
\end{center}
\end{quote}
This question was answered in \cite{celentano2020estimation} under the assumption that the noise matrix $\bW$
 (in the case of low-rank  matrix estimation) or the covariates matrix $\bX$ 
 (for regression in generalized linear models) has i.i.d. normal entries, and under some regularity 
assumptions on the algorithm iterations.
 Namely, \cite{celentano2020estimation} proves that in the proportional
asymptotics $n,d\to\infty$, $n/d\to\delta \in (0, \infty)$, optimal estimation error
is achieved, for any fixed number of iterations $t$, by the Bayes approximate message 
passing (AMP) algorithm. Also this algorithm choice is unique up to reparametrizations.

The proof of \cite{celentano2020estimation} was based on three steps:
\begin{enumerate}
\item[$(I)$] \emph{Reduction.} Any GFOM can be simulated by a certain AMP algorithm, with the same number of
matrix-vector multiplications, plus (eventually) a post-processing step that is independent of
data $\bX$.
\item[$(II)$] \emph{Tree model.} The estimation error achieved by an AMP algorithm after $t$ iterations  is 
asymptotically equivalent to the one achieved by a corresponding message passing algorithm for 
a certain estimation problem on a tree graphical model $T$
after $t$-iterations (this algorithm is $t$-local on the tree).
\item[$(III)$] \emph{Optimality on trees.} Belief propagation is the optimal $t$-local
algorithm for the estimation problem on $T$. As a consequence, Bayes AMP is the optimal first
order method in the original problem (since it achieves the same accuracy as belief propagation
in the tree model).
\end{enumerate}

The main objective of this note is to present a simpler proof of the optimality of 
Bayes AMP that does not take the detour of constructing the equivalent tree model.
Namely, steps $(II)$ and $(III)$ are replaced by the following. 
\begin{enumerate}
\item[$(II')$] \emph{Reduction to orthogonal AMP.} Any AMP algorithm can be simulated 
by a certain orthogonal AMP algorithm, which, after $t$ iterations, generates
$t$ vectors in $\reals^d$ or $\reals^n$ whose projections orthogonal to $\btheta$ are orthonormal.
The algorithm output at iteration $t$ is a function of these $t$ vectors, which is independent of
data $\bX$. 
\item[$(III')$] \emph{Optimality of Bayes AMP.} The asymptotic estimation error
 of the orthogonal AMP estimator is characterized via state evolution \cite{bayati2011dynamics}.
 By minimizing this error among orthogonal AMP algorithms, we obtain the error of Bayes AMP. 
\end{enumerate}

This proof strategy avoids several technicalities that arise because of the tree
equivalence steps and the analysis of belief propagation. 
Also, it is easier to generalize to different settings,
and indeed we establish the following generalizations of the result of \cite{celentano2020estimation}:
\begin{itemize}
\item We treat the case of noise matrices $\bW$ (for low-rank matrix estimation)
or $\bX$ (for regression) with independent entries, satisfying a bound on the fourth moment.
In contrast, the results of \cite{celentano2020estimation} were limited to Gaussian matrices.
\item In the Gaussian case, we cover the case in which the first order method applies,
 at each iteration, a general Lipschitz continuous nonlinearity to previous iterates.
 The only limitation is that this nonlinearity should be independent 
 from the data matrix $\bX$. In contrast, the results of \cite{celentano2020estimation} 
 were limited to separable nonlinearities (i.e. nonlinearities that act 
 row-wise to the previous iterates, see below).
\end{itemize}

In order to motivate our work, we will begin in Section 
\ref{sec:Motivation} by presenting a numerical experiment. 
We will carry out this experiment in the context of  phase
retrieval, since a large number of first order methods have been developed for
this problem.

We will next pass to explaining our new optimality results.
In order to present the new proof technique in the most transparent fashion,
we will devote most of the main text to the simplest possible example,
namely estimating a rank-one symmetric matrix from a noisy observation.
We will describe the setting and state our results in this context in 
Section \ref{sec:SymmetricMatr}. We then prove this result in
Section \ref{sec:proof} for the case of separable nonlinearities. Finally
section \ref{sec:MainRegression} presents our results for the case of regression. 
The appendices presents technical proofs for non-separable nonlinearities and
for the regression setting. 
These follow the same strategy as the proof in the main text with some modifications.

\section{An experiment: benchmarking algorithms for phase retrieval}
\label{sec:Motivation}

As a motivating example, we consider noiseless phase retrieval, 
in which we take measurements $y_i$ of an unknown signal $\btheta\in\reals^d$
according to:
\begin{align*}
	y_i = \< \bx_i, \btheta \>^2, \qquad i \in \{1, \cdots, n\}.
\end{align*}
We let $\bX \in \RR^{n \times d}$ with the $i$-th row being $\bx_i$ and $\by \in \RR^n$ with the 
$i$-th coordinate being $y_i$. We will consider the simple example of random measurements
$\bx_i \iidsim \normal(\mathbf{0}, \id_d / n)$ and assume the normalization 
$\|\btheta\|^2/d=1+o_d(1)$. Given $(\by, \bX)$, our goal is to recover $\btheta$. 
Since the signal $\btheta$ is real, `sign retrieval' would be a more appropriate name here.
We expect that an experiment with complex signal would yield similar results.

Needless to say, first order methods (with spectral initialization or not)
were studied in a substantial body of work, see among others
\cite{schniter2014compressive,candes2015phase,chen2017solving,cai2016optimal,wang2017solving,duchi2019solving,mondelli2018fundamental,waldspurger2018phase,ma2019optimization,maillard2020phase,fannjiang2020numerics,mondelli2021approximate}.

Apart from illustrating the content of our results, this section also demonstrates
a practical use of these results to benchmarking algorithms. 

\subsection{Spectral initialization}

As is common in the literature, we consider first order methods with a spectral initialization.
Since our main objective is to compare various first order methods, we will use a common 
spectral initialization developed in \cite{mondelli2018fundamental}, which is defined as follows.

We define $\bD_n\in\reals^{d\times d}$ as follows:
\begin{align*}
	\bD_n := \sum_{i=1}^n\cT(y_i)\bx_i\bx_i^{\sT}, 
\end{align*}
where $\cT: \RR \rightarrow \RR$ is a preprocessing function 
given in \cite[Eq. (137)]{mondelli2018fundamental}:
\begin{align}
\cT(y) = \frac{y-1}{y+\sqrt{1+\eps}-1}\, .
\end{align}
Here, $\eps>0$ can be taken arbitrarily, but in simulations we fix $\eps = 10^{-3}$.
We then use the initialization ${\btheta}^0 := \sqrt{d } \bv_1(\bD_n)$, where
$\bv_1(\bD_n)$ denotes the leading eigenvector of $\bD_n$. 
Without loss of generality, we assume $\langle {\btheta}^0, {\btheta} \rangle \geq 0$
(the overall sign of $\btheta$ cannot be estimated).
As shown in \cite{mondelli2018fundamental}, this initialization is optimal in the following sense.
Consider $n,d\to\infty$, with $n/d\to\delta$. 
For $\delta>1+\eps$, $\btheta^0$ achieves a positive correlation
with $\btheta$, with probability converging to one as $n,d\to\infty$. 
For $\delta<1$, no estimator can achieve a positive correlation.

In fact, for any $\delta>1$, the 
correlation between $\btheta^0$ and $\btheta$ converges in probability to a deterministic value that 
is given as  follows.  For $\lambda \in (1, \infty)$, we define the functions
\begin{align*}
	\phi(\lambda) := \lambda \E\left[ \frac{\cT(G^2)G^2}{\lambda - \cT(G^2)} \right], \qquad \psi(\lambda) := \frac{\lambda}{\delta} + \lambda \E\left[ \frac{\cT(G^2)}{\lambda - \cT(G^2)} \right]\, ,
\end{align*}
where expectation is with respect to  $G \sim \normal(0,1)$.
We let $\bar\lambda = \argmin_{\lambda \geq 1} \psi(\lambda)$ and,  for $\lambda \in (1, \infty)$, 
define $\zeta(\lambda) := \psi(\max(\lambda, \bar\lambda))$. Denote by 
$\lambda^{\ast}$ the unique solution of the equation $\zeta(\lambda) = \phi(\lambda)$ on $(1, \infty)$. 
Finally, let $a \geq 0$ be given by
\begin{align*}
	a^2 = \frac{\frac{1}{\delta} - \E\left[ \frac{\cT(G^2)^2}{(\lambda^{\ast} - \cT(G^2))^2} \right]}{\frac{1}{\delta} + \E\left[ \frac{\cT(G^2)^2(G^2 - 1)}{(\lambda^{\ast} - \cT(G^2))^2} \right]}.
\end{align*} 
Then, \cite[Lemma 2]{mondelli2018fundamental} proves that $|\<\btheta,\btheta^0\>|/d$ converges to
$a$ as $n,d\to\infty$.  Further, the approximate joint distribution
of these vectors is given by ${\btheta}^0 \approx {a} \btheta + \sqrt{{1 - a^2}} \bg$, 
in the sense that, for any Lipschitz function $\psi: \RR \rightarrow \RR$, 
\begin{align}
\plim_{n,d\to\infty}\frac{1}{d}\sum_{i=1}^{d}\psi\big(\theta^0_{i}-s\, a\theta_i\big)=
\E\big\{\psi(\sqrt{{1 - a^2}} G)\big\}\, .
\end{align}
(This follows from the convergence of the correlation  $|\<\btheta,\btheta^0\>|/d$, together 
with rotational invariance.). Here, $\plim$ denotes convergence in probability, $\bg \sim \normal(0, \id_d)$ and is independent of $\btheta$. Finally, \cite{mondelli2021approximate} shows that initializing
AMP at $\btheta^0$ is (asymptotically) equivalent to running a first order method 
from a warm start initialization independent of $\btheta^0$, and hence the analysis of the next 
sections apply to the present case.

\subsection{First order methods}\label{sec:algs}

We will consider three specific GFOMs for phase retrieval. 
GFOMs will only be introduced formally in Section \ref{sec:SymmetricMatr}
(for low-rank matrix estimation) and Section \ref{sec:MainRegression} (for regression, 
including phase retrieval as a special case). For this section, it is sufficient to say 
that GFOMs operate at each iteration by performing multiplication by $\bX$
or $\bX^{\sT}$ plus, eventually, applying a suitable 
nonlinear operation that is independent of $\bX$.

In the next subsection we will implement the algorithms listed below 
and compare their estimation error with the minimum error among all GFOMs.

\subsubsection*{Bayes AMP}

Bayes AMP is a special type of AMP algorithm and fits the general framework
of \cite{bayati2011dynamics}. The theory presented in Section \ref{sec:MainRegression} suggests 
that it is indeed optimal among all GFOMs.
 A detailed description and analysis of the Bayes AMP 
for phase retrieval is carried out in \cite{mondelli2021approximate}. 
Since 
the precise definition is somewhat technical and not needed for the rest of the 
paper, we omit it here and refer to  \cite{mondelli2021approximate}.

\begin{remark}
It is worth clarifying that ---despite the name--- Bayes AMP does not rely on Bayesian assumptions.

More precisely, the definition Bayes AMP requires specifying a nominal distribution $\mu^{\rm{AMP}}_{\Theta}$
for the entries of the true signal $\btheta$. Here, we are assuming $\btheta$ arbitrary 
(either deterministic or random) and such that $\|\btheta\|_2^2/d=1+o_d(1)$. By rotational invariance of the distribution of the
covariates $\bx_i$, we can achieve at any such $\btheta$ the same error as if $\btheta$ 
was uniformly distributed over the sphere of radius $\|\btheta\|_2$. For large $d$,
this is achieved by setting $\mu^{\rm{AMP}}_{\Theta}$ the standard normal distribution,
which is what we do here.
\end{remark}

\subsubsection*{Gradient descent}

If we attempt to  minimize the $\ell_2$ loss on the training dataset, we can derive the 
corresponding gradient descent algorithm:
\begin{align*}
	\btheta^{t + 1} = \btheta^t + \frac{4\eta\delta^2}{n}  \bX^{\sT}(\by - |\bX\btheta^t|^2) \odot(\bX \btheta^t) ,
\end{align*}
where $\eta > 0$ is the step size, $|\bX \btheta^t|^2 \in \RR^n$ is the vector whose
 $i$-th coordinate is $\langle \bx_i, \btheta^t \rangle^2$, and $\odot$ 
 denotes entrywise multiplication. 

\subsubsection*{Prox-linear algorithm}

The prox-linear algorithm was proposed in \cite{duchi2019solving}. 
The original algorithm sets $L := 2\|\bX\|_{\rm{op}}^2$  and proceeds by solving a 
sequence of sub-problems: 
\begin{align}\label{eq:prox-linear}
	\btheta^{t + 1} = \argmin_{\bvtheta\in\RR^d} \left\{\frac{L}{2}\|\bvtheta - \btheta^t\|^2_2 + \sum_{i = 1}^n \left|\langle \bx_i, \btheta^t \rangle^2 + 2\langle \bx_i, \btheta^t \rangle\langle \bx_i, \bvtheta - \btheta^t\rangle - y_i \right|
	\right\}.
\end{align}
Notice that this  is \emph{not} a GFOM, since each iteration requires solving an optimization problem,
and does not reduce to a pair of matrix-vector multiplications by  $\bX^{\sT}$ and $\bX$.

In order to obtain a first order algorithm we replace the full optimization
of  the subproblem by a single gradient step, with stepsize $\xi$:
\begin{align}\label{eq:1step-prox-linear}
	\btheta^{t + 1} = \btheta^t + {2\xi} \bX^{\top} (\bs^t\odot \bX \btheta^t),
	\;\;\;\;
	s_i^t:= \sign(y_i-\<\bx_i,\btheta^t\>^2)\, .
\end{align}
We will carry out simulations both with the prox-linear algorithm and the 1-step prox-linear algorithm.
It is however important to keep in mind that the comparison between prox-linear algorithm and 
GFOMs is unfair to GFOMs because each prox-linear step potentially requires a large number of 
matrix-vector multiplications.

\subsubsection*{Truncated amplitude flow (TAF)}
Truncated amplitude flow (TAF) was proposed in \cite{wang2017solving},
which claimed superior statistical performances with respect to state of the art. 
Following \cite{wang2017solving}, we fix parameters $\alpha = 0.6$, $\gamma = 0.7$. 
For $t \in \NN$, we define the set
\begin{align*}
	\mathcal{I}_t := \big\{i \in [n]: |\langle \bx_i, \btheta^t \rangle| \geq (1 + \gamma)^{-1} \sqrt{y_i}
	\big\}.
\end{align*}
At the $(t+1)$-th iteration, we perform the following update:
\begin{align*}
	\btheta^{t + 1} = \btheta^t - 
	\alpha \sum_{i \in \mathcal{I}_t}\left(\langle \bx_i, \btheta^t \rangle - 
	\sqrt{y_i} \sign(\langle\bx_i, \btheta^t \rangle) \rangle\right) \bx_i.
\end{align*}

\subsection{Simulation results}

\begin{figure}
     \centering
     \begin{subfigure}[b]{0.5\textwidth}
         \centering
         \includegraphics[width=\textwidth]{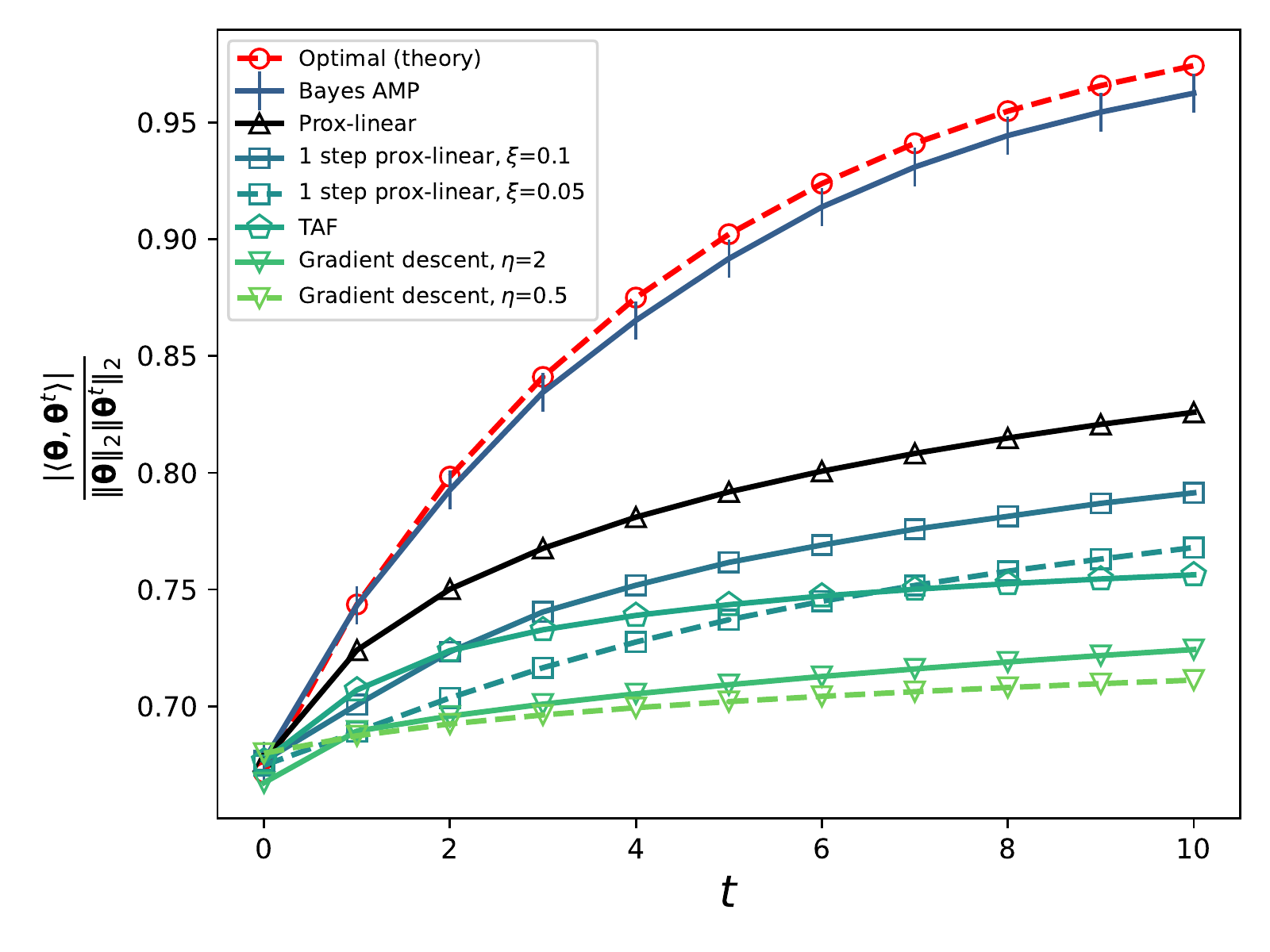}
         \caption{$n=600$.}
     \end{subfigure}
     \hspace{-1em}
     \begin{subfigure}[b]{0.5\textwidth}
         \centering
         \includegraphics[width=\textwidth]{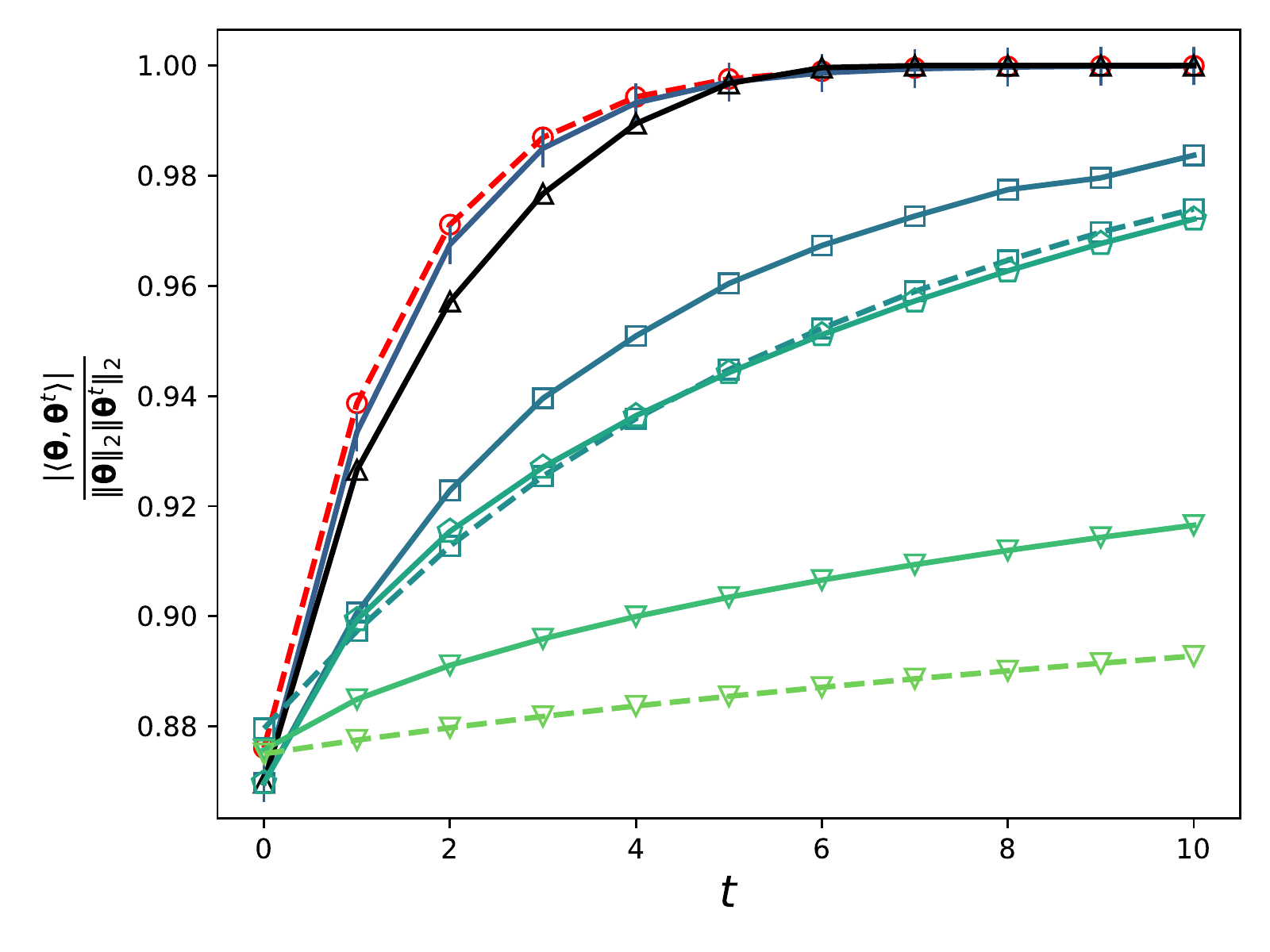}
         \caption{$n=1000$.}
     \end{subfigure}
        \caption{Correlation  $|\<\btheta^t,\btheta\>|/\|\btheta^t\|_2\|\btheta\|_2$
        for various algorithms, as a function of the number of iterations, for $d=400$. All
        algorithms are GFOMs with the exception of prox-linear.
         Red dashed lines represent the optimal correlation of Theorem \ref{thm:GLM}.}
        \label{fig:cor}
\end{figure}

\begin{figure}
     \centering
     \begin{subfigure}[b]{0.5\textwidth}
         \centering
         \includegraphics[width=\textwidth]{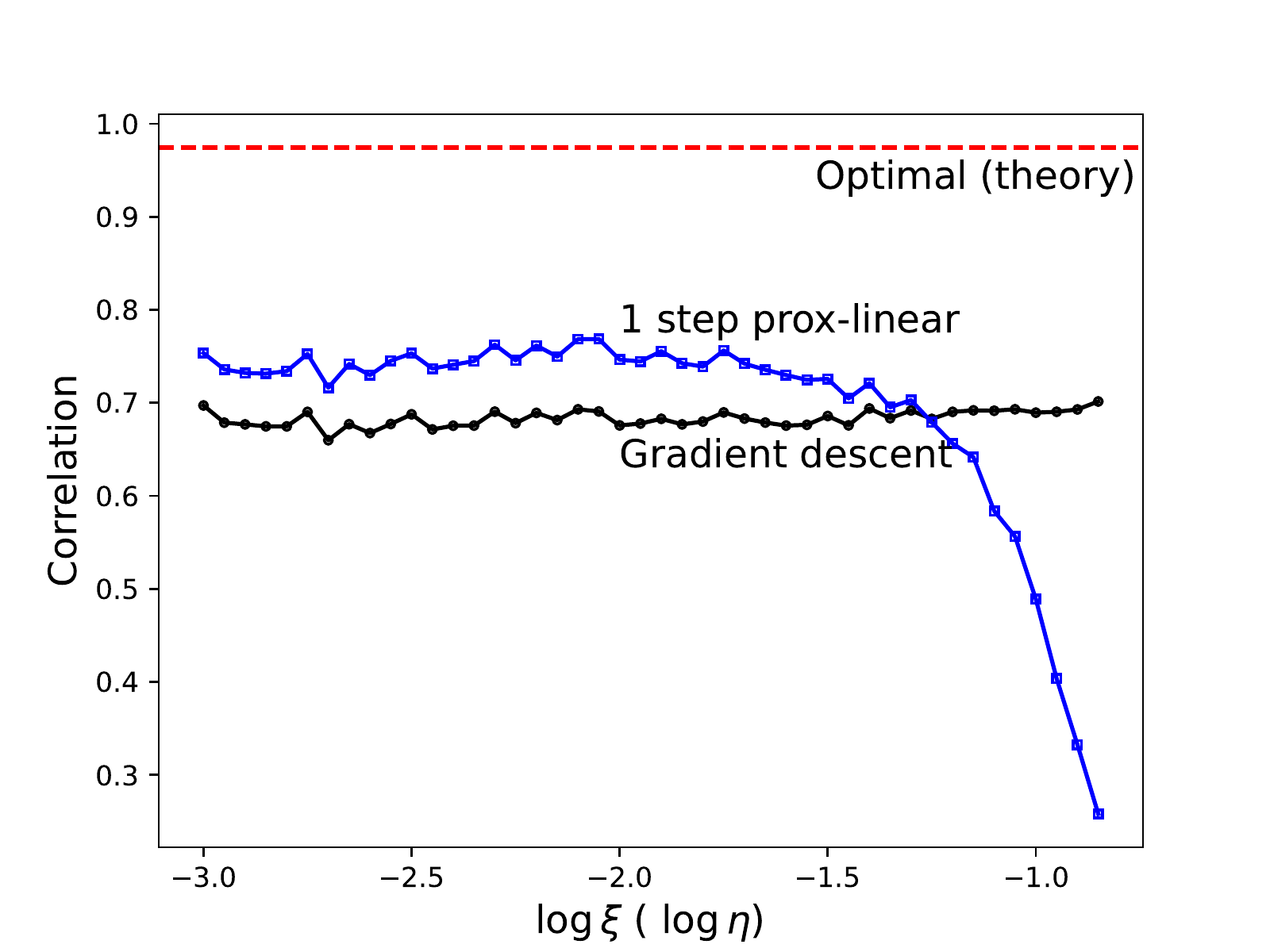}
         \caption{$n=600$.}
     \end{subfigure}
     \hspace{-2.5em}
     \begin{subfigure}[b]{0.5\textwidth}
         \centering
         \includegraphics[width=\textwidth]{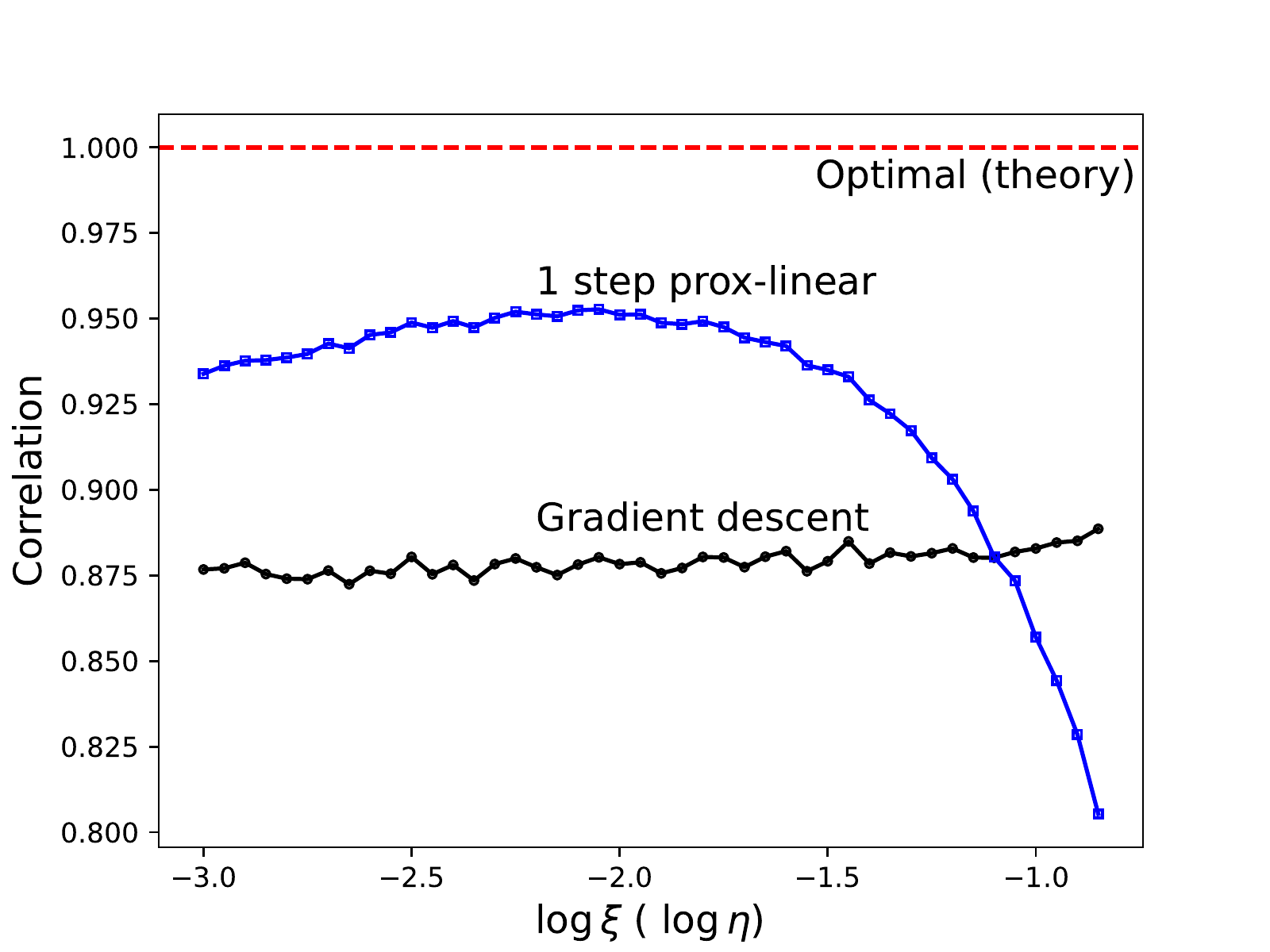}
         \caption{$n=1000$.}
     \end{subfigure}
        \caption{Performance of gradient descent and the one step prox-linear algorithm 
        with $t=10$ iterations as a function of the step sizes. The $x$ axis is the logarithm 
        of the step size $\eta$ (for gradient descent) or $\xi$ (for one step prox-linear algorithm). 
        The $y$ axis is the correlation $|\<\btheta^t,\btheta\>|/\|\btheta^t\|_2\|\btheta\|_2$. 
        Red dashed lines represent the optimal correlation of Theorem \ref{thm:GLM}. Results 
        are averaged over 50 independent trials.}
        \label{fig:step}
\end{figure}

In our first set of simulations, we take  $d = 400$, $n\in\{600,1000\}$, 
and run reconstruction experiments using each of the algorithms described above, 
averaging results over 50 independent trials. 
We compute the correlation between the estimates produced by these algorithms  and the true
signal $\btheta$, and plot the results in Figure \ref{fig:cor}, as a function of the number of iterations
$t\in\{0,1, \cdots, 10\}$. We also plot the theoretical prediction (cf. Theorem \ref{thm:GLM})
for the maximum achievable correlation by any GFOM.

\begin{table}
\centering
	\begin{tabular}{SSSSSS} \toprule
     & {Bayes AMP} & {Gradient descent} & {Prox-linear} & {1 step prox-linear} & {TAF} \\ \midrule
    {Wall clock time}  & {$1.83 \times 10^{-2}$} & {$6.63 \times 10^{-3}$} & {$ \mathbf {5.87 \times 10^1}$} & {$6.23 \times 10^{-3}$} & {$7.43 \times 10^{-3}$}  \\ \bottomrule
\end{tabular}
\caption{Averaged wall clock time for different algorithms. }
\label{table:time}
\end{table}

\begin{figure}
\centering
\begin{subfigure}{0.3\columnwidth}
\includegraphics[width=\textwidth]{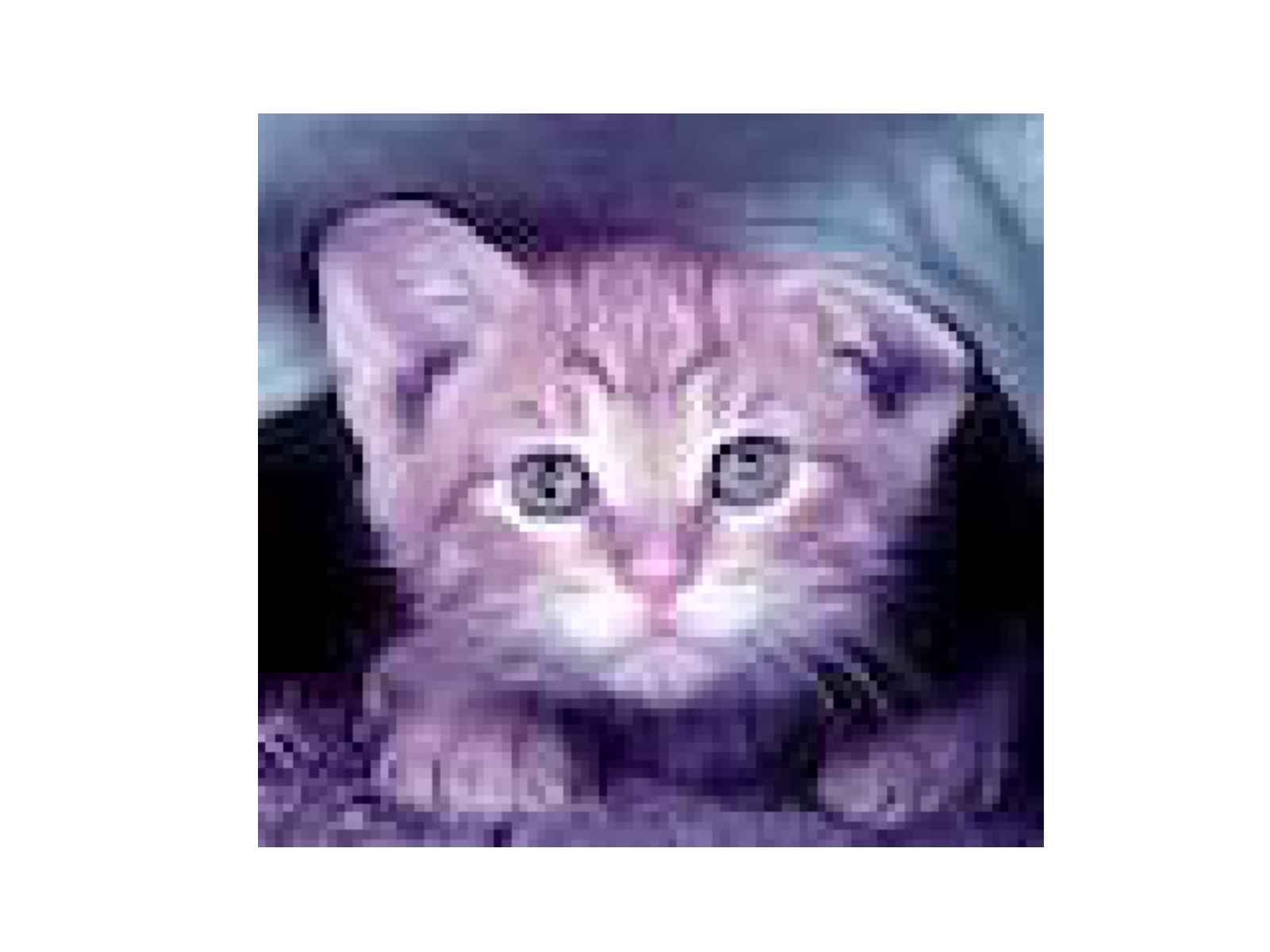}
\caption*{Original image.}
\end{subfigure} \\
\begin{subfigure}{0.3\columnwidth}
\includegraphics[width=\textwidth]{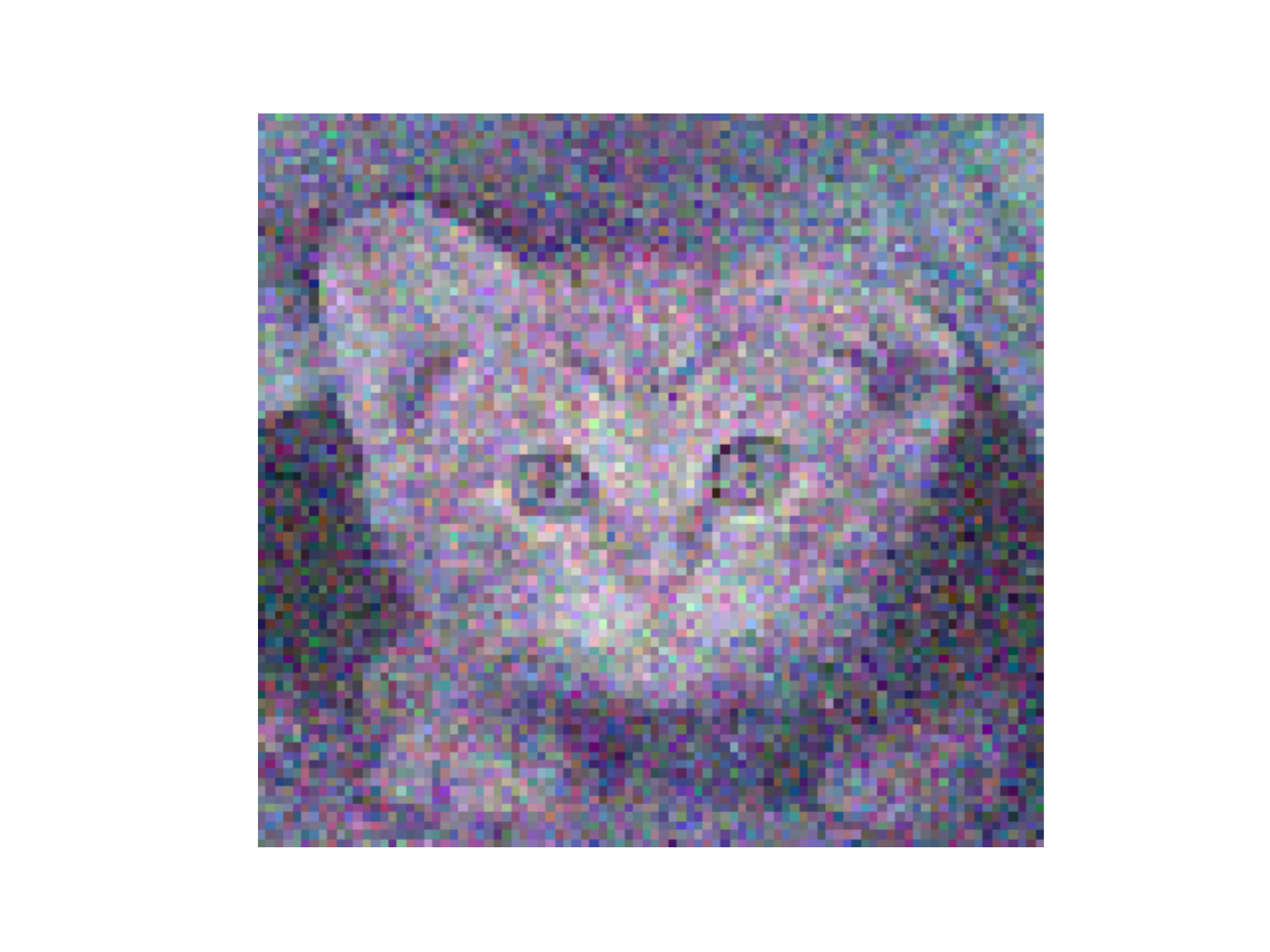}
\caption*{Bayes AMP, $t = 2$.}
\end{subfigure}
\hspace{-3em}
\begin{subfigure}{0.3\columnwidth}
\centering
\includegraphics[width=\textwidth]{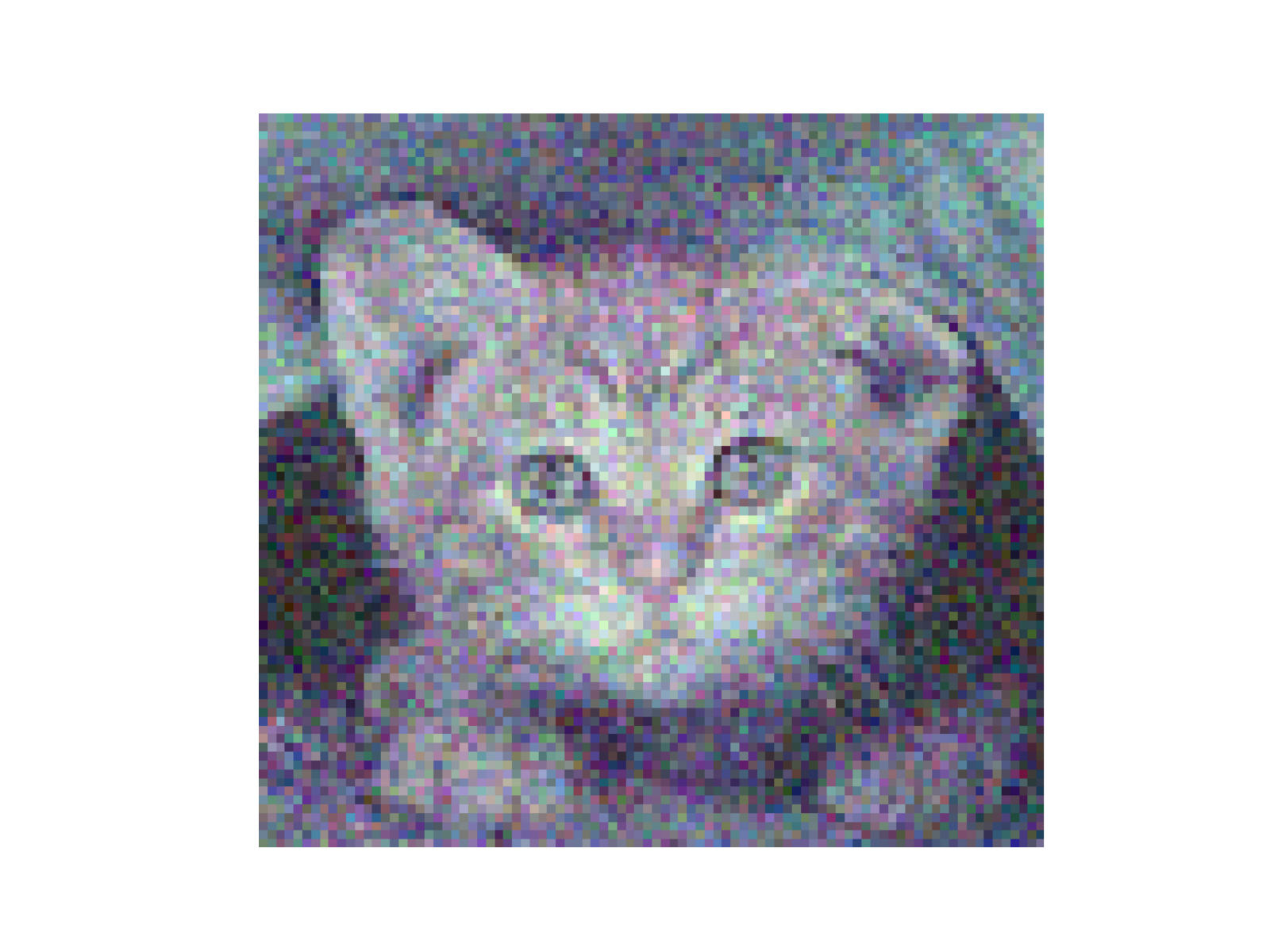}
\caption*{Bayes AMP, $t = 4$.}
\end{subfigure}
\hspace{-3em}
\begin{subfigure}{0.3\columnwidth}
\centering
\includegraphics[width=\textwidth]{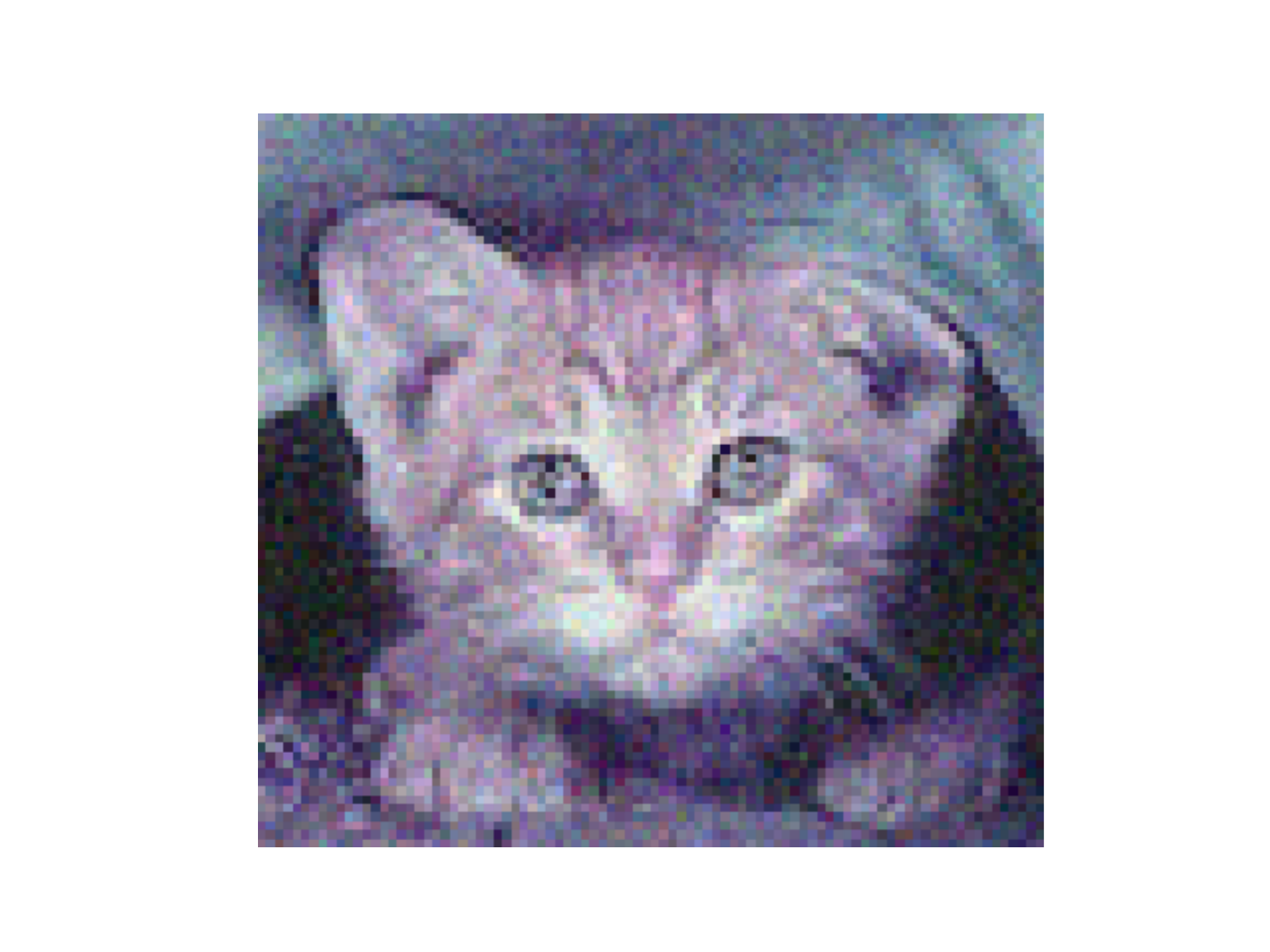}
\caption*{Bayes AMP, $t = 8$.}
\end{subfigure} \\
\begin{subfigure}{0.3\columnwidth}
\includegraphics[width=\textwidth]{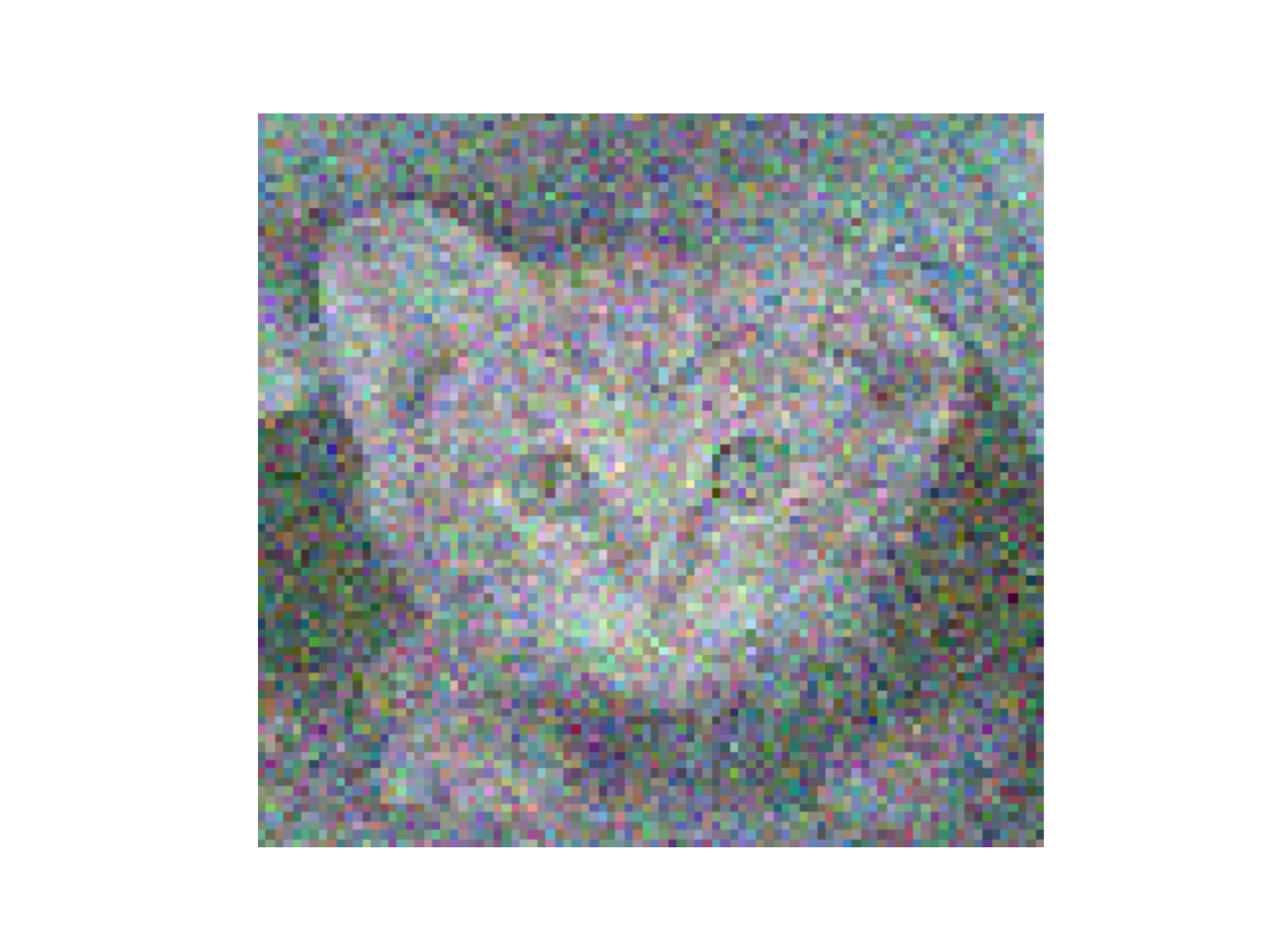}
\caption*{1 step prox-linear, $t = 2$.}
\end{subfigure}
\hspace{-3em}
\begin{subfigure}{0.3\columnwidth}
\centering
\includegraphics[width=\textwidth]{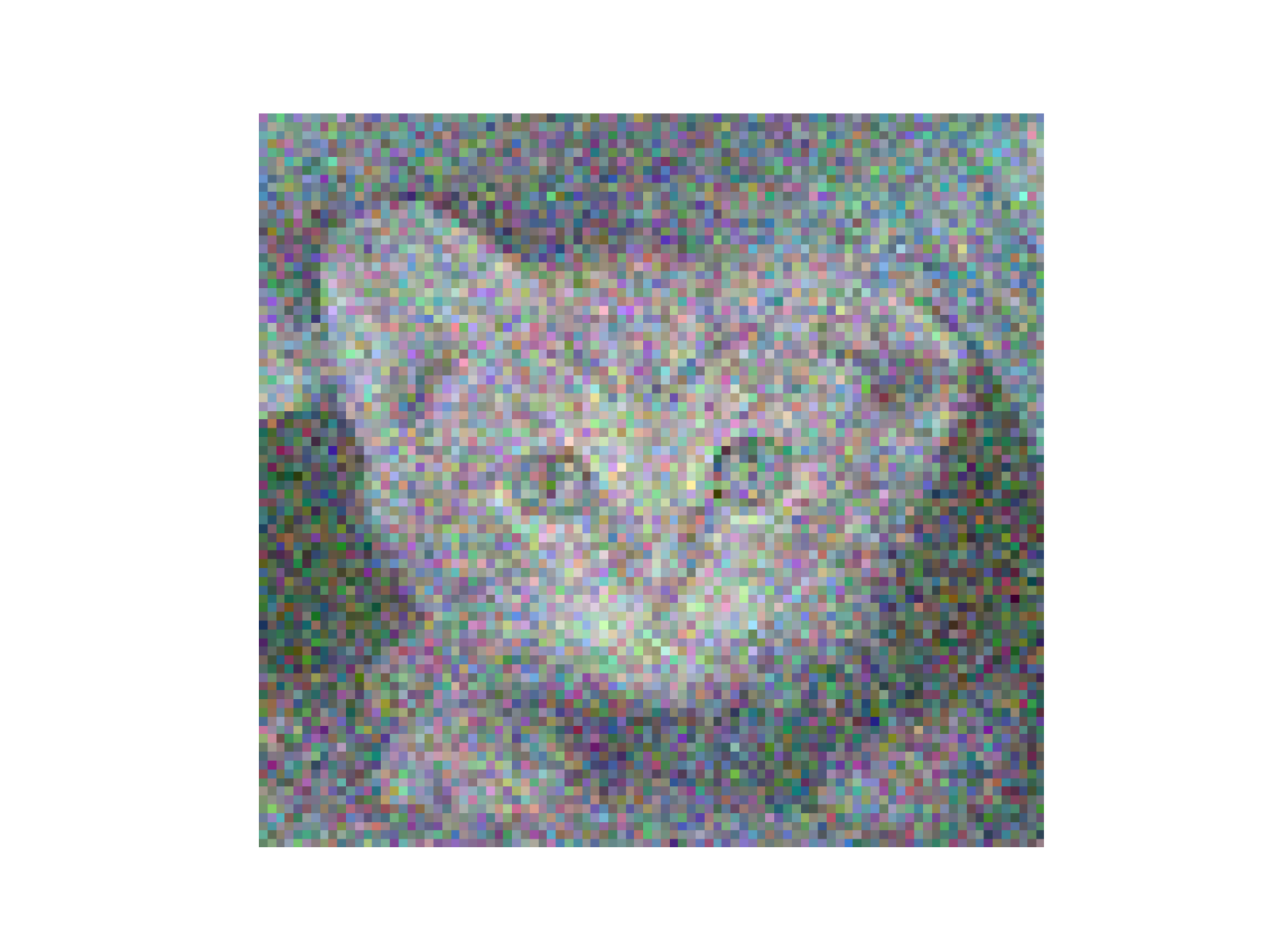}
\caption*{1 step prox-linear, $t = 4$.}
\end{subfigure}
\hspace{-3em}
\begin{subfigure}{0.3\columnwidth}
\centering
\includegraphics[width=\textwidth]{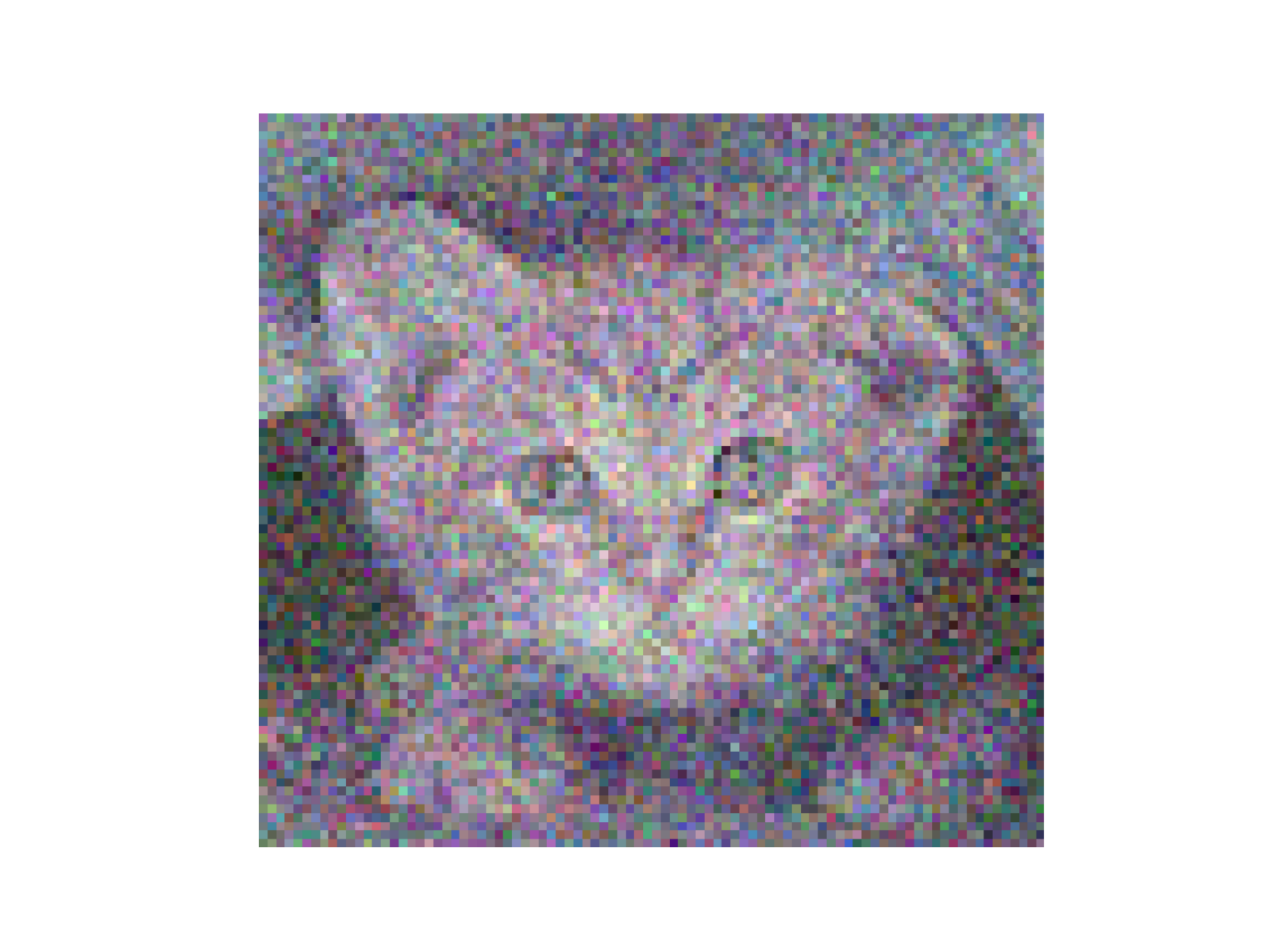}
\caption*{1 step prox-linear, $t = 8$.}
\end{subfigure} \\
\begin{subfigure}{0.3\columnwidth}
\includegraphics[width=\textwidth]{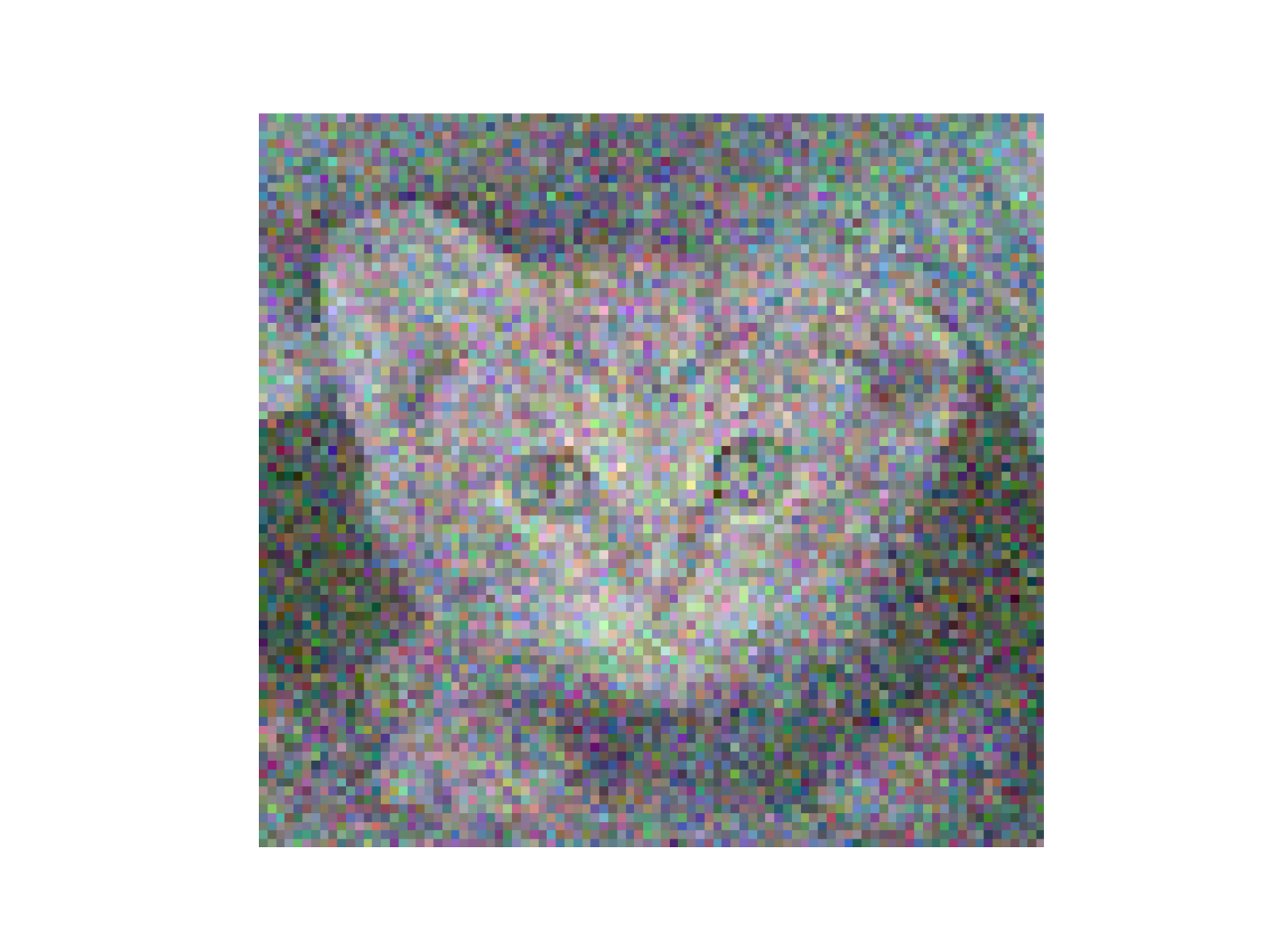}
\caption*{TAF, $t = 2$.}
\end{subfigure}
\hspace{-3em}
\begin{subfigure}{0.3\columnwidth}
\centering
\includegraphics[width=\textwidth]{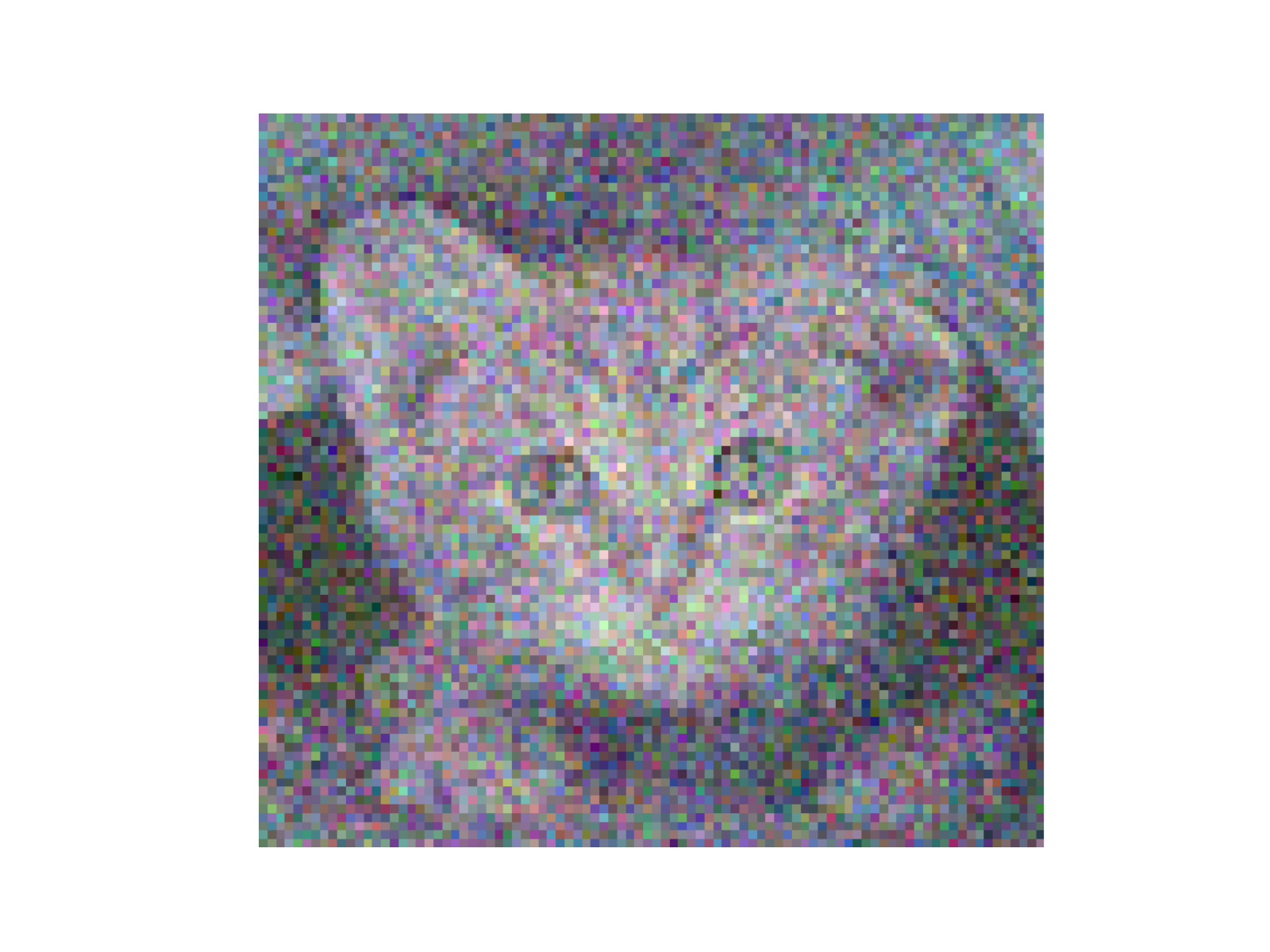}
\caption*{TAF, $t = 4$.}
\end{subfigure}
\hspace{-3em}
\begin{subfigure}{0.3\columnwidth}
\centering
\includegraphics[width=\textwidth]{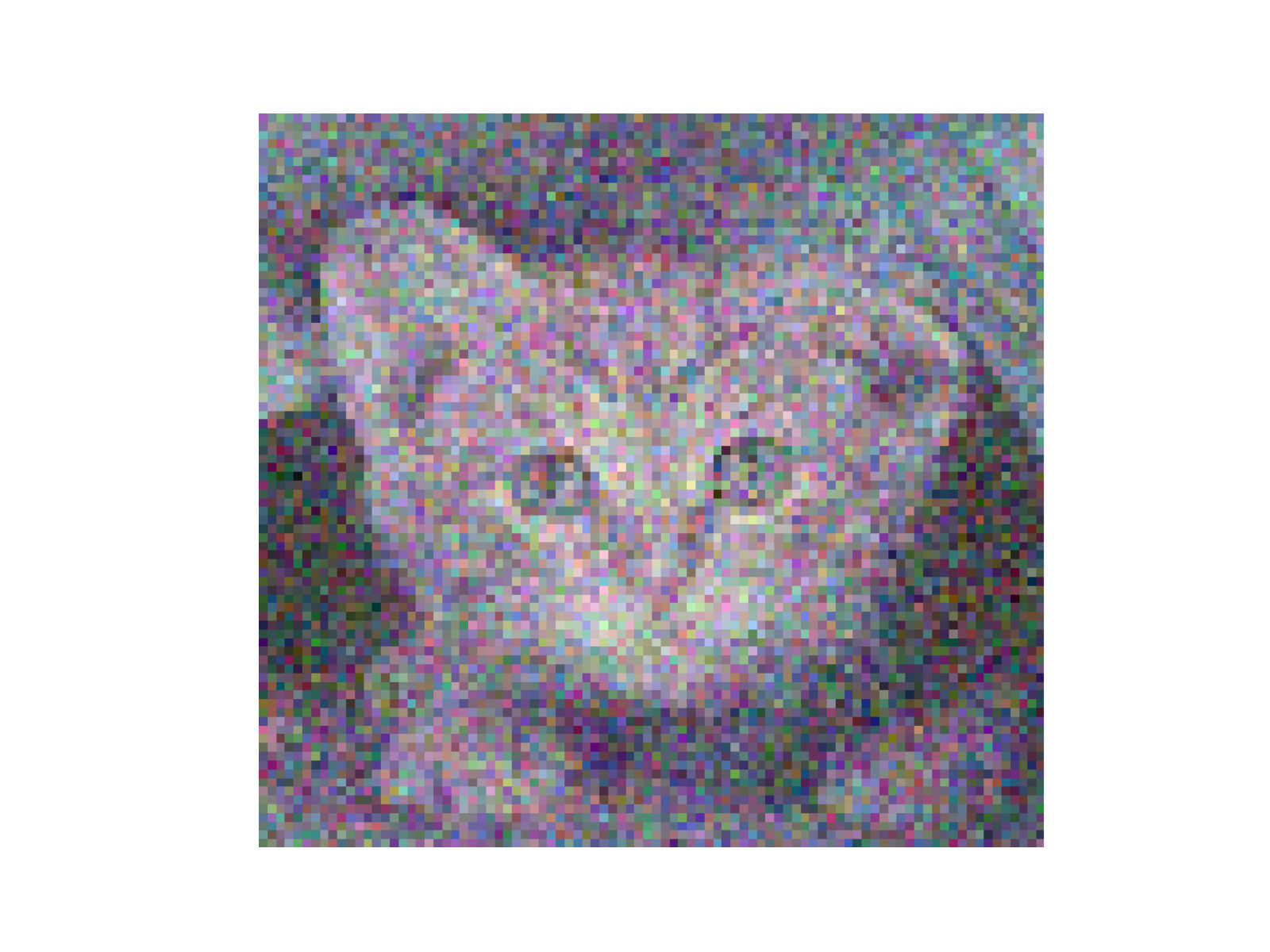}
\caption*{TAF, $t = 8$.}
\end{subfigure}\\
\begin{subfigure}{0.3\columnwidth}
\includegraphics[width=\textwidth]{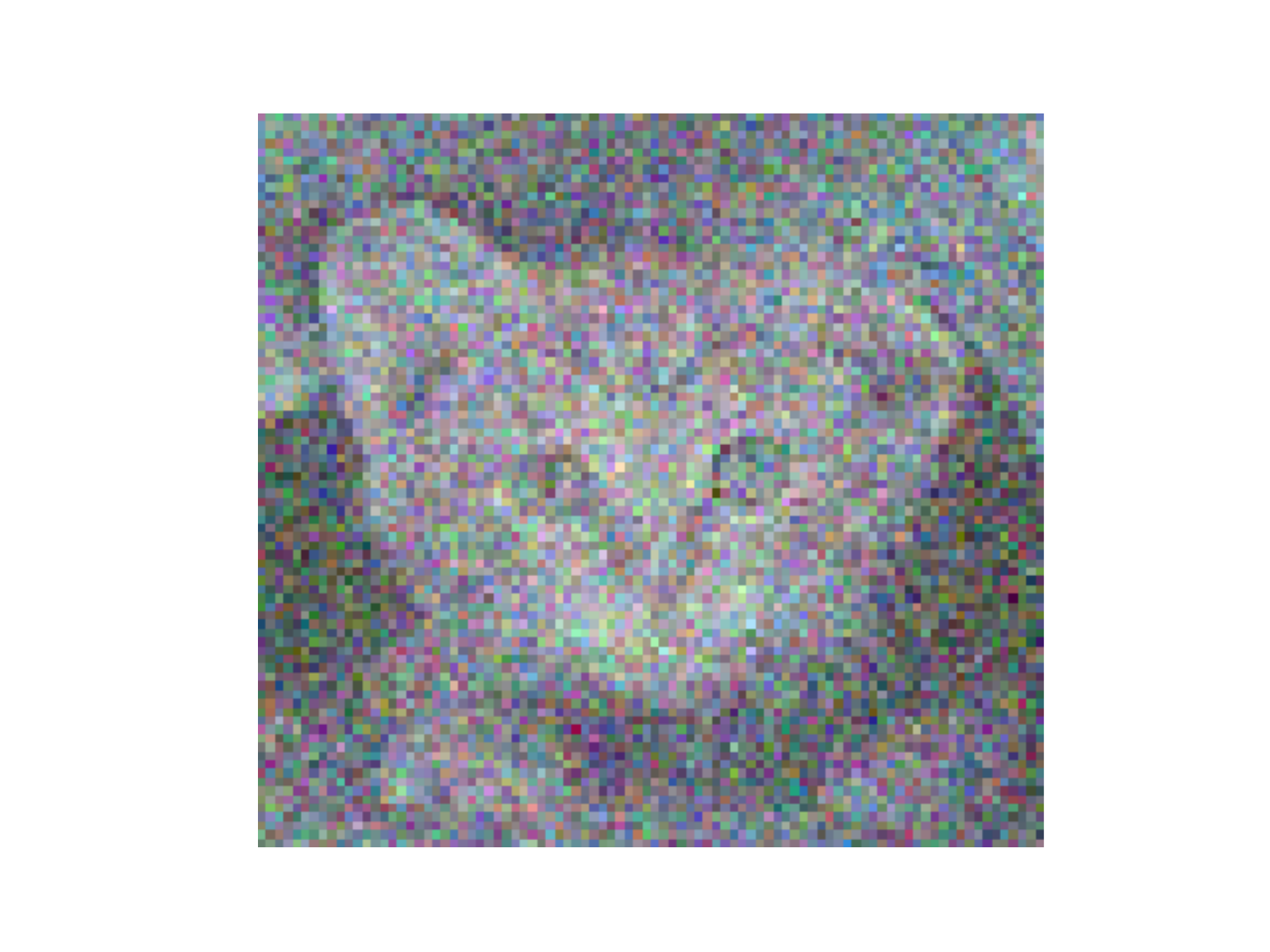}
\caption*{Gradient descent, $t = 2$.}
\end{subfigure}
\hspace{-3em}
\begin{subfigure}{0.3\columnwidth}
\centering
\includegraphics[width=\textwidth]{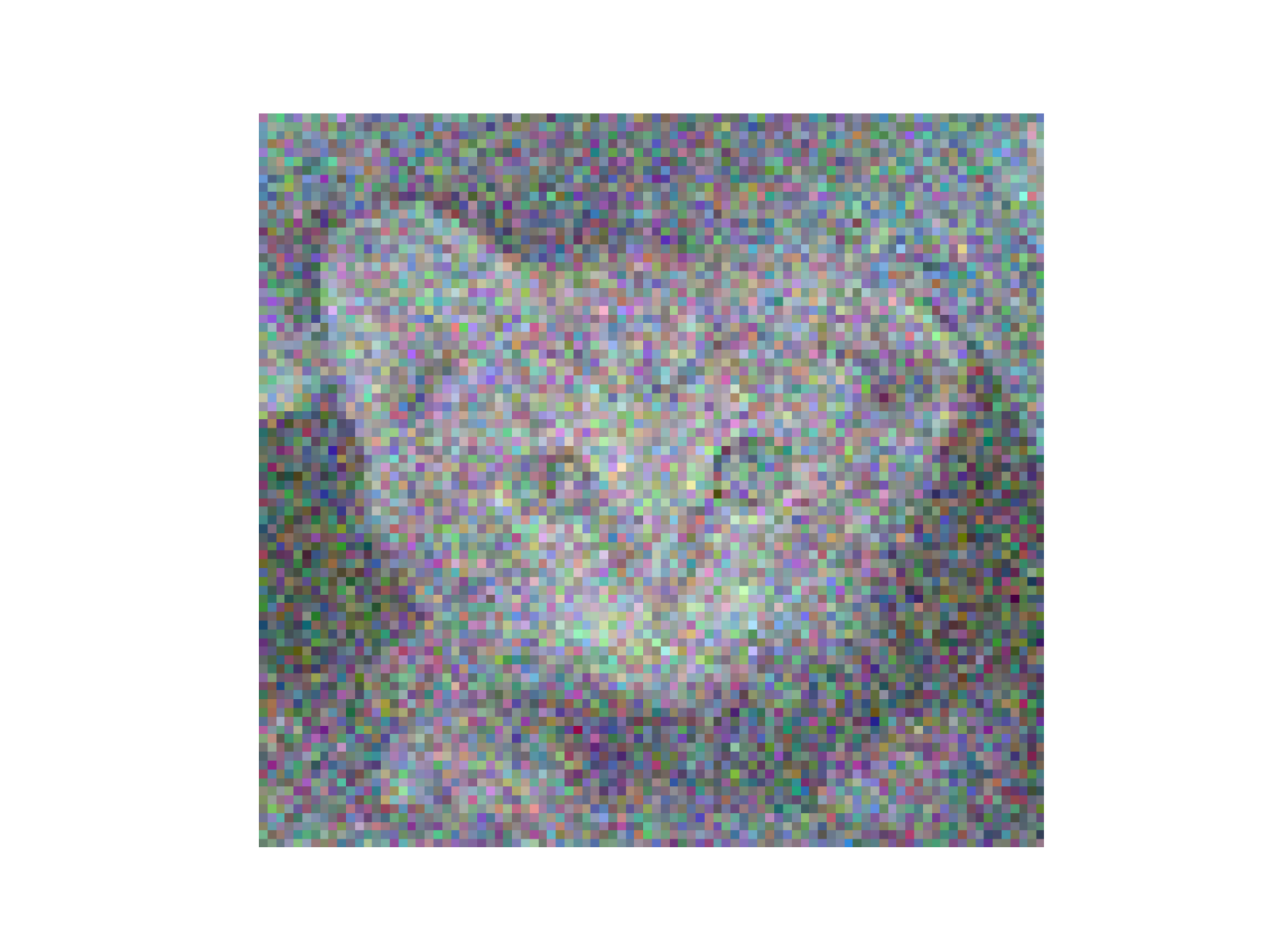}
\caption*{Gradient descent, $t = 4$.}
\end{subfigure}
\hspace{-3em}
\begin{subfigure}{0.3\columnwidth}
\centering
\includegraphics[width=\textwidth]{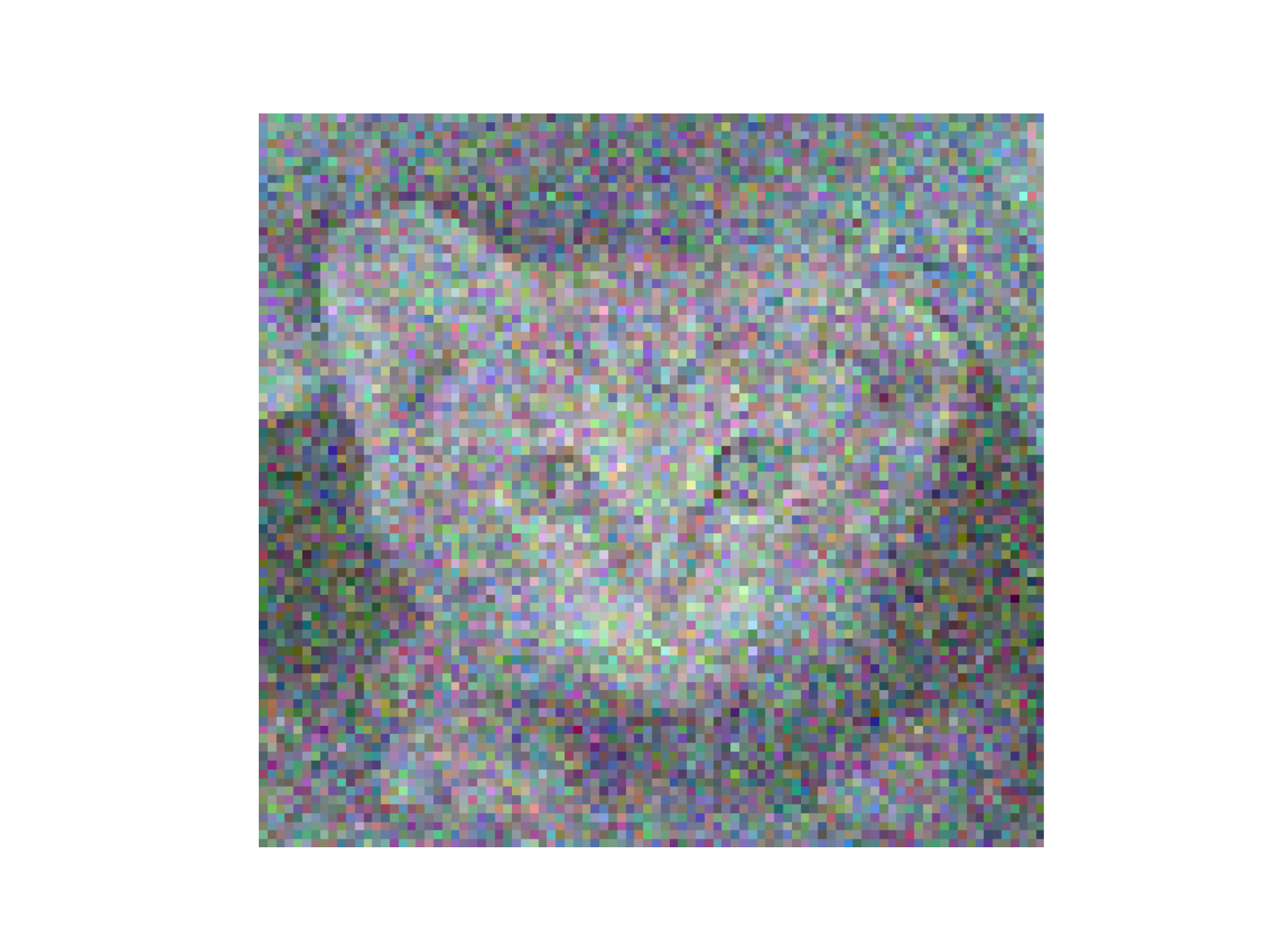}
\caption*{Gradient descent, $t = 8$.}
\end{subfigure}
\caption{Performance comparison between various GFOMs in noiseless phase retrieval (all algorithms use the same spectral initialization). }
\label{fig:cat}
\end{figure}

A few remarks are in order:
\begin{itemize}
\item While the theory developed below applies to $n,d\to\infty$, $n/d \to\delta$,
it appears to be fairly accurate already at moderate values of $n,d$. This is not
surprising given past results on AMP theory.
\item All GFOMs are substantially sub-optimal with the exception of Bayes AMP
that appears to achieve the upper bound correlation, as predicted by the theory. 
\item The prox-linear algorithm (black lines) appears to be nearly optimal for the largest sample size,
at $n/d = 2.5$. 

However, as emphasized above, prox-linear algorithm is not a GFOM.
In each round of iteration, we use \sffamily{cvxpy} \rmfamily 
in Python with the default solver to solve the optimization problem \eqref{eq:prox-linear}. 
In Table \ref{table:time}, we report the averaged wall clock time in seconds for the algorithms 
listed in \cref{sec:algs} with 10 iterations.
All experiments were conducted on a personal computer with 
8GB memory and 2 cores.
\end{itemize}

The step sizes for gradient descent and one-step prox-linear were chosen in 
Figure \ref{fig:cor} via trial and error as to optimize the performance of each algorithm. 
In Figure  \ref{fig:step} we plot accuracy as a function of step size parameter for each algorithm,
in the same setting as Figure \ref{fig:cor}. Our findings appear to be robust to the 
choice of this parameter.

In order to further illustrate the difference in performance and the optimality of Bayes AMP, 
we test the algorithms on a real image in Figure \ref{fig:cat}. 
The measurement matrix $\bX$ is random as above. The image contains $d=7560$ pixels
and we used $n=12000$ (hence $\delta=n/d\approx 1.6$),
and we treated each of the $3$ color channels separately.
The step sizes were chosen for 
gradient descent and one step prox-linear algorithm as to maximize reconstruction accuracy.

\section{Symmetric rank-one matrix estimation}
\label{sec:SymmetricMatr}

We observe a symmetric matrix $\bX \in \RR^{n \times n}$  given by
\begin{align}\label{model:spike}
 \bX = \frac{1}{n}\btheta\btheta^{\sT} +\bW\, ,
 \end{align}
 where  $\bW = \bW^{\sT}$ is a matrix with independent entries above the diagonal, 
 $(W_{ij})_{1\le i\le j\le n}$ such that $\E\{W_{ij}\} = 0$, 
 $\E\{W_{ij}^2\}=1/n$ for $1 \le i<j\le n$, and $\E\{W_{ii}^2\}=C/n$ for $1 \le i\le n$.
 In addition, we observe a vector $\bu\in\reals^n$ that could provide side 
 information about $\btheta$. The case in which this side information is not available is covered
 by setting $\bu=\mathbf 0$.
 Given $\mu_{\Theta, U}$, which is a fixed probability distribution over $\RR^2$ with finite second 
 moment, we assume $\{(\theta_i, u_i)\}_{i \leq n} \iidsim \mu_{\Theta, U}$. 
 Our objective is to estimate $\btheta$ from observations
 $(\bX, \bu)$. 

\subsection{General first order methods (GFOM)}

A GFOM is an iterative algorithm. At the $t$-th iteration performs the following update:
\begin{align}\label{eq:GFOM}
\begin{split}
	&\bu^{t + 1}	= \bX F_t(\bu^{\leq t}; \bu) + G_t(\bu^{\leq t}; \bu)\, ,\\
	&F_t(\bu^{\leq t}; \bu) := F_t(\bu^1, \cdots, \bu^t; \bu)\, ,\;\;\;
	G_t(\bu^{\leq t}; \bu) := G_t(\bu^1, \cdots, \bu^t; \bu) \, .
\end{split}
\end{align}
where $F_t, G_t: \RR^{n(t + 1)} \rightarrow \RR^n$ are functions indexed by 
$t \in \NN$. 
 After $s$ iterations, the algorithm estimates
 $\btheta$  by $\hbtheta^s = F_{\ast}^{(s)}(\bu^{\leq s}; \bu)$, where
  $F_{\ast}^{(s)}: \RR^{n(s + 1)} \rightarrow \RR^n$ is a continuous function. 
  Notice that a GFOM is uniquely determined by the choice of nonlinearities 
  $\{F_t, G_t, F_{\ast}^{(t)}\}_{t \in \NN}$. 
  
  We will consider two specific settings for the functions   
  $\{F_t, G_t, F_{\ast}^{(t)}\}_{t \in \NN}$, and the noise $\bW$.
  The choice of these settings is dictated by the cases in which an asymptotic characterization
  of the AMP algorithms, known as `state evolution'  
  \cite{bayati2011dynamics,javanmard2013state} has been established rigorously.
  Namely, for Setting \ref{setting:Gaussian} we will leverage the results of \cite{berthier2020state},
  while for Setting \ref{setting:Wigner} we will use the results of \cite{bayati2015universality,chen2021universality}.
\begin{setting}\label{setting:Gaussian}
\begin{itemize}
\item The matrix $\bW$ has entries $(W_{ij})_{i<j}\sim_{iid}\normal(0,1/n)$,
and $\E W^2_{ii}\le C/n$ for a constant $C$.
\item The probability measure $\mu_{\Theta,U}$ is sub-Gaussian.
\item The functions
$F_t, G_t, F_{\ast}^{(t)}:\reals^{n(t+1)}\to\reals^n$ are uniformly 
Lipschitz\footnote{We say that sequence of functions $\{f_n: \RR^{a_n} \rightarrow \RR^{b_n}\}_{n \geq 1}$ 
is \emph{uniformly Lipschitz} if there exists $n$-independent constant $L > 0$, such that for all
$n$ and all $\bx , \by \in \RR^{a_n}$, $\|f_n(\bx) - f_n(\by)\|_2 / \sqrt{b_n} 
\leq L\|\bx - \by\|_2 / \sqrt{a_n}$ and $\|f_n(\bzero)\|_2 / \sqrt{b_n} \leq L$.}.
Further, for any fixed $\bmu\in \reals^{\NN}$, $\bSigma\in\reals^{\NN\times \NN}$ positive 
semi-definite and  $(b_{ij})_{i,j \in \NN_{>0}}$,
letting $(\bg_t)_{t\in\NN_{>0}}$ be a sequence of centered Gaussian vectors with 
$\E[\bg_s(\bg_t)^{\sT}] = \Sigma_{s,t}\id_{n}$, the following limits exist and is finite
 for all $s\le t$:
\begin{align*}
& \plim_{n\to\infty} \frac{1}{n} \<F_s(\by^1,\dots ,\by^s;\bu),
F_t(\by^1,\dots ,\by^t;\bu)\>\, , 
\end{align*}
where $\plim$ denotes limit in probability and $\{\by^t\}_{t \geq 1}$ is defined recursively as follows:
\begin{align}\label{eq:yt}
\begin{split}
	& \by^1 = \mu_1 \btheta + \bg_1 + G_0(\bu), \\
	& \by^{t + 1} = \mu_{t + 1} \btheta + \bg_{t + 1} + G_t(\by^1, \cdots, \by^t; \bu) + \sum\limits_{s = 1}^t b_{ts} F_{s - 1}(\by^1, \cdots, \by^{s - 1}; \bu).
\end{split}
\end{align}
Since $F_s$ is uniformly Lipschitz and the input random vectors are all sub-Gaussian, one can verify that $\{\|F_s(\by^1,\dots ,\by^s;\bu)\|_2^2 / n: n \in \NN^+\}$ is uniformly integrable. As a consequence, $\E\langle F_s, F_t \rangle / n$ converges to the same limit. The analogous limits for $\<F_s,G_t\>/n$,  $\<G_s,G_t\>/n$,
$\<F^\ast_s,G_t\>/n$,  $\<F^{\ast}_s,F_t\>/n$  $\<F^{\ast}_s,F^{\ast}_t\>/n$, $\langle F_t, \btheta \rangle / n$, $\langle G_t, \btheta \rangle / n$, $\langle F_t^{\ast}, \btheta \rangle / n$
 are also assumed to exist. Similarly, the limits of their expectations also exist. 
 \end{itemize}
\end{setting}

\begin{setting}\label{setting:Wigner}
\begin{itemize}
\item The matrix $\bW$ has independent entries on and above the diagonal with 
$W_{ij}=\oW_{ij}/\sqrt{n}$ where $(\oW_{ij})_{i<j\le n}$ is a collection of 
i.i.d. random variables with distribution independent of $n$, such that
 $\E \oW_{ij} =0$, $\E \oW_{ij}^2=1$,
 and $\E \oW_{ij}^4<\infty$. Further, there exists an absolute constant $C > 0$, such that $\E\{W^4_{ii}\}\le C/n^2$ for all $i\le n$.
\item The probability measure $\mu_{\Theta,U}$ is sub-Gaussian.
\item Fixed ($n$-independent)  functions 
$F_t, G_t, F_{\ast}^{(t)}:\reals^{t+1}\to\reals$ are  given.
We overload this notation by letting $F_t(\bu^1,\dots,\bu^t;\bu)\in\reals^n$
be the vector with the $i$-th component 
$F_t(\bu^1,\dots,\bu^t;\bu)_i=F_t(u_i^1,\dots,u_i^t;u_i)$. Either of the following is
assumed:
\begin{enumerate}
\item[$(a)$] The functions $F_t, G_t, F^{\ast}_{t}$ are Lipschitz continuous.
\item[$(b)$] The functions $F_t, G_t, F^{\ast}_{t}$ are polynomials,  and in addition
the entries of $\bW$ are sub-Gaussian $\E\{\exp(\lambda W_{ij})\}\le \exp(C\lambda^2/n)$
for some $n$-independent constant $C$. 
\end{enumerate}
\end{itemize}
\end{setting}

\subsection{Main result for rank-one matrix estimation}
\label{sec:Optimality}

In this section we state our optimality result for the case of rank-one matrix estimation.
We refer to the appendices for similar statements in the case of generalized linear models.

Let $(\Theta, U) \sim \mu_{\Theta, U}$, $G \sim \normal(0,1)$, independent of each other. 
Define the minimum mean square error function 
$\mmse_{\Theta, U}:\reals_{\ge 0}\to \reals_{\ge 0}$ via
\begin{align*}
	\mmse_{\Theta, U}(\gamma) &:= \inf_{\htheta:\reals^2\to\reals}
	\E\big\{\big[\Theta-\htheta(\gamma \Theta+G ,U)\big]^2\big\}\\
	&= \E[\Theta^2] - \E[\E[\Theta\mid \gamma\Theta + G, U]^2]\, .
\end{align*}  
Define the sequence $(\gamma_t)_{t\in\NN}$ via the following \emph{state evolution} recursion:
		\begin{align}\label{eq:beta}
			\gamma_{t + 1}^2 = \E[\Theta^2]-\mmse_{\Theta, U}(\gamma_t) \, ,\;\;\;\;\;
			\gamma_{0}= 0\,.
		\end{align}
The following theorem establishes that no GFOM can achieve mean square
error below $\mmse_{\Theta}(\gamma_t)$ after $t$ iterations.
\begin{theorem}\label{thm:main}
		For $t \in \NN_{\ge 0}$, let $\hat{\btheta}^t \in \RR^n$ be the output of any 
		GFOM after $t$ iterations, under either of
		Setting \ref{setting:Gaussian} or Setting \ref{setting:Wigner}.
		Then the following holds
		\begin{align}\label{eq:lower-bound}
			\plim_{n\to \infty}\frac{1}{n}\|\hat{\btheta}^t - \btheta\|_2^2 \geq 
			\mmse_{\Theta, U}(\gamma_t) \, .
		\end{align}
		Further there exists a GFOM which satisfies the above bound with equality.
\end{theorem}
In this statement $\plim_{n\to \infty}$ denotes limit in probability.
 
In the next section we will prove \cref{eq:lower-bound}.
We refer to \cite{celentano2020estimation} for a proof of the 
fact this lower bound is achieved. 
The proof given there implies that the algorithm achieving the lower bound is essentially 
unique and coincides with Bayes AMP.
\begin{remark}
The sequence $(\gamma_t)_{t\ge 0}$ is easily seen to be non-degreasing in $t$,
whence the sequence of lower bounds $\mmse_{\Theta, U}(\gamma_t)$ is non-increasing
and converging to $\mmse_{\Theta, U}(\gamma_\infty)$. The latter quantity therefore
provides the optimal error achieved by first order methods with $O(1)$ matrix-vector multiplications.

In some cases, $\mmse_{\Theta, U}(\gamma_\infty)$ is conjectured to be the optimal error 
achieved by polynomial-time algorithms 
\cite{lelarge2019fundamental,montanari2021estimation}. More precisely, this is expected to be the case if
the noise $\bW$ is Gaussian and  $\E[\E[\Theta\mid  U]^2]>0$
(which is the case for instance if $\E[\Theta]\neq 0$). If these conditions are violated, 
better estimation can be achieved by the following approaches:
\begin{itemize}
\item If $\bW$ has i.i.d. but non-Gaussian entries, applying a nonlinear function entrywise
to $\bX$, and then using a spectral or first order method can improve estimation, 
see  \cite{montanari2018adapting} and references therein.
\item If $\E[\E[\Theta\mid  U]^2]=0$, then using a spectral initialization
improves estimation, see e.g. \cite{montanari2021estimation}.
\end{itemize}
Refined versions of the conjecture mentioned above can be formulated in these cases.
\end{remark}

\section{Proof of Theorem \ref{thm:main}}\label{sec:proof}

In this section we prove Theorem \ref{thm:main} under Setting \ref{setting:Wigner}.
Additionally, we will assume $\bW$ to have sub-Gaussian entries, namely
$\E\{\exp(\lambda W_{ij})\}\le \exp(C\lambda^2/n)$ for all $i,j\le n$ and
some $n$-independent constant $C$.
The proof under Setting \ref{setting:Gaussian} is given 
in Appendix \ref{sec:MatrixGaussian}, and the generalization to 
 Setting \ref{setting:Wigner} without sub-Gaussian assumption is carried out in 
 Appendix \ref{app:SubGaussian}.

Throughout the proof $(\Theta,U)\sim\mu_{\Theta,U}$ are random variables independent of
other random variables unless explicitly stated.

\subsection{Approximate message passing algorithms}

As mentioned above, an important role in the proof is played by approximate message
passing (AMP) algorithms. These are  GFOMs that enjoy special properties:
here we limit ourselves to giving a definition for the problem of symmetric rank-one 
matrix estimation, in the context of Setting \ref{setting:Wigner}.

 An AMP algorithm is defined by a sequence 
of continuous functions $\{f_t: \RR^{t + 1} \rightarrow \RR\}_{t \geq 0}$ 
(also termed the nonlinearities of the AMP algorithm), and produces a sequence of 
vectors $\{\ba^t\}_{t \geq 1} \subseteq \RR^n$ via the following iteration
\begin{align}
	\ba^{t + 1} & = \bX f_t(\ba^{\leq t}; \bu) - \sum_{s = 1}^t b_{t,s} f_{s - 1}(\ba^{\leq s - 1}; \bu)\, .\label{eq:AMP}
\end{align}
Here $\ba^{\leq t}=(\ba^1,\dots,\ba^t)$ and, as before, nonlinearities are applied entrywise.
The term subtracted on the right-hand side is known as Onsager correction term,
and we will introduce the notation
\begin{align}
 \Ons_{\mathrm{AMP}}^t(\ba^{\leq t - 1}; \bu):= \sum_{s = 1}^t b_{t,s} f_{s - 1}(\ba^{\leq s - 1}; \bu)
 \end{align}
 The coefficients $(b_{t,s})_{1\le s\le t}$ are deterministic. Before defining them, we
 introduce the following state evolution recursion to construct the sequences 
 $\bmu = (\mu_t)_{t\ge 1}$, $\bSigma = (\Sigma_{s,t})_{s,t\ge 1}$,
 where $\bSigma=\bSigma^{\sT}$:
\begin{align}
\begin{split}
\mu_{t + 1} & = \E\big\{\Theta\, f_t(\bmu_{\le t} \Theta + \bG_{\le t}; U) \big\}\,, \\
 \Sigma_{s+1,t+1} & = \E\big\{
 f_{ s}(\bmu_{\le s} \Theta + \bG_{\le s}; U)f_{t}(\bmu_{\le t} \Theta + \bG_{\le t}; U)\big\}\, ,\label{eq:SE}\\
	\bG _{\le t} &:=(G_1, \cdots, G_t) \sim \normal(\bzero, \bSigma_{\le t})\, .
\end{split}
\end{align}
In the above equations $\bSigma_{\le t}:= (\Sigma_{ij})_{i,j\le t}$ and 
$\bmu_{\leq t} := (\mu_i)_{i \leq t}$, and it is understood that 
$\bmu_{\le s} \Theta + \bG_{\le s} := (\mu_1 \Theta + G_1, \cdots, \mu_t \Theta + G_t)$. Note that $f_{0}$
only depends on $U$ and therefore the above recursion does not need any specific initialization.
In terms of the above, we define:
\begin{align}
 b_{t,s} = \E\big\{\partial_s f_t(\bmu_{\le t} \Theta + \bG_{\le t}; U)\big\}\,,
 \label{eq:Bdef}
\end{align}
where $\partial_s f_t$ denotes $s$-th entry of the weak derivative of $f$.

After $t$ iterations as in Eq.~\eqref{eq:AMP}, AMP estimates $\btheta$ by applying 
a function $F_t^{\ast}:\RR^{t+1}\to\RR$ entrywise:
\begin{align}
\hbtheta(\bX,\bu) := F_t^{\ast}(\ba^1,\dots,\ba^t;\bu) \, .
\end{align}
For $k,m \in \NN_{>0}$, we say a function $\phi: \RR^m \rightarrow \RR$ is \emph{pseudo-Lipschitz of order $k$} if there exists a constant $L > 0$, such that for all $\bx, \by \in \RR^m$, 
\begin{align*}
	|\phi(\bx) - \phi(\by)| \leq L(1 + \|\bx\|_2^{k - 1} + \|\by\|_2^{k - 1})\|\bx - \by\|_2.
\end{align*}
Notice that if $f_1,f_2: \RR^m \rightarrow \RR$ are pseudo-Lipschitz of order 
$k_1$ and $k_2$ respectively, then their product $f_1f_2$ is pseudo-Lipschitz of order $k_1 + k_2$.

 The following theorem characterizes the asymptotics of the AMP  iteration 
 \eqref{eq:AMP} for Wigner matrices. It was established in 
 \cite{bayati2011dynamics,javanmard2013state} for Gaussian matrices, 
 in \cite{bayati2015universality} for Wigner matrices with sub-Gaussian entries and 
 polynomials nonlinearities and in \cite{chen2021universality} for Wigner matrices
 with sub-Gaussian entries and Lipschitz nonlinearities. 
 (Some small adaptations are required in the last two cases to get the next statement in its 
 full generality. These are carried out in the appendix.)
\begin{theorem}\label{thm:AMP}
Assume the matrix $\bW$, and nonlinearities $f_t$ satisfy the same assumptions as $\bW$
and $F_t$ in Setting \ref{setting:Wigner}.
Then, for any $t \in \NN_{>0}$, and any $\psi: \RR^{t + 2} \rightarrow \RR$ be a pseudo-Lipschitz
	 function of order 2, the AMP algorithm \eqref{eq:AMP} satisfies
	\begin{align}
		\plim_{n\to\infty}\frac{1}{n}\sum_{i = 1}^n \psi(\ba_i^{\le t}, \theta_i, u_i) =
	\E\big\{\psi(\bmu_{\le t} \Theta + \bG_{\le t}, \Theta, U)\big\} \, ,
	\;\;\;\;\;\bG _{\le t}\sim \normal(\bzero, \bSigma_{\le t})
	\, .
	\end{align}
	(Here $\plim$ denotes limit in probability.)
\end{theorem}

\begin{remark}
	Theorem \ref{thm:AMP} under Setting \ref{setting:Wigner}.$(b)$ is 
	a modified version of  \cite[Theorem 4]{bayati2015universality}, but follows from the latter 
	through a standard argument. 
	More precisely:
	\begin{itemize}
	\item In  \cite[Theorem 4]{bayati2015universality}, 
	the nonlinearity $f_t$ depends only on $\ba^t$, while 
	here we allow it to depend on all previous iterates and the initialization $(\ba^{\leq t}, \bu)$.
	However \cite[Theorem 4]{bayati2015universality} covers the case in which iterates
	$\bx^{t}$ are matrices $\bx^t\in\reals^{n\times q}$. We can easily 
	reduce the treatment of nonlinearities that depend on all previous times to 
	this one~\cite{javanmard2013state,montanari2019optimization}. Fix a time horizon $t$ and
	choose $q> t$ (independent of $n$): by suitably choosing the nonlinearities in
	the algorithm that defines $\bx^t$, we can ensure that 
	$(\bx_s^{t})_{1 \leq s \leq t}$ coincides with 
	$(\ba^s)_{1 \leq s \leq t}$.
	\item  In  \cite[Theorem 4]{bayati2015universality},  the matrix $\bX$ has independent 
	centered entries (up to symmetries). The case of rank-one plus noise matrix $\bX$ 
	can be reduced to this one as in \cite{deshpande2014information,deshpande2017asymptotic,montanari2021estimation}. 
\end{itemize}
\end{remark}

\subsection{Any generalized first order method can be reduced to an AMP algorithm}

Following \cite{celentano2020estimation}, we first show that any GFOM  of the form
\eqref{eq:GFOM} can be reduced to an AMP algorithm by a change of variables.
\begin{lemma}\label{lemma:change-of-variable}
Assume the matrix $\bW$, the measure $\mu_{\Theta,U}$, and the 
nonlinearities $(F_s,G_s,F^{\ast}_s)_{s\ge 0}$ satisfy 
the assumptions of  Setting \ref{setting:Wigner}.
Then, there exist non-random functions
	 $\{\varphi_s: \RR^{s + 1} \rightarrow \RR^s\}_{s \geq 1}$ and 
	 $\{f_s: \RR^{s + 1} \rightarrow \RR\}_{s \geq 0}$, satisfying the same assumptions 
	 (and independent of $(\btheta,\bu, \bW)$)
	  such that  the following holds.
Letting $\{\ba^s\}_{s \geq 1}$ be the sequence of vectors produced by the AMP iteration \eqref{eq:AMP} 
	with non-linearities $\{f_s\}_{s \geq 0}$, we have, 
	  for any $t \in \NN_{>0}$, 
	\begin{align*}
		\bu^{\leq t} = \varphi_t(\ba^{\leq t}; \bu).
	\end{align*}
\end{lemma}

\begin{proof}
The proof is by induction over $t$. For the base case $t = 1$, we may simply take
 $f_0(u) = F_0(u)$ and $\varphi_1(\ba^1;\bu):=\ba^1 + G_0(\bu)$. 
	
Suppose the claim holds for the first $t$ iterations. We prove that it holds
 for iteration $t + 1$. By the induction hypothesis, 
	\begin{align*}
		\bu^{t + 1} = \bX F_t(\varphi_t(\ba^{\leq t}; \bu); \bu) + G_t(\varphi_t(\ba^{\leq t}; \bu); \bu).
	\end{align*} 
	Let $f_t(x^{\leq t}; u) = F_t(\varphi_t(x^{\leq t}; u); u)$. Since the composition 
	of Lipschitz functions is still Lipschitz, we may conclude that $f_t$ is a 
	Lipschitz function under Setting \ref{setting:Wigner}.$(a)$. 
	Analogously, it is a polynomial under Setting \ref{setting:Wigner}.$(b)$. 
	Based on the choice of $\{f_s\}_{0 \leq s \leq t}$, we compute the coefficients 
	for the Onsager correction term $\{b_{t,j}\}_{1 \leq j \leq t}$,  
	as per Eq.~\eqref{eq:Bdef}. We then define $\ba^{t+1}$ via Eq.~\eqref{eq:AMP},
	which yields
	\begin{align*}
		\ba^{t + 1} = \bu^{t + 1} - G_t(\varphi_t(\ba^{\leq t}; \bu); \bu) - \sum_{j = 1}^t b_{t,j} f_{j - 1}(\ba^{\leq j - 1}; \bu)\, .
	\end{align*}
	We can therefore define $\varphi_{t + 1}$ via
		\begin{align*}
		\varphi_{t + 1}(\ba^{\leq t + 1}; \bu) = (\varphi_t(\ba^{\leq t}; \bu); \ba^{t + 1} + G_t(\varphi_t(\ba^{\leq t}; \bu) + \sum_{j = 1}^t b_{t,j} f_{j - 1}(\ba^{\leq j - 1}; \bu)).
	\end{align*}
	(Here note that $\varphi_{t + 1}(\ba^{\leq t + 1}; \bu)  \in\reals^{n\times (t+1)}$,
	and $(\bA;\bB)$ denotes concatenation by columns.)
	
	 As above, we see immediately that $\varphi_{t+1}$ is Lipschitz under
	 Setting \ref{setting:Wigner}.$(a)$, and a polynomial under Setting \ref{setting:Wigner}.$(b)$.  
	 This completes the proof by induction.
\end{proof}

As an immediate consequence of the last lemma, AMP algorithms achieve the same error as GFOMs,
for the same number of iterations, under any loss.
(In this statement $\pliminf_{n\to \infty}$ denotes $\liminf$ in probability.
Namely, given a sequence of random variables $Z_n$, and $z\in\reals$, we
write $\pliminf_{n\to\infty}Z_n\ge z$ if, for any $\eps>0$, $\lim_{n\to\infty}
\P(Z_n\le z-\eps) = 0$.)
\begin{corollary}\label{coro:FirstCoro}
	Let $\mathcal{A}_{\mathrm{GFOM}}^t$ be the class of GFOM estimators 
	with $t$ iterations, and $\mathcal{A}_{\mathrm{AMP}}^t$ be the class of AMP algorithms 
	with $t$ iterations (under the assumptions of either Setting \ref{setting:Wigner}.$(a)$, or
	 Setting \ref{setting:Wigner}.$(b)$). 
	 (In particular $\hbtheta(\,\cdot\,)\in \mathcal{A}_{\mathrm{GFOM}}^t$ is defined by a set of
	 $n$-independent functions $\{F_t, G_t, F_{\ast}^{(t)}\}_{t \in \NN}$, and similarly for
	 $\hbtheta(\,\cdot\,)\in \mathcal{A}_{\mathrm{AMP}}^t$.)
	 
	 Then for any loss function $\calL: \RR^{n}\times\RR^{n}  \rightarrow \RR_{\ge 0}$:
	\begin{align}
		\inf_{\hbtheta(\,\cdot\, ) \in \mathcal{A}_{\mathrm{GFOM}}^t} 
		\pliminf_{n\to\infty}
		\calL(\hbtheta(\bX,\bu), \btheta)= 
		\inf_{\hbtheta(\,\cdot\,) \in \mathcal{A}_{\mathrm{AMP}}^t} 
			\pliminf_{n\to\infty}
		\calL(\hbtheta(\bX,\bu), \btheta)\, .\label{eq:FirstCoro}
	\end{align}
\end{corollary}
\begin{proof}
The left-hand side of Eq.~\eqref{eq:FirstCoro} is smaller or equal than the right-hand side
because  $\mathcal{A}_{\mathrm{AMP}}^t\subseteq  \mathcal{A}_{\mathrm{GFOM}}^t$.
To show that they are equal, let $\hbtheta(\,\cdot\, ) \in \mathcal{A}_{\mathrm{GFOM}}^t$
be any GFOM that achieves the infimum on the left with tolerance $\eps$. 
By Lemma \ref{lemma:change-of-variable} we can construct 
$\hbtheta'(\,\cdot\, ) \in \mathcal{A}_{\mathrm{AMP}}^t$ achieving the same loss.
\end{proof}
 
 \begin{remark}
 Note that throughout this section we are assuming $\{F_t, G_t, F_{\ast}^{(t)}\}_{t \in \NN}$
 to be $n$-independent. However, standard compactness arguments allows to
 extend the present treatment to $n$-dependent nonlinearities as long as the 
 constants implicit in the definitions of Setting \ref{setting:Wigner}
 (Lipschitz constant, maximum polynomial degree, and so on) are uniformly bounded.
 
 Appendix \ref{sec:MatrixGaussian} will treat the case of nonlinearities that are non-separable and 
 hence necessarily $n$-dependent.
 \end{remark}

\subsection{Any AMP algorithm can be reduced to an orthogonal AMP algorithm}

In the previous section we reduced GFOMs to AMP algorithms.
We next show that we can in fact limit ourselves to the analysis of a special subset of 
AMP algorithms, whose iterates are approximately orthogonal, after we subtract 
their components along $\btheta$.  We refer to this special subset as orthogonal AMP (OAMP) algorithms. 
\begin{lemma}\label{lemma:OAMP}
Let $\{\ba^t\}_{t \geq 1}$ be a sequence generated by the AMP iteration \eqref{eq:AMP},
under either of Setting \ref{setting:Wigner}.$(a)$ or  Setting \ref{setting:Wigner}.$(b)$. 
Then there exist functions $\{\phi_t: \RR^{t + 1} \rightarrow \RR^t\}_{t\ge 1}$, 
$\{g_t: \RR^{t + 1} \rightarrow \RR\}_{t \geq 0}$, satisfying the same assumptions
(and independent of $(\btheta, \bu,\bW)$)
such that the following holds.
Let  $\{\bv^t\}_{t \geq 1}$ be the sequence generated by an AMP algorithm with 
non-linearities $\{g_t\}_{t \geq 0}$ (and same matrix $\bX$ as for $\{\ba^t\}_{t \geq 1}$),
namely
\begin{align}
\bv^{t + 1} & = \bX g_t(\bv^{\leq t}; \bu) - \sum_{s = 1}^t b'_{t,s} g_{s - 1}(\bv^{\leq s - 1}; \bu) 
\, ,\label{eq:OAMP}
\end{align}
with deterministic coefficients $(b'_{t,s})$ determined by the analogous of Eq.~\eqref{eq:Bdef},
with $f_t$ replaced by $g_t$. Then we have:
\begin{enumerate}
\item[$(i)$] For all  $t\ge 1$,
	\begin{align*}
		\ba^{\leq t} = \phi_t(\bv^{\leq t}; \bu).
	\end{align*}
\item[$(ii)$] For any pseudo-Lipschitz function $\psi: \RR^{t + 2} \rightarrow \RR$ of order 2, 
	\begin{align}
		\plim_{n\to\infty}\frac{1}{n}\sum_{i = 1}^n \psi(\bv_i^{\le t}, \theta_i, u_i) = 
		\E\big\{\psi(V_1,\dots,V_t, \Theta, U) \big\},\label{eq:SE-OAMP}
	\end{align}
	where  $V_{i}:=x_{i-1}(\alpha_i\Theta+Z_i)$, with 
	$(x_{0},\dots,x_{t-1})\in\{0,1\}^t$,
	 $(\alpha_{1},\dots,\alpha_t)\in\reals^t$, 
	and  $\{Z_i\}_{i \in \NN_{\ge 1}} \iidsim \normal(0,1)$
	 standard random variables independent of $(\Theta,U)$.
	 \end{enumerate}
\end{lemma}
\begin{proof}
Throughout this proof, given a probability space $(\Omega,\cF,\P)$, we denote
by $L^2(\P) = L^2(\Omega,\cF,\P)$ the space of random variables with finite second moment. 
Given a closed linear subspace  $\cS\subseteq L^2(\P)$ and a random variable $T \in L^2(\PP)$, we denote by $\Pi_{\cS}(T)$
the projection of $T$ onto $\cS$ (i.e. the unique minimizer of $\|S-T\|_{L^2}^2=
\E\{(S-T)^2\}$ over $S\in\cS$). We denote by $\Pi^\perp_{\cS}=I-\Pi_{\cS}$ the projector 
onto its orthogonal complement.

Given $(\mu_t)_{t\ge 1}$, and $(\Sigma_{s,t})_{s,t\ge 1}$ defined via state evolution,
see Eq.~\eqref{eq:SE}, let $\bG$ be a centered Gaussian process with covariance $\bSigma$,
and define the random variables and subspaces
\begin{align*}
Y_t := f_t(\bmu_{\le t} \Theta + \bG_{\le t}; U) , \qquad \mathcal{S}_t := \mathrm{span}(Y_k: 0 \leq k \leq t)\, .
\end{align*}
Note that by state evolution $\<Y_t,Y_s\>_{L^2} = \Sigma_{t+1,s+1}$.

By linear algebra, there exist deterministic constants $\{c_{ts}\}_{0 \leq s \leq t}$, 
$x_{t} \in \{0,1\}$, such that 
$c_{tt} \neq 0$, and 
\begin{align*}
 R_t := c_{tt}\Pi^{\perp}_{\cS_{t - 1}}(Y_{t}) = \sum_{s = 0}^t c_{ts} Y_{s}, \qquad 
 \E[R_tR_s] = \mathbbm{1}_{s = t}x_{t},
\end{align*}
Indeed if $Y_t$ does not belong to $\cS_{t-1}$ we can simply take $x_t=1$
and $c_{tt} = \|\Pi^{\perp}_{\cS_{t - 1}}(Y_{t})\|_{L^2}^{-1}$. Otherwise we take $R_t=0$,
$c_{tt}=1$, $x_t=0$.

 We prove the lemma by induction. For the base case $t = 1$, we set $g_0(u) = c_{00}f_0(u)$ 
 whence the claim $(i)$ follows trivially. For claim $(ii)$ there are two cases.
 Either $\E\{f_0(U)^2\}=0$, whence $x_0=0$ and therefore $(ii)$ holds with $V_1=0$ almost surely,
 or $\E\{f_0(U)^2\}>0$ whence $x_0 = 1$, $c_{00}=\E\{f_0(U)^2\}^{-1/2}$,  
 and therefore the claim follows by state evolution, where
	\begin{align}\label{eq:alpha1}
		\alpha_1 = \frac{\E[\Theta f_0(U)]}{\E[f_0(U)^2]^{1/2}}.
	\end{align}
	Suppose the lemma holds for the first $t$ iterations. 
	We prove it also holds for the $(t + 1)$-th iteration. Define
	\begin{align}
		g_t(\bv^{\leq t}; u) = \sum_{s = 0}^t c_{ts} f_s(\phi_s(\bv^{\leq s}; u); u).
		\label{eq:Def-gt}
	\end{align}
	Then by the assumptions and the induction hypothesis,
	 $g_t$ is Lipschitz under Setting \ref{setting:Wigner}.$(a)$, and is a polynomial under
	 Setting \ref{setting:Wigner}.$(b)$. Given the nonlinearities 
	 $\{g_t\}_{s\le t}$, we can compute the 
	 coefficients $(b'_{s,j})_{1\le j\le s\le t}$. We denote 
	 the Onsager term for this new iteration by 
	 $ \Ons_{\mathrm{OAMP}}^t(\bv^{\leq t - 1}; \bu):=\sum_{j = 1}^t b_{t,j}' 
	 g_{j - 1}(\bv^{\leq j - 1}; \bu)$. With this notation, Eq.~\eqref{eq:OAMP} can be rewritten as:
	\begin{align*}
		\bv^{t + 1} =&  \sum_{s = 0}^tc_{ts} \bX f_s(\phi_s(\bv^{\leq s}; \bu); \bu) - \Ons_{\mathrm{OAMP}}^t(\bv^{\leq t - 1}; \bu) \, .
		\end{align*}
		Using the AMP iteration that defines $\{\ba^s\}_{s \geq 1}$, we get:
		\begin{align*}
		\bv^{t + 1} = \sum\limits_{s = 0}^tc_{ts}(\ba^{s + 1} + \Ons_{\mathrm{AMP}}^{s}( \ba^{\leq s - 1}; \bu)) - \Ons_{\mathrm{OAMP}}^t(\bv^{\leq t - 1}; \bu).
	\end{align*}
	Solving for $\ba^{t+1}$ and expressing $\ba^{\le t+1}=\phi_t(\bv^{\le t+1}; \bu)$
	(recall that $c_{tt}$ is always non-vanishing) we obtain the desired 
	mapping $\phi_{t+1}$ thus proving claim $(i)$.
	
	In order to prove claim $(ii)$, we distinguish two cases.
	In the first case $x_t = 0$ and $R_t \overset{a.s.}{=} 0$. Using the state evolution  
	for the orthogonal AMP iteration \eqref{eq:OAMP} and the definition 
	\eqref{eq:Def-gt}  we obtain that claim $(ii)$ folds with $V_{t + 1} \overset{a.s.}{=}  0$. 
	 
	 In the second case $ x_t = 1$, then again by state evolution 
	 we obtain that the claim holds with
	 $V_{t + 1} \ed \alpha_{t + 1} \Theta+ Z_{t + 1}$, where
	\begin{align}\label{eq:alphat+1}
		\alpha_{t + 1} = \frac{\E[\Theta\,\Pi^\perp_{\mathcal{S}_{t - 1}}(Y_t)]}
		{\E[\Pi^\perp_{\mathcal{S}_{t - 1}}(Y_t)^2]^{1/2}},
	\end{align}
  this  completes the proof.
	\end{proof}

	Considering the case in which $x_t\neq 0$ for all $t$ (i.e.,
	each new non-linearity is `non-degenerate'), Eq.~\eqref{eq:SE-OAMP} implies
	\begin{align}
	\bv^t = \alpha_t\btheta+\bz^t\, ,\;\;\;\;\;\; \frac{1}{n}\<\bz^t,\bz^s\> 
	= \mathbbm{1}_{s = t} +o_n(1)\,,\;\;\;\;\;\;
	\frac{1}{n}\<\bz^t,\btheta\> 
	= o_n(1)\, .
	\end{align}
    In other words, the iterates are approximately orthonormal along the subspace orthogonal to $\btheta$.
    This justifies the name `orthogonal AMP' (OAMP).
    
    \begin{remark}\label{rmk:xt}
    In the following we can and will  restrict ourselves
	to the case in which, in the notation  of Eq.~\eqref{eq:SE-OAMP}, $x_t=1$ for all $t$. 
	Indeed if $x_t = 0$ for some $t$, we can set to zero the corresponding AMP iterate $\bv_t=0$
	 (i.e. set $g_{t-1}=0$),
	and the resulting algorithm will asymptotically  have the same state evolution. 
	By removing this iteration altogether, we obtain an algorithm with same accuracy and one less iteration.
    \end{remark}
	%
	%
	\subsection{Optimal orthogonal AMP}
	By Lemma \ref{lemma:change-of-variable} and \ref{lemma:OAMP} 
	in order to derive a lower bound of estimation error achieved by GFOMs with $t$ iterations,
	it is sufficient to restrict ourselves to the class of orthogonal AMP
	algorithms (it is understood that the latter can be followed by entrywise post processing).
	
	We therefore have the following consequence of the previous results
	(see also Remark \ref{rmk:xt}).
	\begin{corollary}\label{coro:Final}
	Let $\hbtheta:(\bX,\bu)\mapsto \hbtheta(\bX,\bu)$ be a $t$-iterations GFOM estimator
	under the assumptions of either Setting \ref{setting:Wigner}.$(a)$, or
	 Setting \ref{setting:Wigner}.$(b)$. Then for any loss function $\ell: \RR\times\RR  \rightarrow \RR_{\ge 0}$,
	 pseudo-Lipschitz of order 2, we have
	\begin{align} 
		\plim_{n\to\infty}
	\frac{1}{n}\sum_{i=1}^n\ell(\hat\theta_i(\bX,\bu), \theta_i) 
	\ge 
		\inf_{(\{g_\ell\},\varphi) \in \mathcal{A}_{\mathrm{OAMP}}^t} 
			\E\big\{\ell(\varphi(\balpha_{\le t}\Theta+\bZ_{\le t},U), \Theta)\big\}\, .
			\label{eq:ReductionOAMP}
	\end{align}
	Here  the infimum is over all sequences of Lipschitz
	(Setting \ref{setting:Wigner}.$(a)$) or polynomial 
	 (Setting \ref{setting:Wigner}.$(b)$) nonlinearities for an orthogonal AMP algorithm, and over
	 all functions $\varphi:\reals^{t+1}\to\reals$ with the same properties.
\end{corollary}

	 Recall that a sufficient statistics for $\bTheta$ 
	 given   $\bV_{\le t} := \balpha_{\le t}\Theta+\bZ_{\le t}$ is 
	 $T_0:=\<\balpha_{\le t},\bV_{\le t}\>/\|\balpha_{\le t}\|_2$, and $T_0$ can be rewritten as: 
	 \begin{align}
	 T_0  = \|\balpha_{\le t}\|_2\Theta+ G\,, \;\;\;\;\;\; G\sim\normal(0,1)\,,\;\;\; G\perp \Theta\, .
	 \label{eq:T0Distr}
	 \end{align}
	 Since in addition $U$ is conditionally independent of $\bV_{\le t}$ given 
	 $\Theta$, the function $\varphi$ in Eq.~\eqref{eq:ReductionOAMP} can be taken to
	 be a function of $(U,T_0)$, and precisely the function that minimizes
	 the risk of estimating $\Theta$ with respect to the loss $\ell$.
	 The minimization on the right-hand side of Eq.~\eqref{eq:ReductionOAMP}
	 reduces to the maximization of $\|\balpha_{\le t}\|_2$,
	 which is solved by the next lemma.
	\begin{lemma}\label{lemma:Final}
	Recall the definition of  $(\gamma_s)_{s\ge 0}$ in Eq.~\eqref{eq:beta}. 
	Then, for all $t \in \NN_{>0}$, and all choices of nonlinearities 
	$g_0,\dots,g_t$, we have $\|\balpha_{\le t}\|_2\leq \gamma_t$.
	\end{lemma}
	\begin{proof}
	The proof is by induction over $t$. For the base case $t = 1$, using equation \eqref{eq:alpha1},
	 we have
	\begin{align*}
		\alpha_1^2 &\leq \sup_{f_0} \frac{\E[\Theta f_0(U)]^2}{\E[f_0(U)^2]}
		 = 
		\sup_{f_0} \frac{\E\{\E[\Theta|U] f_0(U)\big\}^2}{\E[f_0(U)^2]}
		 \le 
		\E\big\{\E[\Theta \mid U]^2\big\}\, .
	\end{align*}
	The last step holds by Cauchy-Schwarz inequality.
	
	We next assume that the claim holds for iteration $t$, and will prove it also holds for 
	iteration $t + 1$. Let $\hTheta_t:= \E[\Theta \mid U, V_1, \cdots, V_t]$.
	Using equation \eqref{eq:alphat+1}, we have
	\begin{align*}
		\alpha_{t + 1}^2 &=\frac{\E\big\{\hTheta_t\, \Pi^{\perp}_{\mathcal{S}_{t - 1}}(Y_t)\big\}^2}{\E[\Pi^{\perp}_{\mathcal{S}_{t - 1}}(Y_t)^2]} \\
		&\stackrel{(a)}{\le} \E\{\Pi^{\perp}_{\mathcal{S}_{t - 1}}(\hTheta_t)^2\}\\
		& \stackrel{(b)}{=} \E\{\hTheta_t^2\}-\E\{\Pi_{\mathcal{S}_{t - 1}}(\hTheta_t)^2\}\, ,
		\end{align*}
		where $(a)$ follows by Cauchy-Schwarz and $(b)$ by Pythagora's theorem.
	By construction $\{\Pi^{\perp}_{\mathcal{S}_{s - 1}}(Y_s)/\E[\Pi^{\perp}_{\mathcal{S}_{s - 1}}(Y_s)^2]^{1/2}:0\le s\le t-1\}$
		is an orthonormal basis for $\calS_{t-1}$, whence
		\begin{align*}
		\alpha_{t + 1}^2 &\le  \E[\hTheta_t^2] - \sum_{s = 0}^{t - 1}\frac{\E[\Theta \Pi^{\perp}_{\mathcal{S}_{s - 1}}(Y_s)]^2}{\E[\Pi^{\perp}_{\mathcal{S}_{s- 1}}(Y_s)^2]} \\
		& = \E[\hTheta_t^2] - \sum_{s = 1}^t\alpha_s^2\, ,
	\end{align*}
	Therefore $\|\balpha_{\le t+1}\|_2^2\le  \E[\hTheta_t^2]$.
	Further
	\begin{align*}
		 \E[\hTheta_t^2] &= \E[\E[\Theta \mid U, V_1, \cdots, V_t]^2] \\
		& \stackrel{(a)}{=} \E[\E[\Theta \mid U, \|\balpha_{\le t}\|_{2} \Theta + G]]\\
		& \stackrel{(b)}{\le } \E[\E[\Theta \mid U, \gamma_t \Theta + G]^2]\\
		& \stackrel{(c)}{=}  \gamma_{t + 1}^2,
	\end{align*}
	where $(a)$ follows because, as pointed above,
	  $T_0=\<\balpha_{\le t},\bV_{\le t}\>/\|\balpha_{\le t}\|_2$ is a 
	 sufficient statistics for $\bTheta$ 
	 given $\bV_{\le t}=\balpha_{\le t}\Theta+\bZ_{\le t}$,  and is distributed
	 as in Eq.~\eqref{eq:T0Distr}. Further, $(b)$ follows by Jensen's
	 inequality since, by the induction hypothesis, $\|\balpha_{\le t}\|_2 \leq \gamma_t$,
	 and $(c)$ by the definition of $\gamma_{t+1}$. This completes the induction. 
\end{proof}

The proof of Theorem \ref{thm:main} follows immediately from Corollary~\ref{coro:Final}
and Lemma~\ref{lemma:Final}.

\section{High-dimensional regression}
 \label{sec:MainRegression}

In this section, we generalize our results to regression in generalized linear models.
We observe a vector of responses $\by \in \RR^n$ and a matrix of covariates 
$\bX \in \RR^{n \times d}$ which are related according to 
\begin{align*}
	\by = h(\bX \btheta, \bw), 
\end{align*}
Here $\bw \in \RR^n$ is a noise vector, $\btheta \in \RR^d$  is a vector of parameters, and
 $h: \RR^2 \rightarrow \RR$ is a continuous function which we apply to vectors entrywise.
 Namely, denoting by $\bx_i\in\reals^d$ the $i$-th row of $\bX$, 
the above equation is equivalent to $y_i=h(\<\bx_i, \btheta\>, w_i)$ for $i\le n$.

We assume that $\bX \in \RR^{n \times d}$ has i.i.d. entries with $\E[X_{ij}] = 0$ and 
$\E[X_{ij}^2] = 1 / n$ for all $1 \leq i \leq n$ and $1 \leq j \leq d$. 
In addition, we observe side information 
$\bu \in \RR^n$ and $\bv \in \RR^d$. Given $\mu_{W,U}$ and $\mu_{\Theta, V}$  
two fixed probability distributions over $\RR^2$, we assume 
$\{(w_i, u_i)\}_{i \leq n} \iidsim \mu_{W, U}$ and $\{(\theta_i, v_i)\}_{i \leq d} \iidsim \mu_{\Theta, V}$. 
We consider the asymptotic setting where we have fixed asymptotic aspect ratio: 
$n / d \rightarrow \delta \in (0, \infty)$. The goal is to estimate $\btheta$ given 
$(\bX, \by, \bu, \bv)$.

\subsection{General first order methods}
In this section we introduce our notations for GFOMs for generalized linear models. At the $t$-th iteration, GFOM performs the following updates:
\begin{align}\label{eq:GFOM1}
	\begin{split}
		\bv^{t} := &  \bX^{\top} F_{t - 1}^{(1)}(\bu^{\leq t - 1}; \by, \bu) + F_{t - 1}^{(2)}( \bv^{\leq t - 1}; \bv), \\
	\bu^t := & \bX G_t^{(1)}(\bv^{\leq t}; \bv) + G_t^{(2)}(\bu^{ \leq t - 1}; \by, \bu),	
	\end{split}
\end{align}
where we use the shorthands $F_{s}^{(\ell)}(\bu^{\leq s}; \by, \bu) := F_{s}^{(\ell)}(\bu^1, \cdots, \bu^{s}; \by, \bu)$
and  $G_s^{(\ell)}(\bv^{\leq s}; \bv) := G_s^{(\ell)}(\bv^1, \cdots, \bv^{s}; \bv)$, 
where $F_t^{(1)}, G_{t + 1}^{(2)}: \RR^{n(t + 2)} \rightarrow \RR^n$, 
$F_t^{(2)}, G_t^{(1)}: \RR^{d(t + 1)} \rightarrow \RR^d$ are continuous functions with the 
$F$'s indexed by $t \in \NN$ and $G$'s indexed by $t \in \NN_{>0}$.  
After $s$ iterations, the algorithm estimates $\btheta$ by 
$\hat{\btheta}^{s} = G^{(s)}_{\ast}( \bv^{ \leq s}; \bv)$, where 
$G^{(s)}_{\ast}: \RR^{d(s + 1)} \rightarrow \RR^d$ is a continuous function. 
In this setting, a GFOM is uniquely determined by the set of nonlinearities
 $\{F_{t - 1}^{(1)}, F_{t - 1}^{(2)}, G_t^{(1)}, G_t^{(2)}, G^{(t)}_{\ast}\}_{t \in \NN_{>0}}$.

As in the case of low-rank matrix estimation, we consider two settings for the random matrix $\bX$,
and the nonlinearities 
 $\{F_{t - 1}^{(1)}, F_{t - 1}^{(2)}, G_t^{(1)}, G_t^{(2)}, G^{(t)}_{\ast}\}_{t \in \NN_{>0}}$. 
\begin{setting}\label{setting:3}
	\begin{itemize}
		\item The matrix $\bX$ has entries $X_{ij} \iidsim \normal(0, 1 / n)$.
		\item The probability measures $\mu_{\Theta, V}$, $\mu_{W, U}$ are sub-Gaussian. 
		\item The functions $F_{t}^{(1)}, F_{t}^{(2)}, G_t^{(1)}, G_t^{(2)}, G^{(t)}_{\ast}$ are uniformly Lipschitz. 
		Further, for any $\bmu \in \RR^{\NN}$, $\bSigma,\bar \bSigma \in \RR^{\NN \times \NN}$ 
		positive semi-definite and $(b_{ij})_{1 \leq i,j \leq t}, (\bar b_{ij})_{1 \leq  i,j \leq t}$ 
	   $n$-independent constants, we let $(\bg_t)_{t \in \NN_{>0}}$ and $(\bar\bg_t)_{t \in \NN}$ be 
	   centered Gaussian processes with $\E[\bg_s \bg_t^{\sT} ] = \Sigma_{st} \id_d$ 
	   and $\E[{\bar\bg_s} {\bar\bg_t}^{\sT}] = \bar\Sigma_{st}\id_n$, we assume the following limits exist for all $s \leq t$,
		\begin{align*}
			& \plim_{n,d \rightarrow \infty}\frac{1}{d} \langle F_t^{(2)} (\by^1, \cdots, \by^t; \bv), F_s^{(2)}(\by^1, \cdots, \by^s; \bv) \rangle, \\
			& \plim_{n,d \rightarrow \infty} \frac{1}{n}\langle F_t^{(1)}(\bar\by^1, \cdots, \bar\by^t; h(\bar\bg_0, \bw), \bu),  F_s^{(1)}(\bar\by^1, \cdots, \bar\by^s; h(\bar\bg_0, \bw), \bu) \rangle,
		\end{align*}
		where $\{\by^t\}_{t \geq 1}$, $\{\bar\by_t\}_{t \geq 1}$ are defined recursively as follows:
		\begin{align*}
			& \by^1 = \mu_1 \btheta + \bg_1 + F_0^{(2)} (\bv), \\
			& \by^{t + 1} = \mu_{t + 1} \btheta + \bg_{t + 1} + F_{t}^{(2)}(\by^{\leq t}; \bv) + \sum_{s = 1}^t b_{ts} G_{s}^{(1)}(\by^{\leq s}; \bv), \\
			& \bar{\by}^1 = \bar{\bg}_1 + G_1^{(2)}(h(\bar\bg_0, \bw), \bu) + \bar{b}_{11}F_0^{(1)} (h(\bar\bg_0, \bw), \bu), \\
			& \bar{\by}^{t + 1} = \bar\bg_{t + 1} + G_{t + 1}^{(2)}(\bar\by^1, \cdots, \bar\by^t; h(\bar\bg_0, \bw), \bu) + \sum_{s  =1}^{t + 1} \bar{b}_{t+1,s}F_{s - 1}^{(1)}(\bar\by^1,\cdots, \bar\by^{s - 1}; h(\bar\bg_0, \bw), \bu ).
		\end{align*}
		The analogous limits for $\langle G_t^{(1)}, G_s^{(1)} \rangle / d$, $\langle G_t^{(1)}, F_s^{(2)}\rangle / d$, $\langle G_{\ast}^{(t)}, G_s^{(1)} \rangle / d$, $\langle G_{\ast}^{(t)}, F_s^{(2)} \rangle / d$, $\langle G_{\ast}^{(t)}, G_{\ast}^{(s)}\rangle / d$, $\langle \btheta, G_t^{(1)} \rangle / d$, $\langle \btheta, F_t^{(2)} \rangle / d$,  $\langle \btheta, G_{\ast}^{(t)} \rangle / d$,  $\langle G_t^{(2)}, G_s^{(2)} \rangle / n$, $\langle  G_t^{(2)}, F_s^{(1)} \rangle / n$, $\langle F_t^{(1)}, \bar{\bg}_s\rangle / n$ and $\langle G_t^{(1)}, \bg_s \rangle / d$ are also assumed to exist. 
	\end{itemize}
\end{setting}

\begin{setting}\label{setting:4}
	\begin{itemize}
		\item  The matrix $\bX$ has independent entries with 
$X_{ij}=\oX_{ij}/\sqrt{n}$ where $(\oX_{ij})_{i\le n, j \le d}$ is a collection of 
i.i.d. random variables with distribution independent of $(n,d)$, such that
 $\E \oX_{ij} =0$, $\E \oX_{ij}^2=1$,
 and $\E \oX_{ij}^4<\infty$.
 	\item The probability measures $\mu_{\Theta, V}$, $\mu_{W, V}$ are sub-Gaussian.
		\item We have $n$-independent functions $F_{t - 1}^{(1)}, F_{t}^{(2)}, G_t^{(1)}, G_{t}^{(2)}, G^{(t)}_{\ast}: \RR^{t + 1} \rightarrow \RR$. We overload these notations by letting $F_t^{(1)}(\bu^1, \cdots, \bu^t; \by, \bu) \in \RR^n$ be the vector with the $i$-th component $F_t(\bu^1, \cdots, \bu^t; \by, \bu)_i = F_t(u_i^1, \cdots, u_i^t; y_i, u_i)$. Similar notations apply for $F_t^{(2)}, G_t^{(1)}, G_t^{(2)}$ and $G_{\ast}^{(t)}$. We assume either of the following conditions:
		\begin{enumerate}
			\item[(a)] The functions $F_{t - 1}^{(1)}, F_{t - 1}^{(2)}, G_t^{(1)}, G_t^{(2)}, G^{(t)}_{\ast}$ are Lipschitz continuous. 
			\item[(b)] The functions $F_{t - 1}^{(1)}, F_{t - 1}^{(2)}, G_t^{(1)}, G_t^{(2)}, G^{(t)}_{\ast}$ are polynomial, and in addition the entries of $\bX$ are sub-Gaussian $\E[\exp(\lambda X_{ij})] \leq \exp(C\lambda^2 / n)$ for some $n$-independent constant $C$.
		\end{enumerate}
	\end{itemize}
\end{setting}

\subsection{Main result for generalized linear models}

Unless explicitly stated, in the rest parts of the proof we let $(\Theta, V) \sim \mu_{\Theta, V}$, $(W, U) \sim \mu_{W, U}$ and $Z, Z_0, Z_1 \iidsim \normal(0,1)$ independent of each other. We define the minimum mean squared error function $\mmse_{\Theta, V}$: $\RR_{\geq 0} \rightarrow \RR_{\geq 0}$ via
\begin{align*}
	\mmse_{\Theta, V}(\alpha) := & \inf_{\hat{\Theta}: \RR^2 \rightarrow \RR^2}\E\big\{ [\Theta - \hat{\Theta}(\alpha\Theta + Z, V)]^2 \big\} \\
	=& \E[\Theta^2] - \E\big\{ \E[\Theta \mid \alpha \Theta + Z, V]^2 \big\}.
\end{align*} 
We let 
$\beta_0 := 0$, $\sigma_1 := \delta^{-1/2} \E[\Theta^2]^{1/2}$
 and $\tilde{\sigma}_1 := 0$. Then for $s \in \NN^+$, we define the following quantities recursively:
\begin{align}\label{eq:beta1}
\begin{split}
	& \beta_s^2 = \frac{1}{\sigma_s^2}\E[\E[Z_0 \mid h(\sigma_s Z_0 + \tilde{\sigma}_s Z_1, W ), U, Z_1]^2], \qquad \beta_s \geq 0,\\
	& \sigma_{s + 1}^2 = \frac{1}{\delta}\mmse_{\Theta,V}(\beta_{s}), \qquad \tilde{\sigma}_{s + 1}^2 = \frac{1}{\delta}(\E[\Theta^2] - \mmse_{\Theta,V}(\beta_{s})).
\end{split}
\end{align}
The following theorem establishes that no GFOM can achieve mean squared error below $\mmse_{\Theta, V}(\beta_{t})$ after $t$ iterations. 

\begin{theorem}\label{thm:GLM}
	For $t \in \NN_{>0}$, let $\hat\btheta^t \in \RR^d$ be the output of any GFOM after $t$ iterations, then under either Setting \ref{setting:3} or \ref{setting:4}, the following holds:
	\begin{align}
	\label{eq:LowerBoundGLM}
		\plim_{n,d \rightarrow \infty} \frac{1}{d} \|\hat\btheta^t - \btheta\|_2^2 \geq \mmse_{\Theta, V}(\beta_{t}). 
	\end{align}
	Further, there exists a GFOM which satisfies the above bound with equality. 
\end{theorem}
The proof of the lower bound \eqref{eq:LowerBoundGLM} is presented in Appendix \ref{sec:ProofGLM_1} under 
Setting \ref{setting:4} and in Appendix \ref{sec:ProofGLM_2} under 
Setting \ref{setting:3}.
We refer to \cite{celentano2020estimation} for a proof that there exists a GFOM
achieving the bound with equality.

\subsection*{Acknowledgements}

This work was supported by the NSF grant CCF-2006489 and the ONR grant N00014-18-1-2729. 

\bibliographystyle{alpha}

\newpage

\begin{appendices}

\section{Proof of Theorem \ref{thm:main} under Setting \ref{setting:Gaussian}}
\label{sec:MatrixGaussian}
In this section we prove Theorem \ref{thm:main} in the context of Setting \ref{setting:Gaussian}. 
Therefore,  $F_t, G_t, F_{\ast}^{(t)}$ are non-separable,
namely they do not necessarily act on vectors entrywise.

Before we proceed, we first generalize the definition of pseudo-Lipschitz 
functions given in the main text. For any $m,l,k \in \NN_{>0}$, a function 
$\phi: \RR^l \rightarrow \RR^m$ is called a pseudo-Lipschitz function of order $k$ if 
there exists a  constant $L > 0$, such that for any $\bx, \by \in \RR^l$, 
\begin{align}
	\frac{1}{\sqrt{m}}\|\phi(\bx) - \phi(\by)\|_2 \leq & L\left( 1 + \left( \frac{\|\bx\|_2}{\sqrt{l}} \right)^{k-1} +\left( \frac{\|\by\|_2}{\sqrt{l}} \right)^{k-1} \right) \frac{\|\bx - \by\|_2}{\sqrt{l}},\label{eq:PL1} \\
	\frac{1}{\sqrt{m}}\|\phi(\bx)\|_2 \leq & L\left(1 + \left( \frac{\|\bx\|_2}{\sqrt{l}}\right)^k \right).\label{eq:PL2} 
\end{align}
In what follows, we will often consider sequences of functions
$\phi_n: \RR^{l_n} \rightarrow \RR^{m_n}$ indexed by $n$ (even if we often do not write explicitly 
that we are considering a sequence). We say that such a sequence $\{\phi_n\}_{n \geq 1}$
 is \emph{uniformly pseudo-Lipschitz} of order $k$ if Eqs.~\eqref{eq:PL1},
 \eqref{eq:PL2} hold with $L$ a constant that is independent of $n$.   
 
\subsection{Approximate message passing algorithms}

As before, the first step is to define the AMP algorithm for this setting. An AMP algorithm is defined by Lipschitz non-linearities $\{f_t: \RR^{n(t + 1)} \rightarrow \RR^n\}_{t \geq 0}$, and produces vectors $\{ \ba^t\}_{t \geq 1} \subseteq \RR^n$ via the following iteration:
\begin{align}\label{eq:nonseparable-AMP}
	\ba^{t + 1} = \bX f_t(\ba^{\leq t}; \bu) - \sum_{s = 1}^t b_{t,s} f_{s - 1}(\ba^{\leq s - 1}; \bu).
\end{align}
For each $t \in \NN$, $f_t$ stands for a sequence of functions which are uniformly Lipschitz continuous. As before, we introduce the notation $\Ons_{\mathrm{AMP}}(\ba^{\leq t - 1}; \bu) := \sum_{s = 1}^t b_{t,s} f_{s - 1}(\ba^{\leq s - 1}; \bu)$. Under Setting \ref{setting:Gaussian}, the state evolution recursion to construct $\bmu = (\mu_t)_{t \geq 1}$ and $\bSigma = (\Sigma_{s,t})_{s,t \geq 1}$ is defined as follows:
\begin{align}\label{eq:nonseparable-SE}
\begin{split}
	\mu_{t + 1} = & \lim_{n \rightarrow \infty}\frac{1}{n}\E[\btheta^{\top} f_t(\bmu_{\leq t} \btheta + \bg_{\leq t}; \bu)]\, ,\\
	\Sigma_{s + 1, t + 1} = & \lim_{n \rightarrow \infty}\frac{1}{n}\E[f_s(\bmu_{\leq s} \btheta + \bg_{\leq s}; \bu)^{\top} f_t(\bmu_{\leq t} \btheta + \bg_{\leq t}; \bu)]\, , \\
	\bg_{\leq t} := & (\bg_1, \cdots, \bg_t) \sim \normal(\mathbf{0}, \bSigma_{\leq t} \otimes \id_n)\, , 	
\end{split} 
\end{align}
where we adopted the notation $\bmu_{\leq t} \btheta + \bg_{\leq t} := 
 (\mu_1 \btheta + \bg_1, \cdots, \mu_t \btheta + \bg_t)$ and  we assume the above
 limits exist. Given $\bmu$ and $\bSigma$, we define
\begin{align}\label{eq:nonseparable-bts}
	b_{t,s} = \frac{1}{n}\sum_{i = 1}^n\E[\partial_{i,s} f_{t,i}(\bmu_{\leq t } \btheta + \bg_{\leq t}; \bu)],
\end{align}
where $f_{t,i}$ is the $i$-th coordinate of $f_t$, and $\partial_{i,s}$ denotes the weak derivative with respect to the $s$-th variable of the $i$-th row of the input matrix. To give an example, for variables $\bx_1, \cdots, \bx_t \in \RR^n$ and a function $f(\bx_1, \cdots, \bx_t)$ mapping from $\RR^{nt}$ to $\RR$, we have $\partial_{i,s} f(\bx_1, \cdots, \bx_t) = \partial_{(\bx_s)_i} f(\bx_1, \cdots, \bx_t)$.  Notice that here $b_{t,s}$ depends on $n$. Since $f_t$ is uniformly Lipschitz in terms of $n$, for all $t,s \in \NN_{>0}$, $b_{t,s}$ is uniformly bounded as a sequence in $n$. 

After $t$ iterations as in Eq.~\eqref{eq:nonseparable-AMP}, the AMP algorithm estimates $\btheta$ by applying a uniformly Lipschitz function $f_t^{\ast}: \RR^{n(t + 1)} \rightarrow \RR^n$ to $(\ba^{\leq t}, \bu)$: 
\begin{align*}
	\hat{\btheta}(\bX, \bu) = f_t^{\ast}(\ba^{\leq t}; \bu). 
\end{align*}
The following theorem characterizes the asymptotic performance of the AMP algorithm \eqref{eq:nonseparable-AMP}.
\begin{theorem}\label{thm:nonseparatble-SE}
	Assume that $\{(\theta_i, u_i)\}_{i \leq n} \iidsim \mu_{\Theta, U}$, and $\bW$ satisfies the same assumption as $\bW$ under Setting \ref{setting:Gaussian}. For all $t \in \NN$, assume $f_t$ is uniformly Lipschitz. Furthermore, we assume the limits
	\begin{align*}
		& \lim_{n \rightarrow \infty}\frac{1}{n}\E[\btheta^{\top} f_t(\bmu_{\leq t} \btheta + \bg_{\leq t}; \bu)] , \\
		& \lim_{n \rightarrow \infty} \frac{1}{n}\E[f_s(\bmu_{\leq s} \btheta + \bg_{\leq s}; \bu)^{\top} f_t(\bmu_{\leq t} \btheta + \bg_{\leq t}; \bu)], \\
		& \lim_{n \rightarrow \infty}\frac{1}{n}\E[\btheta^{\top} f_t^{\ast}(\bmu_{\leq t} \btheta + \bg_{\leq t}; \bu)], \\
		& \lim_{n \rightarrow \infty}\frac{1}{n}\E[f_s^{\ast}(\bmu_{\leq s} \btheta + \bg_{\leq s}; \bu)^{\top} f_t^{\ast}(\bmu_{\leq t} \btheta + \bg_{\leq t}; \bu)]
	\end{align*}
	 exist for all $n$-independent $(\bmu, \bSigma)$ and $t,s \in \NN$. Then, for any $t \in \NN_{>0}$ and $\{\psi_n: \RR^{n(t + 1)} \rightarrow \RR\}_{n \geq 1}$  uniformly pseudo-Lipschitz of order $2$, 
	\begin{align*}
		\plim_{n\to\infty}\Big|\psi_n(\ba^{\leq t}; \bu) - \E[\psi_n(\bmu_{\leq t} \btheta + \bg_{\leq t}; \bu)]\Big|=0\, . 
	\end{align*}
\end{theorem}
\begin{remark}
	Theorem \ref{thm:nonseparatble-SE} is a generalized version of 
	\cite[Theorem 1]{berthier2020state}. In \cite{berthier2020state}
	the non-linearity $f_t$ only depends on $(\ba^t, \bu)$, while here we allow it to depend 
	on all previous iterates $(\ba^{\leq t}, \bu)$. 
	
	This generalization can be conducted through the following steps: 
	$(1)$~Replace the vectors $f_t(\ba^t; \bu), \ba^t \in \RR^n$ by matrices 
	$f_t(\ba^t; \bu), \ba^t \in \RR^{n \times q}$, and replace the coefficients for the 
	Onsager correction term $b_{t,t}$ by $q \times q$ matrices (see, e.g., \cite{javanmard2013state}).
	 Such generalization follows exactly by the same proof as in \cite{berthier2020state}.
	 $(2)$~Fix a time horizon $t$, and choose an $n$-independent $q$ such that $q \geq t$. 
	 With initialization $\bx_1^0 = \cdots = \bx_q^0 = \textbf{0}$, we set the 
	 non-linearity corresponding to the $(s + 1)$-th iteration as
	\begin{align*}
		(\bx_1^{s}, \cdots, \bx_q^{s}, \bu) \mapsto (f_0(\bu), \cdots, f_s(\bx_1^s, \cdots, \bx_s^s; \bu), \mathbf{0}, \cdots, \mathbf{0}) \in \RR^{n \times q}.
	\end{align*}
	In this way, the vectors $(\bx_s^{t})_{1 \leq s \leq t}$ coincides with 
	$(\ba^s)_{1 \leq s \leq t}$.
\end{remark}

\subsection{Any GFOM can be reduced to an AMP algorithm}
In this section we show that, under Setting \ref{setting:Gaussian}, any GFOM can be reduced 
to an AMP algorithm via a change of variables. 
\begin{lemma}\label{lemma:GFOM-AMP1}
	Under the assumptions of Setting \ref{setting:Gaussian}, for all $t \in \NN_{>0}$, 
	there exist uniformly Lipschitz functions $\varphi_t: \RR^{n(t + 1)} \rightarrow \RR^{nt}$ 
	and $f_{t - 1}: \RR^{nt} \rightarrow \RR^n$ that are independent of $(\btheta, \bu, \bW)$, 
	such that the following holds. Let $\{\ba^t\}_{t \geq 1}$ be the sequence of 
	vectors produced by the AMP iteration \eqref{eq:nonseparable-AMP} with non-linearities 
	$\{f_s\}_{s \geq 0}$, then for any $t \in \NN_{>0}$, we have
	\begin{align*}
		\bu^{\leq t} = \varphi_t(\ba^{\leq t}; \bu), \qquad f_{t - 1}(\ba^{\leq t - 1}; \bu) = F_{t - 1}(\varphi_t(\ba^{\leq t - 1}; \bu); \bu). 
	\end{align*}
	Furthermore, $\{\varphi_t\}_{t \geq 1}$ satisfies the following conditions. Let $(\bmu, \bSigma)$ be the state evolution of the AMP algorithm defined in \cref{eq:nonseparable-SE}. For any $t \in \NN_{>0}$, there exist uniformly bounded numbers $(b_{ij})_{1 \leq i,j \leq t}$ (which depend on $n$), such that for $\by_{\leq t}$ defined in Eq.~\eqref{eq:yt}, we have $\by_{\leq t} = \varphi_t(\bmu_{\leq t} \btheta + \bg_{\leq t}; \bu)$. 
\end{lemma}

\begin{proof}
	We prove the lemma by induction over $t$. For the base case $t = 1$, we may simply take $f_0(\bu) = F_0(\bu)$ and $\varphi_1(\ba^1; \bu) := \ba^1 + G_0(\bu)$. Then $\by^1 = \varphi_1(\mu_1 \btheta + \bg_1; \bu)$ by definition.  
	
	Suppose the claim holds for the first $t$ iterations, then we prove it holds for the $(t + 1)$-th iteration. By the induction hypothesis, 
	\begin{align*}
		\bu^{t + 1} = \bX F_t(\varphi_t(\ba^{\leq t}; \bu); \bu) + G_t(\varphi_t(\ba^{\leq t}; \bu); \bu). 
	\end{align*}
	Let $f_t(\bx^{\leq t}; \bu) = F_t(\varphi_t(\bx^{\leq t}; \bu); \bu)$. The composite of uniformly Lipschitz functions is still uniformly Lipschitz, thus, we conclude that $f_t$ is uniformly Lipschitz. Based on the choice of $\{f_s\}_{0 \leq s \leq t}$, we compute the coefficients for the Onsager correction term $\{b_{t,s}\}_{1 \leq s \leq t}$ according to Eq.~\eqref{eq:nonseparable-bts}. Then we define $\ba^{t + 1}$ via Eq.~\eqref{eq:nonseparable-AMP}, which gives
	\begin{align*}
		\ba^{t + 1} = \bu^{t + 1} - G_t((\varphi_t(\ba^{\leq t}; \bu); \bu) - \sum_{s = 1}^t b_{t,s} f_{s - 1}(\ba^{s - 1}; \bu). 
	\end{align*}
	Therefore, we define $\varphi_{t + 1}$ as
	\begin{align*}
		\varphi_{t + 1}(\ba^{\leq t + 1}; \bu) = (\varphi_t(\ba^{\leq t}; \bu); \ba^{t + 1} + G_t(\varphi_t(\ba^{\leq t}; \bu); \bu) + \sum_{s = 1}^t b_{t,s} f_{s - 1}(\ba^{\leq s - 1}; \bu)).
	\end{align*}
	By induction hypothesis and the fact that $b_{t,s}$ is uniformly bounded with respect to $n$ for all fixed $t,s \in \NN_{>0}$, we have that $\varphi_{t + 1}$ is uniformly Lipschitz. Furthermore, 
	\begin{align*}
		&\varphi_{t + 1}(\bmu_{\leq t + 1} \btheta + \bg_{\leq t + 1}; \bu) \\
		 =& (\varphi_t(\bmu_{\leq t} \btheta + \bg_{\leq t}; \bu), \mu_{t + 1} \btheta + \bg_{t + 1} + G_t(\varphi_t(\bmu_{\leq t} \btheta + \bg_{\leq t}; \bu); \bu) + \sum_{s = 1}^t b_{t,s} f_{s - 1}(\bmu_{\leq s - 1} \btheta + \bg_{\leq s - 1}; \bu)) \\
		 =& (\by^{\leq t}, \by^{t + 1}),
	\end{align*}
	thus completes the proof of the lemma by induction. 
\end{proof}

The next lemma enables us to check the conditions of Theorem \ref{thm:nonseparatble-SE}.
\begin{lemma}\label{lemma:Subseq}
	Under the assumptions of Setting \ref{setting:Gaussian}, let $\{f_{t - 1}, \varphi_t\}_{t \in \NN^+}$ be the functions 
	defined in Lemma \ref{lemma:GFOM-AMP1}. For any $\bmu = (\mu_i)_{i\ge 1}$, 
	$\bSigma = (\Sigma_{ij})_{i,j\ge 1}\succeq \bzero$, let $(\bg_t)_{t>0}$
be a centered Gaussian process 	with covariance $\E\{\bg_s\bg_t^{\sT}\} = \Sigma_{s,t}\id_n$.
Then, for any $t \in \NN$ and any infinite subsequence $\cS\subseteq \NN_{>0}$ there exists a further
subsequence $\cS'\subseteq \cS$ along which the following limits exist for all $0 \leq s \leq r \leq t$:
	\begin{align}\label{eq:ExistLim1}
	\begin{split}
		& \lim_{n \rightarrow \infty; n\in S'}\frac{1}{n}\E[f_r(\bmu_{\leq r} \btheta + \bg_{\leq r}; \bu)^{\top}f_s(\bmu_{\leq s} \btheta + \bg_{\leq s}; \bu)], \\
		&  \lim_{n \rightarrow \infty; n\in S'} \frac{1}{n}\E[\btheta^{\top}f_s(\bmu_{\leq s} \btheta + \bg_{\leq s}; \bu)], \\
& \lim_{n \rightarrow \infty; n\in S'}\frac{1}{n}\E[F_{\ast}^{(r)}(\varphi_r(\bmu_{\leq r} \btheta + \bg_{\leq r};\bu); \bu)^{\top}F_{\ast}^{(s)}(\varphi_s(\bmu_{\leq s} \btheta + \bg_{\leq s}; \bu); \bu)], \\
& \lim_{n \rightarrow \infty; n\in S'} \frac{1}{n}\E[\btheta^{\top}F_{\ast}^{(s)}(\varphi_s(\bmu_{\leq s} \btheta + \bg_{\leq s};\bu); \bu)].
	\end{split}
	\end{align}
\end{lemma}
\begin{proof}
We can assume that the subsequence $\cS$ does coincide with the whole sequence, i.e. 
$\cS= \NN_{>0}$, as the general case follows by a simple change of notations.

Fix $t \in \NN$. Since $(b_{i,j})_{1 \leq i,j \leq t}$ are uniformly bounded, 
	there exists a subsequence $\{n_k\}_{k >0}$ of $\NN_{>0}$, such that for all 
	$1 \leq s,r \leq t$, $b_{s,r}$ converges to limit $b_{s,r}^{\ast}$. Suppose we replace 
	$(b_{i,j})_{1 \leq i,j \leq t}$ with $(b_{i,j}^{\ast})_{1 \leq i,j \leq t}$ in Eq.~\eqref{eq:yt}, 
	and we denote the resulting vectors by $(\by_t^{\ast})_{t \geq 1}$. It follows by 
	induction and using the uniform Lipschitz property that for all $0 \leq s,r \leq t$, along $\{n_k\}_{k > 0}$,
	\begin{align*}
		& \frac{1}{n}F_r(\by_{\leq r}^{\ast}; \bu)^{\top} F_s(\by_{\leq s}^{\ast}; \bu) - \frac{1}{n}F_r(\by_{\leq r}; \bu)^{\top} F_s(\by_{\leq s}; \bu) \overset{P}{\rightarrow} 0, \\
		& \frac{1}{n}F^{(r)}_{\ast}(\by_{\leq r}^{\ast}; \bu)^{\top} F^{(s)}_{\ast}(\by_{\leq s}^{\ast}; \bu) - \frac{1}{n}F^{(r)}_{\ast}(\by_{\leq r}; \bu)^{\top} F^{(s)}_{\ast}(\by_{\leq s}; \bu) \overset{P}{\rightarrow} 0, \\
		& \frac{1}{n}\btheta^{\top} F_s(\by_{\leq s}^{\ast}; \bu) - \frac{1}{n}\btheta^{\top} F_s(\by_{\leq s}; \bu) \overset{P}{\rightarrow} 0. \\
		& \frac{1}{n}\btheta^{\top} F^{(s)}_{\ast}(\by_{\leq s}^{\ast}; \bu) - \frac{1}{n}\btheta^{\top} F^{(s)}_{\ast}(\by_{\leq s}; \bu) \overset{P}{\rightarrow} 0.
	\end{align*}
	By the third assumption of Setting \ref{setting:Gaussian}, the limits of 
	$F_r(\by_{\leq r}^{\ast}; \bu)^{\top} F_s(\by_{\leq s}^{\ast}; \bu) / n$, 
	$F^{(r)}_{\ast}(\by_{\leq r}^{\ast}; \bu)^{\top} F^{(s)}_{\ast}(\by_{\leq s}^{\ast}; \bu) / n$, 
	$\btheta^{\top} F^{(s)}_{\ast}(\by_{\leq s}^{\ast}; \bu)/ n$  and $\btheta^{\top} F_s(\by_{\leq s}^{\ast}; \bu) / n$ 
	exist in probability as $n,d \rightarrow \infty$. Combining these results
	 and the results of Lemma \ref{lemma:GFOM-AMP1}, we conclude that the  limits 
	 of Eqs.~\eqref{eq:ExistLim1} exist along $\{n_k\}_{k \in \NN_{>0}}$:
\end{proof}

The following corollary is an immediate consequence of Lemma \ref{lemma:GFOM-AMP1}. 
\begin{corollary}\label{coro:Reduction-AMP-NonSep}
	Under the assumptions of Setting \ref{setting:Gaussian}, let $\mathcal{A}_{\mathrm{GFOM}}^t(L)$
	 be the class of GFOM estimators with $t$ iterations and uniform Lipschitz constant $L$, and 
	 $\mathcal{A}_{\mathrm{AMP}}^t(L')$ be the class of AMP algorithms with $t$ iterations
	 and uniform Lipschitz constant $L'$. Then for  any $L<\infty$ there exist $L'<\infty$
	 (independent of $n$), such that the following holds.
	 For any $z\in\RR$ and any loss function $\calL: \RR^n \times \RR^n \rightarrow \RR_{\geq 0}$: 
	\begin{align*}
		\inf_{\hat\btheta(\cdot) \in \mathcal{A}_{\mathrm{GFOM}}^t(L)} 
		\P\Big(\calL(\hat{\btheta}(\bX, \bu), \btheta)\le z\Big)
		\le  \inf_{\hat\btheta(\cdot) \in \mathcal{A}_{\mathrm{AMP}}^t(L')} 
		\P\Big(\calL(\hat{\btheta}(\bX, \bu), \btheta)\le z\Big)\, .
	\end{align*}
\end{corollary}

Notice that in this corollary $\hat\btheta(\, \cdot\, )\in \mathcal{A}_{\mathrm{GFOM}}^t(L)$ is (implicitly) a sequence of 
estimators indexed by $n$, which is uniformly Lipschitz with constant $L$.
The corollary also implies an asymptotic statement. Namely,
write $\mathcal{A}_{\mathrm{GFOM}}^t:=\cup_{L\ge 1}\mathcal{A}_{\mathrm{GFOM}}^t(L)$ for 
the class of (sequences of) GFOM estimators with $t$ 
iterations and any uniform Lipschitz constant $L$, and 
similarly for $\mathcal{A}_{\mathrm{AMP}}^t$.
Then we have 
\begin{align}
		\inf_{\hat\btheta(\cdot) \in \mathcal{A}_{\mathrm{GFOM}}^t} 
		\pliminf_{n\to\infty}\mathcal{L}(\hat{\btheta}(\bX, \bu), \btheta)
		=  \inf_{\hat\btheta(\cdot) \in \mathcal{A}_{\mathrm{AMP}}^t} 
		\pliminf_{n\to\infty}\mathcal{L}(\hat{\btheta}(\bX, \bu), \btheta)\, .
		\label{eq:Reduction-AMP-NonSep}
	\end{align}
Here equality holds because 
$\mathcal{A}_{\mathrm{AMP}}^t\subseteq \mathcal{A}_{\mathrm{GFOM}}^t$.

\subsection{Any AMP algorithm can be reduced to an orthogonal AMP algorithm}

By Corollary \ref{coro:Reduction-AMP-NonSep}, and in particular Eq.~\eqref{eq:Reduction-AMP-NonSep},
we can limit ourselves to lower-bounding the error of AMP algorithms. 
By Lemma \ref{lemma:Subseq} we can assume ---possibly taking subsequences---
that such algorithm satisfies the conditions of Theorem \ref{thm:nonseparatble-SE}.
To simplify notations, we will assume hereafter that these conditions
are satisfied along $n\in \NN$. There is no loss of generality in this.

Here we show that it is in fact sufficient to lower bound the error for OAMP
algorithms.
\begin{lemma}\label{lemma:AMP-OMAP1}
	Let $\{\ba^t\}_{t \geq 1}$ be a sequence generated by the AMP iteration
	 \eqref{eq:nonseparable-AMP} under the conditions of Theorem \ref{thm:nonseparatble-SE}. 
	 Then for all $t \in \NN^+$, there exist uniformly Lipschitz functions 
	 $\phi_t: \RR^{n(t + 1)} \rightarrow \RR^{nt}$, $g_{t - 1}: \RR^{nt} \rightarrow \RR^n$ 
	 such that the following holds. Let $\{\bv^t\}_{t \geq 1}$ be the sequence of
	  vectors produced by AMP iteration with non-linearities $\{g_t\}_{t \geq 0}$
	   (and the same matrix $\bX$ as for $\{\ba^t\}_{t \geq 1}$). Namely, 
	\begin{align}\label{eq:OAMP1}
		\bv^{t + 1} = \bX g_t(\bv^{\leq t}; \bu) - \sum_{s = 1}^t b_{t,s}' g_{s - 1}(\bv^{\leq s - 1}; \bu)
	\end{align}
	with deterministic coefficients $(b_{t,s}')$ determinied by the analogous of 
	Eq.~\eqref{eq:nonseparable-bts}, with $f_t$ replaced by $g_t$. Then we have
	\begin{enumerate}
		\item[(i)] For all $t \in \NN_{>0}$, $\ba^{\leq t} = \phi_t(\bv^{\leq t}; \bu)$. Further, there exists $n$-independent constants $\{c_{ts}\}_{0 \leq s \leq t}$, such that we can write $\bv^t = \sum_{s = 0}^{t - 1} c_{t - 1, s} \ba^{s + 1}$.  
		\item[(ii)] For all $t \in \NN_{>0}$, there exist $(x_0, \cdots, x_{t - 1}) \in \{0,1\}^t$
		 and  $(\alpha_1, \cdots, \alpha_t) \in \RR^t$, such that for any
		  $\{\psi_n: \RR^{n(t + 2)} \rightarrow \RR\}_{n \geq 1}$ 
		  uniformly pseudo-Lipschitz of order 2, 
		\begin{align*}
			\psi_n(\bv^{\leq t}, \btheta, \bu) = \E[\psi_n(\bnu^{\leq t}, \btheta, \bu)] + o_P(1),
		\end{align*}
		where $\bnu^i = x_{i - 1}(\alpha_i \btheta + \bz_i)$ and  $\{\bz_i\}_{i \geq 1} \iidsim \normal(\mathbf{0}, \id_n)$ independent of $(\btheta, \bu)$. 
	\end{enumerate}
\end{lemma}

\begin{proof}
	Recall that, as in the proof of Lemma \ref{lemma:OAMP}, $\Pi_{\mathcal{S}}$ denotes the 
	orthogonal projection onto the closed linear subspace $\mathcal{S} \subseteq L^2(\PP)$,
	and $\Pi_{\mathcal{S}}^{\perp} := I - \Pi_{\mathcal{S}}$.
		
	We denote by $(\mu_t)_{t \geq 1}$, $(\Sigma_{s,t})_{s,t \geq 1}$ the state evolution
	sequence corresponding to $\{\ba^t\}_{t \geq 1}$, defined via Eq.~\eqref{eq:nonseparable-SE}. Let $(\bg_t)_{t \geq 1}$ be a centered Gaussian process in $\RR^n$ such that $\Cov(\bg_s, \bg_t) = \Sigma_{s,t} \id_n$. We define the following random vectors and subspaces:
	\begin{align*}
		\bh_t = f_t(\bmu_{\leq t} \btheta + \bg_{\leq t}; \bu), \qquad \mathcal{S}_t = \mbox{span}(\bh_k: 0 \leq k \leq t). 
	\end{align*}
	By assumption, for all $s,t \in \NN$, 
	\begin{align}\label{eq:24}
		\frac{1}{n}\E\langle \bh_s, \bh_t \rangle \rightarrow \Sigma_{s+1,t+1}, \qquad \frac{1}{n} \E\langle \btheta, \bh_t \rangle \rightarrow \mu_{t + 1}. 
	\end{align}
	By linear algebra, there exist deterministic $n$-independent constants $\{c_{ts}\}_{t,s \in \NN}$, $\{x_t\}_{t \in \NN} \in \{0,1\}^{\NN}$, such that $c_{tt} \neq 0$ and
	\begin{align*}
		\sum_{i = 0}^t\sum_{j = 0}^s c_{ti} c_{sj} \Sigma_{i+1,j+1} = \mathbbm{1}_{s = t} x_t.
	\end{align*}
	If we let $\br_t = \sum_{s = 0}^t c_{ts} \bh_s$, then by the convergence of second moments given in Eq.~\eqref{eq:24}, for all $s,t \in \NN$
\begin{align*}
	\lim_{n \rightarrow \infty}\frac{1}{n} \E\langle \br_t, \br_s \rangle = \mathbbm{1}_{s = t} x_t. 
\end{align*}
Then we prove the lemma by induction. For the base case $t = 1$, we let $g_0(\bu) = c_{00} f_0(\bu)$, thus $\bv^1  = c_{00} \ba^1$ and claim $(i)$ follows trivially. As for claim $(ii)$, first notice that the limits exist for both $\E\langle g_0(\bu), g_0(\bu) \rangle  / n$ and $\E\langle g_0(\bu), \btheta \rangle  / n$ by the assumption on the original AMP iteration. Then we consider two cases. In the first case $x_0 = 0$, thus $\Sigma_{11} = 0$, 
$\mu_1^2 \leq c_{00}^{-2}\E[\|\btheta\|_2^2 / n]\E[\|g_0(\bu)\|_2^2 / n] \rightarrow 0$, and $(ii)$ holds with $\bnu^1 = \bzero$ by Theorem \ref{thm:nonseparatble-SE}. In the second case $x_0 = 1$, whence $c_{00} = \Sigma_{11}^{-1/2}$, and claim $(ii)$ again follows from state evolution. Furthermore, 
\begin{align}\label{eq:25}
	\alpha_1 = \lim_{n \rightarrow \infty} \frac{\E [\langle f_0(\bu), \btheta \rangle]}{\sqrt{n}\E[\langle f_0(\bu), f_0(\bu) \rangle]^{1/2}}. 
\end{align}
Suppose the lemma holds for the first $t$ iterations. We prove it also holds for the $(t + 1)$-th iteration. We let
\begin{align*}
	g_t(\bv^{\leq t}; \bu) = \sum_{s = 0}^t c_{ts} f_s(\phi_s(\bv^{\leq s}; \bu); \bu).
\end{align*}
By induction hypothesis and assumptions, $g_t$ is uniformly Lipschitz. Given $\{g_s\}_{0 \leq s \leq t}$, we can derive the coefficients $(b_{s,j}')_{1 \leq j \leq s \leq t}$ via Eq.~\eqref{eq:nonseparable-bts}, and we denote the Onsager correction term of this new AMP iteration by $\Ons_{\mathrm{OAMP}}^t(\bv^{\leq t - 1}; \bu) = \sum_{s = 1}^t b_{t,s}' g_{s - 1}(\bv^{\leq s - 1}; \bu)$. Then Eq.~\eqref{eq:OAMP1} can be rewritten as 
\begin{align*}
	\bv^{t + 1} = \sum_{s = 0}^t c_{ts} \bX f_s(\phi_s(\bv^{\leq s}; \bu); \bu) - \Ons_{\mathrm{OAMP}}^t(\bv^{\leq t - 1}; \bu).
\end{align*}
Plugging in the AMP iteration that defines $\{\ba^t\}_{t \geq 1}$, we have
\begin{align}\label{eq:27}
	\bv^{t + 1} = \sum_{s = 0}^t c_{ts}(\ba^{s + 1} + \Ons_{\mathrm{AMP}}^s(\ba^{\leq s - 1}; \bu)) - \Ons_{\mathrm{OAMP}}^t(\bv^{\leq t - 1}; \bu).
\end{align}
Recall that $c_{tt}$ is non-vanishing, thus, we can solve for $\ba^{t + 1}$ and express $\ba^{\leq t + 1}$ as a function of $(\bv^{\leq t + 1}; \bu)$. We denote this function by $\phi_{t + 1}$. By induction hypothesis, $\phi_{t + 1}$ is uniformly Lipschitz. Plugging the definition of $\Ons_{\mathrm{AMP}}^s$ and $\Ons_{\mathrm{OAMP}}^t$ into Eq.~\eqref{eq:27} gives 
\begin{align}\label{eq:28}
	\bv^{t + 1} = \sum_{s = 0}^t c_{ts} \ba^{s + 1} + \sum_{i = 1}^t \big( \sum_{s = i}^t c_{ts} b_{si} - \sum_{s = i}^t b_{ts}' c_{s - 1, i - 1} \big) f_{i - 1}(\ba^{\leq i - 1}; \bu).
\end{align} 
By induction hypothesis, $g_t(c_{00} \bx^1, \cdots, \sum_{s = 0}^{t - 1} c_{t-1,s} \bx^{s + 1}; \bu) = \sum_{s = 0}^t c_{ts} f_s(\bx^{\leq s}; \bu)$. Taking the gradient on both sides with respect to $\bx^i$, then compute the expected average of the coordinates of the gradient with respect to the distribution $\bx^{\leq t} \overset{d}{=} \bmu^{\leq t} \btheta + \bg^{\leq t}$  gives $\sum_{s = i}^t c_{ts} b_{si} - \sum_{s = i}^t b_{ts}' c_{s - 1, i - 1} = 0$. Plugging this into Eq.~\eqref{eq:28} finishes the proof of claim $(i)$. 

One can verify that the non-linearities $\{g_s\}_{0 \leq s \leq t}$ defined in this way satisfy the conditions of Theorem \ref{thm:nonseparatble-SE}, thus the asymptotics of OAMP can be characterized by state evolution. As for the proof of claim $(ii)$, again we consider two cases. If $x_t = 0$, then $\E\langle \br_t, \br_t \rangle / n \rightarrow 0$, and $\E\langle \br_t, \btheta \rangle / n \rightarrow 0$. Using the state evolution for OAMP \eqref{eq:OAMP1}, we obtain that $(ii)$ holds with $\bnu^{t + 1} = \bzero$. If $x_t = 1$, then again by state evolution for OAMP, claim $(ii)$ holds with $\bnu^{t + 1} = \alpha_{t + 1} \btheta + \bz_{t + 1}$ where
\begin{align}
	\alpha_{t + 1} = \lim_{n \rightarrow \infty} \frac{\E\langle \btheta, \Pi_{\mathcal{S}_{t - 1}}^{\perp}(\bh_t)  \rangle}{\sqrt{n}\E[\langle \Pi_{\mathcal{S}_{t - 1}}^{\perp}(\bh_t), \Pi_{\mathcal{S}_{t - 1}}^{\perp}(\bh_t)  \rangle]^{1/2 }},
\end{align}
thus completes the proof by induction. 
\end{proof}

\subsection{Optimal orthogonal AMP}

Following the same reasoning of Remark \ref{rmk:xt}, in the following we will restrict 
to the cases in which $x_t = 1$ for all $t \in \NN$. 

Combining Lemma \ref{lemma:GFOM-AMP1} and \ref{lemma:AMP-OMAP1},
we conclude that it is sufficient to lower bound the error
 of OAMP algorithms. The following corollary is a direct consequence of the proceeding results,
 and extends Eq.~\eqref{eq:Reduction-AMP-NonSep}.
\begin{corollary}
	Under the assumptions of Setting \ref{setting:Gaussian}, recall that
	 $\mathcal{A}_{\mathrm{GFOM}}^t$ denotes the
	  class of uniformly Lipschitz GFOM estimators with $t$ iterations, 
	 and denote by $\mathcal{A}_{\mathrm{OAMP}}^t$ the class of OAMP estimators
	  with $t$ iterations (i.e., AMP estimators whose state evolution yields $\Sigma_{s,t}= \bfone_{s=t}$).
	 
 Then we have
	\begin{align}\label{eq:29}
		\inf_{\hat\btheta(\,\cdot\,) \in \mathcal{A}_{\mathrm{GFOM}}^t} 
		\pliminf_{n\to\infty}\frac{1}{n}\big\|\hat{\btheta}(\bX, \bu)- \btheta\big\|^2_2
		=  \inf_{\hat\btheta(\,\cdot\,) \in \mathcal{A}_{\mathrm{OAMP}}^t} 
		\pliminf_{n\to\infty}\big\|\hat{\btheta}(\bX, \bu)- \btheta\big\|^2_2\, .
	\end{align}
\end{corollary}
Notice that a sufficient statistics for $\btheta$ given $\balpha_{\leq t} \btheta + \bz_{\leq t}$ is $T_0 := \|\balpha_{\leq t}\|_s \btheta + \bz$ with $\bz \overset{d}{=} \normal(\vec{0}, \id_n)$ independent of $\btheta$. Therefore, in order to derive the minimum of the right hand side of Eq.~\eqref{eq:29}, it is sufficient to compute the maximum value of $\|\balpha_{\leq t}\|_2$, which is provided by the following lemma. The proof of Theorem \ref{thm:main} under Setting \ref{setting:Gaussian} directly follows. 
\begin{lemma}
	Recall that $(\gamma_s)_{s \geq 0}$ is defined in Eq.~\eqref{eq:beta}. Then, for all $t \in \NN$ and all choice of non-linearities $g_0, \cdots, g_t$, we have $\|\balpha_{\leq t}\|_2 \leq \gamma_t$. 
\end{lemma}

\begin{proof}
	The proof is by induction over $t$. For the base case $t = 1$, notice that 
	\begin{align*}
		\sup_{f_0}\frac{\E[\langle f_0(\bu), \btheta \rangle]^2}{n\E[\langle f_0(\bu), f_0(\bu) \rangle]} =  \frac{\E[\langle f_0(\bu), \E[\btheta \mid \bu] \rangle]^2}{n\E[\langle f_0(\bu), f_0(\bu)  \rangle]} \leq \gamma_1^2.
	\end{align*}
	The last step above is via application of Cauchy-Schwarz inequality. Then the base case holds by taking the limit $n \rightarrow \infty$ in Eq.~\eqref{eq:25}. 
	
	We assume that the claim holds for the first $t$ iterations, and we prove by induction that it also holds for iteration $t + 1$. We let $\hat{\btheta}_t := \E[\btheta \mid \br_1, \cdots, \br_t, \bu]$, then
	\begin{align*}
		\frac{\E[\langle \btheta, \Pi_{\mathcal{S}_{t - 1}}^{\perp}(\bh_t)  \rangle]^2}{n\E[\langle \Pi_{\mathcal{S}_{t - 1}}^{\perp}(\bh_t), \Pi_{\mathcal{S}_{t - 1}}^{\perp}(\bh_t)  \rangle]} =&  \frac{\E[\langle \hat{\btheta}_t, \Pi_{\mathcal{S}_{t - 1}}^{\perp}(\bh_t)  \rangle]^2}{n\E[\langle \Pi_{\mathcal{S}_{t - 1}}^{\perp}(\bh_t), \Pi_{\mathcal{S}_{t - 1}}^{\perp}(\by_t)  \rangle]} \\
		\overset{(a)}{\leq} & \, \frac{1}{n}\E[\|\Pi_{\mathcal{S}_{t - 1}}^{\perp}(\hat{\btheta}_t)\|_2^2]\\
		 \overset{(b)}{=}&  \frac{1}{n}\E[\|\hat{\btheta}_t\|_2^2] - \frac{1}{n} \E[\|\Pi_{\mathcal{S}_{t - 1}}(\hat{\btheta}_t)\|_2^2],
	\end{align*}
	where $(a)$ follows from Cauchy-Schwartz inequality and $(b)$ from Pythagora's theorem. Notice that 
	\begin{align*}
		\{\Pi_{\mathcal{S}_{s - 1}}(\bh_s) / \E[\|\Pi_{\mathcal{S}_{s - 1}}(\bh_s)\|_2^2]^{1/2}: 0 \leq s \leq t - 1\} 
	\end{align*}
	is an orthonormal basis for $\mathcal{S}_{t - 1}$, thus,
	\begin{align*}
		\frac{\E[\langle \btheta, \Pi_{\mathcal{S}_{t - 1}}^{\perp}(\bh_t)  \rangle]^2}{n\E[\langle \Pi_{\mathcal{S}_{t - 1}}^{\perp}(\bh_t), \Pi_{\mathcal{S}_{t - 1}}^{\perp}(\bh_t)  \rangle]} \leq \frac{1}{n}\E[\|\hat{\btheta}_t\|_2^2] - \sum_{s = 0}^{t - 1} \frac{\E[\langle \btheta, \Pi^{\perp}_{\mathcal{S}_{s - 1}}(\bh_s)  \rangle]^2}{n\E[\| \Pi^{\perp}_{\mathcal{S}_{s - 1}}(\bh_s)\|_2^2]}.
	\end{align*}
	Taking the limits on both sides of the above inequality gives $\alpha_{t + 1}^2 \leq \E[\|\hat{\btheta}_t\|_2^2] / n - \sum_{s = 1}^t \alpha_s^2$. By induction, 
	\begin{align*}
		\frac{1}{n}\E[\|\hat{\btheta}_t\|_2^2] =&  \frac{1}{n}\E[\|\E[\btheta \mid \br_1, \cdots, \br_t, \bu]\|_2^2] \\
		\overset{(a)}{=} & \frac{1}{n}\E[\|\E[\btheta \mid \|\balpha_{\leq t}\|_2 \btheta + \bz, \bu]\|_2^2] \\
		\overset{(b)}{\leq} & \frac{1}{n}\E[\|\E[\btheta \mid \gamma_{t} \btheta + \bz, \bu]\|_2^2] \\
		\overset{(c)}{=} & \gamma_{t+1}^2,
	\end{align*}
	where $(a)$ follows because $T_0$ is a sufficient statistics for $\btheta$, $(b)$ is by induction hypothesis and Jensen's inequality, and $(c)$ is by the definition of $\gamma_{t + 1}$. This concludes the proof of the lemma. 
\end{proof}

\section{Proof of Theorem \ref{thm:GLM} under Setting \ref{setting:4}}
\label{sec:ProofGLM_1}

In this section we prove Theorem \ref{thm:GLM} under the assumptions of Setting \ref{setting:4}. As in 
Section \ref{sec:proof} in the main text, we will additionally assume $\bX$ has sub-Gaussian
 entries, and relax this assumption in Appendix \ref{app:SubGaussian}. Namely, in this section we assume $\E[\exp(\lambda X_{ij})] \leq \exp(C\lambda^2 / n)$ for all $i \in [n]$, $j \in [d]$ and some $n$-independent constant $C$. 

\subsection{AMP algorithm}
As before, the first step of our proof is to define the class of AMP algorithms
 for the current setting. An AMP algorithm for solving generalized linear models under Setting \ref{setting:4} is defined by a sequence of continuous functions (also known as the non-linearities) $\{f_t: \RR^{t + 2} \rightarrow \RR\}_{t \geq 0}$ and $\{g_t: \RR^{t + 1} \rightarrow \RR\}_{t \geq 1}$, and produces vectors $\{\bb^t\}_{t \geq 1} \subseteq \RR^d$ and $\{\ba^t\}_{t \geq 1} \subseteq \RR^n$ via the following iteration:
\begin{align}\label{eq:AMP1}
\begin{split}
	\left\{ \begin{array}{ll}
		\bb^{t + 1} = \bX^{\top} f_t(\ba^{\leq t}; \by, \bu) - \sum\limits_{s = 1}^t \xi_{t,s} g_s(\bb^{\leq s}; \bv), \\
		\ba^t = \bX g_t(\bb^{\leq t}; \bv) - \sum\limits_{s = 1}^t \eta_{t,s} f_{s - 1}(\ba^{ \leq s - 1}; \by, \bu).
	\end{array} \right.
\end{split}
\end{align}
As before, non-linearities are applied entrywise. 
We denote the Onsager terms by
\begin{align*}
	& \Ons_{\mathrm{AMP}}^t(\ba^{\leq t - 1}; \by, \bu) := \sum_{s = 1}^t \eta_{t,s}f_{s - 1}(\ba^{ \leq s - 1}; \by, \bu), \\
	& \Ons_{\mathrm{AMP}}^{t + 1}(\bb^{\leq t}; \bv) :=  \sum\limits_{s = 1}^t \xi_{t,s} g_s(\bb^{\leq s}; \bv).
\end{align*}
The coefficients $(\xi_{t,s})_{1 \leq s \leq t}$ and $(\eta_{t,s})_{1 \leq s \leq t}$ are deterministic, defined via:
\begin{align}\label{eq:DefOnsCoeffGLM}
\begin{split}
	& \xi_{t,s} = \E\big[{\partial_s} f_t(\bar\bG_{\leq t}; Y, U)\big], \qquad Y:= h( \bar G_0, W)\\
	& \eta_{t,s} = \frac{1}{\delta}\E \big[{\partial_s}g_t(\bmu_{\leq t} \Theta + \bG_{\leq t}; V)\big], 
	\end{split}
\end{align}
 where we use the notations
 $\bar{\bG}_{\leq t} := (\bar{G}_1, \cdots, \bar{G}_t)$,  
 $\bG_{\leq t} := (G_1, \cdots, G_t)$, 
 the joint distributions of $(\bar\bG_{\leq t}, Y, U)$ and of
 $(\bG_{\leq t}, \Theta, V)$ is defined via the  following state evolution recursion
 %
	\begin{align}\label{eq:SE1}
\begin{split}
	& (\bar G_0, \bar{\bG}_{ t}) \sim\normal(\mathbf{0}_{ t + 1}, \bar\bSigma_{\leq t}), \qquad \bG_{\leq t} \sim \normal(\textbf{0}_t,  \bSigma_{\leq t}),  \\
	& \bar\Sigma_{ij}  = \frac{1}{\delta}\E[g_i(\bmu_{\leq i} \Theta + \bG_{\leq i}; V)g_j(\bmu_{\leq j} \Theta + \bG_{\leq j}; V)], \qquad i,j \geq 1, \\
	& \bar\Sigma_{i0} =  \bar\Sigma_{0i} = \frac{1}{\delta}\E[g_i(\bmu_{\leq i} \Theta + \bG_{\leq i}; V)\Theta], \qquad \bar{\Sigma}_{00} = \frac{1}{\delta} \E[\Theta^2], \qquad i \geq 1, \\
	&  \Sigma_{ij} = \E[f_{i - 1}(\bar\bG_{\leq i - 1}; Y, U) f_{j - 1}(\bar\bG_{\leq j - 1}; Y, U)], \qquad i, j \geq 1, \\
	& \mu_{t + 1} = \E\big[ {\partial_{\bar G_0}} f_{t}(\bar\bG_{\leq t}; Y, U) \big].
\end{split}
\end{align}
Here it is understood that $(\Theta,V)\sim\mu_{\Theta,V}$ is independent of $(G_i)_{i\ge 1}$
and $(W,U)\sim\mu_{W,U}$ is independent of  $(\bar G_i)_{i\ge 0}$. 
Further, $\bar{\bSigma}_{\leq t} = (\bar\Sigma_{ij})_{0 \leq i,j \leq t}$, $\bSigma_{\leq t} = (\Sigma_{ij})_{1 \leq i,j \leq t}$ and $\bmu_{\leq t} = (\mu_i)_{1 \leq i \leq t}$. Here, $\partial_s$ refers to the partial derivative with respect to the $s$-th variable, and $\partial_{\bar G_0}$ refers to the partial derivative with respect to $\bar G_0$. To be precise, $\partial_{\bar{G}_0} f_t(\bx_{\leq t}; h(x_0, w), u) = \partial_{x_0} f_t(\bx_{\leq t}; h(x_0, w), u)$. Note that $f_0$ depends only on $(Y, U)$. Thus, the above recursion does not need any specific initialization. After $t$ iterations as in Eq.~\eqref{eq:AMP1}, the AMP algorithm estimates $\btheta$ by applying a Lipschitz function $g_t^{\ast}: \RR^{t + 1} \rightarrow \RR$ row-wise to $(\bb^{\leq t}, \bv)$:
\begin{align*}
	\hat{\btheta}(\bX, \by, \bu, \bv) = g_t^{\ast}(\bb^{\leq t}; \bv).
\end{align*}
The following theorem characterizes the asymptotic performance of the AMP iteration \eqref{eq:AMP1}:
\begin{theorem}\label{thm:SE4}
	Assume the matrix $\bX$ and non-linearities $(f_t, g_t)$ satisfy the same assumptions as $\bX$ and $(F_t^{(1)}, G_t^{(1)})$ under either Setting \ref{setting:4}.$(a)$ or Setting \ref{setting:4}.$(b)$.  Then for any $t \in \NN_{>0}$, and any $\psi: \RR^{t + 2} \rightarrow \RR$ pseudo-Lipschitz of order 2, the AMP iteration \eqref{eq:AMP1} satisfies 
	\begin{align*}
		\plim_{n,d \rightarrow \infty} \frac{1}{d} \sum_{i = 1}^d\psi(\bb_i^{\leq t}, \theta_i, v_i) = \E[\psi(\bmu_{\leq t}\Theta + \bG_{\leq t}, \Theta, V)], \qquad \bG_{\leq t} \sim \normal(\mathbf{0}, \bSigma_{\leq t}).
	\end{align*} 
\end{theorem}

\subsection{Any GFOM can be reduced to an AMP algorithm}

As for the case of low-rank matrix estimation, we first show that any GFOM \eqref{eq:GFOM1} 
can be reduced to an AMP algorithm via a change of variables.
 The proof of the next lemma is very similar to the one 
 of Lemma \ref{lemma:change-of-variable} and we omit it.
\begin{lemma}\label{lemma:change-of-variable1}
Assume the matrix $\bX$ and non-linearities $(F_t^{(1)}, F_t^{(2)}, G_t^{(1)}, G_t^{(2)}, G_{\ast}^{(t)})$ 
satisfy the assumptions of either Setting \ref{setting:4}.$(a)$ or Setting \ref{setting:4}.$(b)$. 
Then there exist functions $\{\varphi_t: \RR^{t + 1} \rightarrow \RR\}_{t \geq 1}$, $\{\bar\varphi_t: \RR^{t + 2} \rightarrow \RR\}_{t \geq 1}$, $\{f_t: \RR^{t + 2} \rightarrow \RR\}_{t \geq 0}$ and $\{g_t: \RR^{t + 1} \rightarrow \RR \}_{t \geq 1}$ 
satisfying the same assumptions such that the following holds. 
Let $\{\ba^t\}_{t \geq 1}$ and $\{\bb^t\}_{t \geq 1}$ be sequences of vectors produced by the 
AMP iteration \eqref{eq:AMP1} with non-linearities $\{f_t\}_{t \geq 0}$ and $\{g_t\}_{t \geq 1}$. 
Then for any $t \in \NN_{>0}$, we have 
	\begin{align*}
		\bu^{\leq t} = \bar\varphi_t(\ba^{\leq t}; \by, \bu), \qquad \bv^{\leq t} =  \varphi_t(\bb^{\leq t}; \bv).
	\end{align*} 
\end{lemma}
Lemma \ref{lemma:change-of-variable1} implies that the class of AMP algorithms achieve the 
same minimum expected error as the class of GFOM for the same number of iterations under any loss.
This is formalized by the next corollary, which is analogous to Corollary \ref{coro:FirstCoro}.
\begin{corollary}\label{cor:GFOM-AMP}
Let $\mathcal{A}_{\mathrm{GFOM}}^t$ be the class of GFOM estimators 
	with $t$ iterations, and $\mathcal{A}_{\mathrm{AMP}}^t$ be the class of AMP algorithms 
	with $t$ iterations (under the assumptions of either Setting \ref{setting:4}.$(a)$, or
	 Setting \ref{setting:4}.$(b)$). 
	 (In particular $\hbtheta(\,\cdot\,)\in \mathcal{A}_{\mathrm{GFOM}}^t$ is defined by a set of
	 $n$-independent functions $\{F_t^{(1)}, F_t^{(2)}, G_{t + 1}^{(1)}, G_{t + 1}^{(2)}, G_{\ast}^{(t + 1)}\}_{t \in \NN}$, and similarly for
	 $\hbtheta(\,\cdot\,)\in \mathcal{A}_{\mathrm{GFOM}}^t$.)
	 
	 Then for any loss function $\calL: \RR^{d}\times\RR^{d}  \rightarrow \RR_{\ge 0}$:
	\begin{align}
		\inf_{\hbtheta(\,\cdot\, ) \in \mathcal{A}_{\mathrm{GFOM}}^t} 
		\pliminf_{n\to\infty}
		\calL(\hbtheta(\bX,\by,\bu,\bv), \btheta)= 
		\inf_{\hbtheta(\,\cdot\,) \in \mathcal{A}_{\mathrm{AMP}}^t} 
			\pliminf_{n\to\infty}
		\calL(\hbtheta(\bX,\by,\bu,\bv), \btheta)\, .\label{eq:FirstCoro-GLM}
	\end{align}
\end{corollary}

\subsection{Orthogonalization}

In this section we show that we can further restrict ourselves to lower bounding the error of
orthogonal AMP (OAMP) algorithms.
\begin{lemma}\label{lemma:ortho}
	Let $\{\ba^t\}_{t \geq 1}$, $\{\bb^t\}_{t \geq 1}$ be sequences produced by the AMP iteration \eqref{eq:AMP1} under either Setting \ref{setting:4}.$(a)$ or Setting \ref{setting:4}.$(b)$. Then there exist functions $\{\phi_t: \RR^{t + 1} \rightarrow \RR^t\}_{t \geq 1}$ satisfying the same assumptions as the non-linearities in the AMP iteration, such that the following holds:
	\begin{enumerate}
		\item[(i)] For all $t \in \NN_{>0}$ we have $\bb^{\leq t} = \phi_t(\bq^{\leq t}; \bv).$
		\item[(ii)] For any $\psi: \RR^{t + 2} \rightarrow \RR$ pseudo-Lipschitz of order 2, 
		\begin{align*}
		\plim_{n,d\to\infty}\frac{1}{d} \sum_{i = 1}^d \psi(q_i^1, \cdots, q_i^t, v_i, \theta_i) 
		= \E[\psi(Q_1, \cdots, Q_t, V, \Theta)],
	\end{align*}
	where $Q_i = x_{i - 1}(\alpha_i \Theta + Z_i)$ with 
	$(x_0, \cdots, x_{t - 1}) \in \{0, 1\}^t$  and $(\alpha_1, \cdots, \alpha_{t}) \in \RR^t$
	deterministic vectors, and 
	$(Z_i)_{i \geq 1} \iidsim \normal(0,1)$ independent of $(\Theta, V)$.
	\end{enumerate} 
\end{lemma}

\begin{proof}
	Given the state evolution of the AMP iteration defined via Eq.~\eqref{eq:SE1}, we let
	\begin{align*}
		Y_t := f_t(\bar\bG_{\leq t}; Y, U), \qquad \mathcal{S}_t = \mbox{span}(Y_k: 0 \leq k \leq t)\, ,
		\;\;\;\; Y = h(\bar G_0;W).
	\end{align*}
	Note that by state evolution, $\E[Y_tY_s] = \Sigma_{t+1,s+1}$. By linear algebra, 
	for all $t \in \NN$, there exist deterministic constants $\{c_{ts}\}_{0 \leq s \leq t}$ and $x_t \in \{0,1\}$, such that $c_{tt} \neq 0$ and
	\begin{align*}
		R_t := c_{tt}\Pi^{\perp}_{\mathcal{S}_{t - 1}}(Y_t) = \sum_{s = 0}^t c_{ts} Y_s, \qquad \E[R_t R_s] = \mathbbm{1}_{s = t} x_t. 
	\end{align*}
	Indeed, proceeding by induction,
	 if $Y_t$ does not belong to $\mathcal{S}_{t - 1}$, then we can take $x_t = 1$ and
	  $c_{tt} = \|\Pi_{\mathcal{S}_{t - 1}}^{\perp}(Y_t)\|_{L^2}^{-1}$. 
	  Otherwise we take $R_t = 0$, $c_{tt} = 1$ and $x_t = 0$. 
	
	We prove the lemma by induction. For the base case $t = 1$, we let $\bq^1 = c_{00} \bb^1$, 
	thus, claim \emph{(i)} follows. As for claim \emph{(ii)}, we consider two cases. If $x_0 = 0$,
	 then $\E[f_0(Y, U)^2] = 0$. By Stein's lemma, $\E[\partial_{\bar{G}_0}f_0(h(\bar{G}_0, W), U)] = \E[\bar{G}_0 f_0(h(\bar{G}_0, W), U)] / \Var[\bar{G}_0] = 0$. 
	 Thus, claim \emph{(ii)} holds with $Q_1 \equiv 0$. If $x_0 = 1$, then $c_{00} = \E[f_0(Y, U)^2]^{1/2}$, and claim $(ii)$ follows from state evolution \eqref{eq:SE1} with
	\begin{align}\label{eq:alpha1-4}
		\alpha_1 = \frac{\E[\partial_{\bar G_0} f_0(h(\bar G_0, W), U)]}{\E[f_0(h(\bar G_0, W), U)^2]^{1/2}} \overset{(a)}{=} \frac{\E[\bar G_0 f_0(h(\bar G_0, W), U)]}{\Var[\bar{G}_0]\E[f_0(h(\bar G_0, W), U)^2]^{1/2}}.
	\end{align}
	where $(a)$ holds by Stein's lemma.
	
	 Suppose the lemma holds for the first $t$ iterations, then we prove it also holds for the $(t + 1)$-th iteration. We let $\bq^{t + 1} = \sum_{s = 0}^t c_{ts} \bb^{s + 1}$. Since $c_{tt} \neq 0$, we can solve for $\bb^{t + 1}$. Thus, we obtain the transformation $\phi_{t + 1}$ that satisfies the desired properties. As a consequence, claim \emph{(i)} follows. 
	
	As for claim \emph{(ii)}, first notice that the mapping
	\begin{align*}
		(b_1, \cdots, b_t, v, \theta) \mapsto \psi(c_{00}b_1, \cdots, \mbox{$\sum_{s = 0}^{t - 1}$}c_{t-1,s} b_{s + 1}, v, \theta)
	\end{align*}
	is pseudo-Lipschitz of order two. Then we consider two cases. In the first case $x_t = 0$, then $R_t \overset{a.s.}{=} 0$. By state evolution \eqref{eq:SE1} and an application of Stein's lemma, we obtain that \emph{(ii)} holds with $Q_{t + 1} \overset{a.s.}{=} 0$. In the second case, $x_t = 1$, then again by the state evolution \eqref{eq:SE1}, $Q_{t + 1} \overset{d}{=} \alpha_{t + 1} \Theta + Z_{t + 1}$, where
	\begin{align}\label{eq:alpha-t+1-4}
		\alpha_{t + 1} = \frac{\E[\partial_{\bar G_0}\Pi^{\perp}_{\mathcal{S}_{t - 1}}(Y_t)]}{\E[\Pi^{\perp}_{\mathcal{S}_{t - 1}}(Y_t)^2]^{1/2}} \overset{(b)}{=} \frac{\E[\bar{G}_0 ^{\perp, t} \Pi^{\perp}_{\mathcal{S}_{t - 1}}(Y_t)]}{\Var[\bar{G}_0^{\perp, t}]\E[\Pi^{\perp}_{\mathcal{S}_{t - 1}}(Y_t)^2]^{1/2}}. 
	\end{align}
	Here, $\bar{G}_0^{\perp, t} = \Pi^{\perp}_{\bar{\mathcal{G}}_t}(\bar G_0)$ with $\bar{\mathcal{G}}_t = \mbox{span}(\bar G_i: 1 \leq  i \leq t)$ and $(b)$ follows from Stein's lemma. Thus, we complete the proof by induction. 
\end{proof}
By similar arguments as discussed in Remark \ref{rmk:xt}, in the following parts of the paper, 
we will set $x_t = 1$ for all $t \in \NN$ without loss of generality. 
	
\subsection{Optimal orthogonal AMP}

 Recall that a sufficient statistics for $\bTheta$ 
	 given $\bS_{\le t} := \balpha_{\le t}\Theta+\bZ_{\le t}$ is 
	 $T_0:=\<\balpha_{\le t},\bS_{\le t}\>/\|\balpha_{\le t}\|_2$, and $T_0$ can be rewritten as: 
	 \begin{align}
	 T_0  = \|\balpha_{\le t}\|_2\Theta+ G\,, \;\;\;\;\;\; G\sim\normal(0,1)\,,\;\;\; G\perp \Theta\, .
	 \label{eq:T0Distr-Bid}
	 \end{align}
	 Further  $\bS_{\le t}$ and $V$ are conditionally independent, given $\Theta$.
	 Hence, the proof of Theorem~\ref{thm:GLM} follows exactly as for Theorem~\ref{thm:main},
	 once we upper bound  the value of $\|\balpha_{\le t}\|_2$ achieved by any OAMP algorithm.
	 Before proving such a bound, we establish some useful identities.
	 \begin{lemma}\label{lemma:conditional-distribution}
	Recall that $(\bar{G}_0,\bar{\bG}_{\leq t}) \sim\normal(\mathbf{0}_{ t + 1}, \bar\bSigma_{\leq t})$,
	where
	\begin{align}
	\bar\Sigma_{ij}  = \frac{1}{\delta}\E[g_i(\phi_i(\balpha_{\leq i} \Theta + \bZ_{\leq i}; V); V)g_j(\phi_j(\balpha_{\leq j} \Theta + \bZ_{\leq j}; V);V)]
	\end{align}
	with $(Z_i)_{i\ge 1}\sim_{i.i.d.}\normal(0,1)$. 
	Further recall that $\bar{G}_0^{\perp, t} = \Pi^{\perp}_{\bar{\mathcal{G}}_t}(\bar G_0)$ with 
	$\bar{\mathcal{G}}_t = \mbox{span}(\bar G_i: 1 \leq  i \leq t)$.
	Define
	\begin{align}
	\omega_t^2 :=  \Var[\bar{G}_0^{\perp, t}], \qquad 
	\zeta_t^2 := \frac{1}{\delta} (\E[\Theta^2] - \omega_t^2). \label{eq:OmegaZetaDef}
\end{align}
Then, the following holds for all $s,t \in \NN$ with $s \leq t$, 
	\begin{align*}
		& \E[\bar{G}_0^{\perp, t} \mid h(\bar{G}_0, W), U, \bar{\bG}_{\leq t}] \overset{d}{=} \E[\omega_t Z_0 \mid h(\omega_t Z_0 + \zeta_t Z_1, W), U, Z_1], \\
		& \E[\bar{G}_0^{\perp, t} \mid h(\bar{G}_0, W), U,\bar{\bG}_{\leq s}] = \frac{\omega_t^2}{\omega_s^2} \E[\bar{G}_0^{\perp, s} \mid  h(\bar{G}_0, W), U, \bar{\bG}_{\leq s} ],
	\end{align*}
	where $Z_0, Z_1 \iidsim \normal(0,1)$,
\end{lemma}
\begin{proof}
We let $\bar{G}_0^{\parallel, t} := \bar{G}_0 - \bar{G}_0^{\perp, t}$, then we can write $\bar{G}_0^{\parallel, t}$ as a deterministic function of $\bar{\bG}_{\leq t}$, and we denote this function by $\bar{G}_0^{\parallel, t} = c_t(\bar{\bG}_{\leq t})$. For $s \leq t$, we observe that $(\bar{G}_0^{\perp, t}, \bar{G}_0^{\parallel, t} - \bar{G}_0^{\parallel, s}, \bar{G}_0^{\parallel, s}) \sim \normal(\mathbf{0}, \diag((\omega_t^2, \omega_s^2 - \omega_t^2  , \zeta_s^2)))$. In the following parts, with a slight abuse of notations, we use $p$ to represent probability density functions for various distributions. Then the following formula regarding the conditional probability density holds:
\begin{align}\label{eq:50}
	& p( \bar{G}_0^{\perp, t} = z \mid h(\bar{G}_0, W) = h, U = u,  \bar{\bG}_{\leq s} = z_{\leq s})\nonumber \\
	\propto & \int p( \bar{\bG}_{\leq s} = z_{\leq s}) p(\bar{G}_0^{\perp, t} = z) \mathbbm{1}\{h(z + c_s(z_{\leq s} ) + y, w) = h\} \mu_{W \mid U = u}(\dd w) \phi(y)\dd y \nonumber\\
	\propto & \int  p(\bar{G}_0^{\perp, t} = z) \mathbbm{1}\{h(z + c_s(z_{\leq s}) + y, w) = h\}\mu_{W \mid U = u}(\dd w)\phi(y)\dd y \nonumber \\
	\propto &  \int p(\bar{G}_0^{\parallel, s} = c_s(z_{\leq s}))  p(\bar{G}_0^{\perp, t} = z) \mathbbm{1}\{h(z + c(z_{\leq t}) + y, w) = h\}\mu_{W \mid U = u}(\dd w)\phi(y)\dd y \nonumber\\
	\propto & \, p( \bar{G}_0^{\perp, t} = z \mid h(\bar{G}_0, W) = h, U = u,  \bar{G}_0^{\parallel, s} = c_s(z_{\leq s})), 
\end{align}
where $\phi$ is the probability density function for $\normal(0, \omega_s^2- \omega_t^2)$. Notice that $(\bar{G}_0^{\perp, t}, \bar{G}_0^{\parallel, t}, U, W) \overset{d}{=} (\omega_tZ_0, \zeta_t Z_1, U, W)$, therefore, we take $s = t$ in \cref{eq:50} and conclude that
\begin{align*}
	\E[\bar{G}_0^{\perp, t} \mid h(\bar{G}_0, W), U, \bar{\bG}_{\leq t}] = \E[\bar{G}_0^{\perp, t} \mid h(\bar{G}_0, W), U, \bar{G}_0^{\parallel , t}]\overset{d}{=} \E[\omega_t Z_0 \mid h(\omega_t Z_0 + \zeta_t Z_1, W), U, Z_1],
\end{align*}
which completes the proof of the first claim. 

As for the second claim, notice that there exists $Z_2, Z_3, Z_4 \iidsim \normal(0,1)$, such that $(\bar{G}_0^{\perp, t}, \bar{G}_0^{\parallel, t} - \bar{G}_0^{\parallel, s}, \bar{G}_0^{\parallel, s}) = (\omega_t Z_2, \sqrt{\omega_s^2 - \omega_t^2}Z_3, \zeta_s Z_4)$. Therefore, using \cref{eq:50}, we have
\begin{align*}
	\E[\bar{G}_0^{\perp, t} \mid h(\bar{G}_0, W), U,\bar{\bG}_{\leq s}] =&  \E[\bar{G}_0^{\perp, t} \mid h(\bar{G}_0, W), U,\bar{G}_0^{\parallel, s}] \\
	= & \E[\omega_t Z_2 \mid h(\omega_t Z_2 + \sqrt{\omega_s^2 - \omega_t^2} Z_3 + \zeta_s Z_4, W), U,Z_4] \\
	\overset{(a)}{=} & \frac{\omega_t^2}{\omega_s^2} \E\big[\omega_t Z_2 + \sqrt{\omega_s^2 - \omega_t^2} Z_3 \mid h(\omega_t Z_2 + \sqrt{\omega_s^2 - \omega_t^2} Z_3 + \zeta_s Z_4, W), U,Z_4\big] \\
	= & \frac{\omega_t^2}{\omega_s^2} \E[\bar{G}_0^{\perp, s} \mid h(\bar{G}_0, W), U, \bar{G}_0^{\parallel, s}] \\
	\overset{(b)}{=} &  \frac{\omega_t^2}{\omega_s^2} \E[\bar{G}_0^{\perp, s} \mid h(\bar{G}_0, W), U, \bar{\bG}_{\leq s}],
\end{align*}
where $(a)$ is by Lemma \ref{lemma:cauchy}, and $(b)$ is by \cref{eq:50}. Thus, we complete the proof of the lemma.

\end{proof}

	The next lemma proves the desired upper bound on $\|\balpha_{\leq t}\|_2$.
\begin{lemma}\label{lemma:FinalLemmaRegression}
	Recall the definition of $\{\beta_t\}$ in Eq.~\eqref{eq:beta1}. Then for all $t \in \NN_{>0}$ 
	and all AMP algorithms we have $\|\balpha_{\leq t}\|_2 \leq \beta_{t}$. 
\end{lemma}
\begin{proof}
Recall the definition of $\omega_t$, $\zeta_t$ in Eq.~\eqref{eq:OmegaZetaDef},
and of $(\sigma_t)_{t \in \NN_{>0}}$  in Eq.~\eqref{eq:beta1}.
We will prove the following claims by induction over $t$:  $\|\balpha_{\leq t}\|_2 \leq \beta_{t}$
and $\omega_{t - 1} \geq \sigma_{t}$. 

For the base case $t = 1$, $\omega_0 \geq \sigma_{1}$ holds by definition. Using Eq.~\eqref{eq:alpha1-4} we have
	\begin{align*}
	\alpha_1^2 = \frac{\E[\bar G_0 f_0(h(\bar G_0, W), U)]^2 }{\Var[\bar G_0]^2 \E[f_0(h(\bar G_0, W), U)^2]} \leq \sup_{X \in \sigma\{h(\bar G_0, W), U\}}  \frac{\E [\bar G_0X]^2}{\Var[\bar G_0]^2 \E[X^2]} \leq  \frac{1}{\sigma_1^2}\E[\E[Z_0 \mid h(\sigma_1 Z_0, W), U]^2],
\end{align*}
where $Z_0 \sim \normal(0,1)$ and the last step follows from Cauchy-Schwarz inequality.

Next we assume the induction claim holds for the first $t$ iterations, and we prove 
 it holds for the $(t + 1)$-th iteration. 
 Notice that the random 
variables $\{Y_0 / \E[Y_0^2]^{1/2}, \cdots, \Pi^{\perp}_{\mathcal{S}_{t - 1}}(Y_t) / \E[\Pi^{\perp}_{\mathcal{S}_{t - 1}}(Y_t)^2]^{1/2}\}$ are orthonormal. Then we have:
\begin{align*}
	\alpha_{t + 1}^2 =&  \frac{\E[\E[\bar{G}_0^{\perp, t}\mid h(\bar{G}_0, W), \bar{\bG}_{\leq t}, U ]\, \Pi^{\perp}_{\mathcal{S}_{t - 1}}(Y_t)]^2}{\omega_t^4\E[\Pi^{\perp}_{\mathcal{S}_{t - 1}}(Y_t)^2]} \\
	\overset{(a)}{\leq} & \frac{1}{\omega_t^4}\E[\E[\bar{G}_0^{\perp, t}\mid h(\bar{G}_0, W), \bar{\bG}_{\leq t}, U]^2] - \sum_{s = 0}^{t - 1}\frac{\E[\bar{G}_0^{\perp, t}\Pi^{\perp}_{\mathcal{S}_{s - 1}}(Y_s)]^2}{\omega_t^4\E[\Pi^{\perp}_{\mathcal{S}_{s - 1}}(Y_s)^2]} \\
	\overset{(b)}{= } & \frac{1}{\omega_t^2} \E[\E[Z_0 \mid h(\omega_t Z_0 + \zeta_t Z_1, W), U, Z_1]^2] - \sum_{s = 0}^{t - 1}\frac{\E[\bar{G}_0^{\perp, s}\Pi^{\perp}_{\mathcal{S}_{s - 1}}(Y_s)]^2}{\omega_s^4\E[\Pi^{\perp}_{\mathcal{S}_{s - 1}}(Y_s)^2]} \\
	\overset{(c)}{\leq} & \frac{1}{\sigma_{t + 1}^2}\E[\E[Z_0 \mid h(\sigma_{t + 1}Z_0 + \tilde{\sigma}_{t + 1}Z_1, W), U, Z_1]^2] - \sum_{s = 1}^t \alpha_s^2,
\end{align*}
where $(a)$ holds by Eq.~\eqref{eq:alpha-t+1-4} and Pythagora's theorem, $(b)$ by Lemma \ref{lemma:conditional-distribution}, and 
$(c)$ is by induction hypothesis and Lemma \ref{lemma:monotone}. The last inequality above gives $\sum_{s = 1}^{t + 1}\alpha_s^2 \leq \beta_{t + 1}^2$. 

For $t \in \NN_{>0}$ we define 
\begin{align*}
	Y_t' := g_t(\phi_t(\balpha_{\leq t} \Theta + \bZ_{\leq t}; V); V), \qquad \mathcal{S}_t' := \mbox{span}(Y'_i: 1 \leq i \leq t).
\end{align*}
By state evolution \eqref{eq:SE1}, $\omega_{t + 1}^2 = \E[\Pi^{\perp}_{\mathcal{S}_{t + 1}'}(\Theta)^2] / \delta$. Further we have
\begin{align*}
	\omega_{t + 1}^2 \overset{(d)}{=} & \frac{1}{\delta}  \E[\Theta^2] - \frac{1}{\delta}  \E[\Pi_{\mathcal{S}_{t + 1}'}(\Theta)^2] \\
	\overset{(e)}{\geq} & \frac{1}{\delta}  \E[\Theta^2] - \frac{1}{\delta}  \E[\E[\Theta \mid \balpha_{\leq t + 1} \Theta + \bZ_{\leq t + 1}, V]^2] \\
	\overset{(f)}{=} & \frac{1}{\delta}  \E[\Theta^2] - \frac{1}{\delta}  \E[\E[\Theta \mid \|\balpha_{\leq t + 1}\|_2 \Theta + G, V]^2] \\
	\overset{(g)}{\geq} & \frac{1}{\delta} \E[\Theta^2] - \frac{1}{\delta}\E[\E[\Theta \mid \beta_{t + 1}\Theta + G, V]^2] = \sigma_{t + 2}^2,
\end{align*}
where $(d)$ holds by Pythagora's theorem, $(e)$ by Jensen's inequality, $(f)$
  by property of sufficient statistics and $(g)$ is by induction hypothesis and Jensen's inequality.

This completes the proof of the lemma by induction. 
\end{proof}

\begin{lemma}\label{lemma:monotone}
	Let $Z_0, Z_1 \iidsim \normal(0,1)$. For any fixed  $\omega_0^2 \ge 0$, 
	 the following function is non-increasing in $a\in (0,\omega_0^2])$:
	\begin{align*}
		a \mapsto \frac{1}{a^2} \E[\E[Z_0 \mid h(aZ_0 + (\omega_0^2 - a^2)^{1/2} Z_1, W), U, Z_1]^2].
	\end{align*}
\end{lemma}
\begin{proof}
	For $\delta > 0$, we introduce the decomposition $Z_1 = \delta Z_2 + \sqrt{1 - \delta^2} Z_3$, with  $Z_2, Z_3 \iidsim \normal(0,1)$ that are independent of $Z_0$. Then by Jensen's inequality, 
	\begin{align*}
		& \frac{1}{a^2} \E[\E[Z_0 \mid h(aZ_0 + (\omega_0^2 - a^2)^{1/2} Z_1, W), U, Z_1]^2] \\
		 = &  \frac{1}{a^2} \E[\E[Z_0 \mid h(aZ_0 + (\omega_0^2 - a^2)^{1/2} \delta Z_2 + ((\omega_0^2 - a^2)(1 - \delta^2))^{1/2}Z_3, W), U, Z_2, Z_3]^2] \\
		 \geq & \frac{1}{a^2} \E[\E[Z_0 \mid h(aZ_0 + (\omega_0^2 - a^2)^{1/2} \delta Z_2 + ((\omega_0^2 - a^2)(1 - \delta^2))^{1/2}Z_3, W), U,  Z_3]^2] \\
		 = & \frac{1}{a^2 + \delta^2(\omega_0^2 - a^2)} \E[\E[Z_0 \mid h((a^2 + \delta^2(\omega_0^2 - a^2))^{1/2}Z_0 + ((\omega_0^2 - a^2)(1 - \delta^2))^{1/2}Z_3, W), U,  Z_3]^2]. 
	\end{align*}
	The above inequality holds for all $\delta \in [0,1]$, thus completes the proof of the lemma. 
\end{proof}

\begin{lemma}\label{lemma:cauchy}
	We let $Z_1, Z_2$ be independent mean-zero Gaussian random variables with variance  $\sigma_1^2$ and $\sigma_2^2$, respectively. For $\sigma_1^2 \geq q \geq 0$, we let $G_q$ be a mean-zero Gaussian random variable such that $\Cov(G_q, Z_2) = 0$ and $\Var(G_q) = \Cov(G_q, Z_1) = q$. Then for all $h: \RR^2 \rightarrow \RR$, we have
	\begin{align*}
		f_h(q) := \E[G_q \mid h(Z_1 + Z_2, W), Z_2] = \frac{q}{\sigma_1^2} \E[Z_1 \mid h(Z_1 + Z_2, W), Z_2].
	\end{align*} 
\end{lemma}
\begin{proof}
	For $q_1, q_2 \geq 0$ with $q_1 + q_2 \leq \sigma_1^2$, there exist $G_{q_1}, G_{q_2}$ independent of each other, and satisfy the above constraints. Then, we have $\Cov(G_{q_1} + G_{q_2}, Z_2) = 0$, $\Cov(G_{q_1} + G_{q_2}, Z_1) = \Var(G_{q_1} + G_{q_2}) = q_1 + q_2$. Therefore, 
	\begin{align*}
		f_h(q_1 + q_2) = \E[G_{q_1} + G_{q_2} \mid h(Z_1 + Z_2, W), Z_2] = f_h(q_1) + f_h(q_2).
	\end{align*}
	For all fixed $(h(Z_1 + Z_2, W), Z_2)$, $f_h$ is continuous, thus the lemma follows from Cauchy's equation. 
\end{proof}

\section{Proof of Theorem \ref{thm:GLM} under Setting \ref{setting:3}}
\label{sec:ProofGLM_2}

In this section we prove Theorem \ref{thm:GLM} under the assumptions of Setting \ref{setting:3}. 
%
\subsection{AMP algorithm}

As in previous proofs, we start with the definition of AMP algorithms with non-separable non-linearities. 
Under Setting \ref{setting:3}, an AMP algorithm for solving generalized linear models is defined by a sequence of uniformly Lipschitz functions $\{f_t: \RR^{n(t + 2)} \rightarrow \RR^n\}_{t \geq 0}$ and $\{g_t: \RR^{d(t + 1)} \rightarrow \RR^d\}_{t \geq 1}$, and produces $\{\bb^t\}_{t \geq 1} \subseteq \RR^d$ and $\{\ba^t\}_{t \geq 1} \subseteq \RR^n$ via the following iteration:
\begin{align}\label{eq:AMP3}
\begin{split}
	\left\{ \begin{array}{ll}
		\bb^{t + 1} = \bX^{\top} f_t(\ba^{\leq t}; \by, \bu) - \sum\limits_{s = 1}^t \xi_{t,s} g_s(\bb^{\leq s}; \bv), \\
		\ba^t = \bX g_t(\bb^{\leq t}; \bv) - \sum\limits_{s = 1}^t \eta_{t,s} f_{s - 1}(\ba^{ \leq s - 1}; \by, \bu).
	\end{array} \right.
\end{split}
\end{align}
Here, $(\xi_{t,s})_{1 \leq s \leq t}$ and $(\eta_{t,s})_{1 \leq s \leq t}$ are deterministic coefficients defined via 
\begin{align}\label{eq:SE3-Ons-Def}
\begin{split}
	& \xi_{t,s} = \frac{1}{n}\sum_{i = 1}^n\E\big[{\partial_{i,s}} f_{t,i}(\bar\bg_{\leq t}; \by_{\ast}, \bu)\big], \qquad \by_{\ast}:= h( \bar \bg_0, \bw)\\
	& \eta_{t,s} = \frac{1}{n}\sum_{i = 1}^d\E \big[{\partial_{i,s}}g_{t,i}(\bmu_{\leq t} \btheta + \bg_{\leq t}; \bv)\big]. 
	\end{split}
\end{align}
Here we introduced the notations 
$\bar{\bg}_{\leq t} := (\bar{\bg}_1, \cdots, \bar{\bg}_t) \in \RR^{n \times t}$, 
$\bg_{\leq t} := (\bg_1, \cdots, \bg_t) \in \RR^{d \times t}$, and the joint distributions of 
$(\btheta,\bv,(\bg_{i})_{i\ge 1})$ and of $(\by_{\ast},\bu,\bw,(\bar\bg_{i})_{i\ge 0})$ 
are determined by the following state evolution recursions
%
	\begin{align}\label{eq:SE3}
\begin{split}
	& (\bar \bg_0, \bar{\bg}_{\leq t})\sim \normal(\mathbf{0}, \bar\bSigma_{\leq t + 1} \otimes \id_n ), \qquad \bg_{\leq t} \sim \normal(\textbf{0},  \bSigma_{\leq t} \otimes\id_d), \\
	& \bar\Sigma_{ij}  = \lim_{n,d \rightarrow \infty} \frac{1}{ n}\E[g_i(\bmu_{\leq i} \btheta + \bg_{\leq i}; \bv)^{\top}g_j(\bmu_{\leq j} \btheta + \bg_{\leq j}; \bv)],  \qquad i, j \geq 1,\\
	& \bar\Sigma_{i0} = \bar\Sigma_{0i} = \lim_{n,d \rightarrow \infty} \frac{1}{ n}\E[g_i(\bmu_{\leq i} \btheta + \bg_{\leq i}; \bv)^{\top}\btheta], \qquad \bar{\Sigma}_{00} = \frac{1}{\delta} \E[\Theta^2], \qquad i \geq 1. \\
	&  \Sigma_{ij} = \lim_{n,d \rightarrow \infty} \frac{1}{n} \E[f_{i - 1}(\bar\bg_{\leq i - 1}; \by_{\ast}, \bu)^{\top} f_{j - 1}(\bar\bg_{\leq j - 1}; \by_{\ast}, \bu)], \\
	& \mu_{t + 1} = \lim_{n,d \rightarrow \infty}\frac{1}{n}\sum_{i = 1}^n\E\big[  {\partial_{\bar g_{0,i}}} f_{t,i}(\bar\bg_{\leq t}; \by_{\ast}, \bu) \big].
\end{split}
\end{align}
In the above equations $\bSigma_{\leq t} = (\Sigma_{ij})_{1 \leq i,j \leq t}$, $\bar\bSigma_{\leq t} = (\bar\Sigma_{ij})_{0 \leq i, j \leq t}$ and $\bmu_{\leq t} = (\mu_i)_{1 \leq i \leq t}$, and the limits are assumed to exist. Here, $\partial_{i,s}$ refers to the partial derivative with respect to the $s$-th variable of the $i$-th row of the input matrix, and $\partial_{\bar{g}_{0,i}}$ refers to the partial derivative with respect to $\bar{g}_{0,i}$. Note that $f_0$ depends only on $(\by_{\ast}, \bu)$, thus, the state evolution does not need any specific initialization. After $t$ iterations as in Eq.~\eqref{eq:AMP3}, the AMP algorithm estimates $\btheta$ by applying a uniformly Lipschitz function $g_t^{\ast}: \RR^{d(t + 1)} \rightarrow \RR^d$ to $(\bb^{\leq t}, \bv)$:
\begin{align*}
	\hat{\btheta}(\bX, \by, \bu, \bv) = g_t^{\ast}(\bb^{\leq t}; \bv).
\end{align*}
The following theorem describes the state evolution of the AMP iteration \eqref{eq:AMP3}.
\begin{theorem}\label{thm:SE3}
	Assume $X_{ij} \iidsim \normal(0, 1/ n)$ for all $i \in [n]$ and $j \in [d]$, $(\theta_i, v_i)_{i \leq d} \iidsim \mu_{\Theta, V}$, $(w_i, u_i)_{i \leq n} \iidsim \mu_{W, U}$, and for all $t \in \NN$, the non-linearities $(f_t, g_{t + 1})$ are uniformly Lipschitz. Furthermore, we assume the following limits exist for all $(\bmu, \bSigma, \bar\bSigma)$:
	\begin{align*}
		& \lim_{n,d \rightarrow \infty} \frac{1}{n}\E[f_t(\bar{\bg}_{\leq t}; \by_{\ast}, \bu)^{\top}f_s(\bar{\bg}_{\leq s}; \by_{\ast}, \bu)],\\
		& \lim_{n,d \rightarrow \infty} \frac{1}{n}\E[f_t(\bar{\bg}_{\leq t}; \by_{\ast}, \bu)^{\top}\bar{\bg}_0], \\
		& \lim_{n,d \rightarrow \infty} \frac{1}{d}\E[g_t(\bmu_{\leq t} \btheta + \bg_{\leq t}; \bv)^{\top} g_s(\bmu_{\leq s} \btheta + \bg_{\leq s}; \bv)], \\
		& \lim_{n,d \rightarrow \infty} \frac{1}{d}\E[g_t(\bmu_{\leq t} \btheta + \bg_{\leq t}; \bv)^{\top} \btheta], \\
		& \lim_{n,d \rightarrow \infty} \frac{1}{d}\E[g_t^{\ast}(\bmu_{\leq t} \btheta + \bg_{\leq t}; \bv)^{\top} g_s^{\ast}(\bmu_{\leq s} \btheta + \bg_{\leq s}; \bv)], \\
		& \lim_{n,d \rightarrow \infty} \frac{1}{d}\E[g_t^{\ast}(\bmu_{\leq t} \btheta + \bg_{\leq t}; \bv)^{\top} \btheta].
	\end{align*}
	Then for $\{\psi_n: \RR^{d(t + 2)} \rightarrow \RR\}_{n \geq 1}$ uniformly pseudo-Lipschitz of order 2, 
	\begin{align*}
		\psi_n(\bb^{\leq t}, \btheta, \bv) = \E[\psi_n(\bmu_{\leq t} \btheta + \bg_{\leq t}, \btheta, \bv)] + o_P(1). 
	\end{align*}
\end{theorem}

\subsection{Any GFOM can be reduced to an AMP algorithm}

Again we show that GFOM \eqref{eq:GFOM1} can be reduced to an AMP algorithm \eqref{eq:AMP3} under Setting \ref{setting:3}. To be specific, we have the following lemma:
\begin{lemma}\label{lemma:GFOM-AMP3}
	Under the assumptions of Setting \ref{setting:3}, for all $t \in \NN_{>0}$, there exist uniformly Lipschitz functions $\varphi_t: \RR^{d(t + 1)} \rightarrow \RR^{dt}$, $\bar\varphi_{t}: \RR^{n(t + 2)} \rightarrow \RR^{nt}$, $f_{t - 1}: \RR^{n(t + 1)} \rightarrow \RR^n$ and $g_t: \RR^{d(t + 1)} \rightarrow \RR^d$ that satisfy the following conditions. We let $\{\ba^t\}_{t \geq 1}$ and $\{\bb^t\}_{t \geq 1}$ be sequences of vectors produced by the AMP iteration \eqref{eq:AMP3} with non-linearities $\{f_t\}_{t \geq 0}$ and $\{g_t\}_{t \geq 1}$. Then for any $t \in \NN_{>0}$, we have
	\begin{align*}
		& \bu^{\leq t} = \bar{\varphi}_t(\ba^{\leq t}; \by, \bu), \qquad \bv^{\leq t} = \varphi_t(\bb^{\leq t}; \bv), \\
		& f_{t - 1}(\ba^{\leq t - 1}; \by, \bu) = F_{t - 1}^{(1)}(\bar{\varphi}_{t - 1}(\ba^{\leq t - 1}; \by, \bu); \by, \bu), \qquad g_t(\bb^{\leq t}; \bv) = G_t^{(1)}(\varphi_t(\bb^{\leq t}; \bv); \bv).  
	\end{align*}
	Furthermore, $\{\varphi_t\}_{t \geq 1}$ and $\{\bar\varphi_t\}_{t \geq 1}$ satisfy the following conditions. For any $(\bmu, \bSigma, \bar\bSigma)$ and $t \in \NN_{>0}$, there exist uniformly bounded $(b_{ij})_{1 \leq j \leq i \leq t}$, $(\bar{b}_{ij})_{1 \leq j \leq i \leq t}$, which are sequences with respect to $n$, such that for $\by_{\leq t}$, $\bar\by_{\leq t}$ as defined in Setting \ref{setting:3}, we have $\bar\by_{\leq t} = \bar\varphi_t(\bar\bg_{\leq t}; \by_{\ast}, \bu)$ and $\by_{\leq t} = \varphi_t(\bmu_{\leq t} \btheta + \bg_{\leq t}; \bv) $.
	
\end{lemma}
\begin{remark}
	For all $t \in \NN_{>0}$, since $(b_{ij})_{1 \leq j \leq i \leq t}$ and $(\bar{b}_{ij})_{1 \leq j \leq i \leq t}$ are uniformly bounded, there exists a subsequence of $\NN_{>0}$, which we denote by $\{n_k\}_{k \in \NN_{>0}}$, such that for all $s,r \leq t$, $b_{s,t}$ and $\bar{b}_{s,r}$ converge to n-independent limits along $\{n_k\}_{k \in \NN_{>0}}$. As a consequence, the following limits exist in probability along the subsequence $\{n_k\}_{k \in \NN_{>0}}$ by the third assumption of Setting \ref{setting:3}:
	\begin{align*}
		& \lim_{n,d \rightarrow \infty} \frac{1}{n}f_t(\bar{\bg}_{\leq t}; \by_{\ast}, \bu)^{\top}f_s(\bar{\bg}_{\leq s}; \by_{\ast}, \bu), \,\,\,\,\,\,\,\,\,\,\,\,\,\,\,\,\,\,\,\,\,\,\,\,\,\, \lim_{n,d \rightarrow \infty} \frac{1}{n}f_t(\bar{\bg}_{\leq t}; \by_{\ast}, \bu)^{\top}\bar{\bg}_0, &\\
		& \lim_{n,d \rightarrow \infty} \frac{1}{d}g_t(\bmu_{\leq t} \btheta + \bg_{\leq t}; \bv)^{\top} g_s(\bmu_{\leq s} \btheta + \bg_{\leq s}; \bv), \,\,\,\, \lim_{n,d \rightarrow \infty} \frac{1}{d}g_t(\bmu_{\leq t} \btheta + \bg_{\leq t}; \bv)^{\top} \btheta, & \\
		& \lim_{n,d \rightarrow \infty} \frac{1}{d}g_t^{\ast}(\bmu_{\leq t} \btheta + \bg_{\leq t}; \bv)^{\top} g_s^{\ast}(\bmu_{\leq s} \btheta + \bg_{\leq s}; \bv), \,\,\,\, \lim_{n,d \rightarrow \infty} \frac{1}{d}g_t^{\ast}(\bmu_{\leq t} \btheta + \bg_{\leq t}; \bv)^{\top} \btheta. & 
	\end{align*}
	As a consequence, the new AMP iteration satisfies all assumptions of Theorem \ref{thm:SE3}, thus, its asymptotics can be characterized by the state evolution \eqref{eq:SE3} along the subsequence.
\end{remark}

\begin{proof}
	We prove the lemma by induction over $t$. For the base case $t = 1$, we set $f_0(\by, \bu) := F_0^{(1)}(\by, \bu)$, $\varphi_1(\bb^1; \bv) := \bb^1 + F_0^{(2)}(\bv)$, $g_1(\bb^1; \bv) := G_1^{(1)}(\varphi_1(\bb^1; \bv); \bv)$ and $\bar\varphi_1(\ba^1; \by, \bu) := \ba^1 + G_1^{(2)}(\by, \bu) + \eta_{1,1} f_0(\by, \bu)$, where $\eta_{1,1}$ is defined via state evolution \eqref{eq:SE3}. Notice that $\eta_{1,1}$ is a function of $n$. By the uniform Lipschitzness assumption, $\eta_{1,1}$ is uniformly bounded as a sequence in $n$. Thus, $\varphi_1, \bar{\varphi}_1$ are uniformly Lipschitz. By definition, $\by^1 = \varphi_1(\mu_1 \btheta + \bg_1; \bv)$ and $\bar{\by}^1 = \bar{\varphi}_1(\bar{\bg}_1; \by_{\ast}, \bu)$ with $\bar{b}_{11} = \eta_{1,1}$, which completes the proof for the base case. 
	
	Next, suppose the lemma holds for the first $t$ iterations, we then prove it holds for the $(t + 1)$-th iteration. By induction hypothesis, 
	\begin{align*}
		& \bv^{t + 1} = \bX^{\top} F_t^{(1)}(\bar{\varphi}_t(\ba^{\leq t}; \by, \bu); \by, \bu) + F_t^{(2)}(\varphi_t(\bb^{\leq t}; \bv); \bv), \\
		& \bu^{t + 1} = \bX G_{t + 1}^{(1)}(\varphi_t(\bb^{\leq t + 1}; \bv); \bv) + G_{t + 1}^{(2)} (\bar{\varphi}_t(\ba^{\leq t}; \by, \bu); \by, \bu). 
	\end{align*} 
	We let $f_t(\bx^{\leq t}; \by, \bu) := F_t^{(1)}(\bar{\varphi}_t(\bx^{\leq t}; \by, \bu); \by, \bu)$ and $g_{t + 1}(\bx^{\leq t+1}; \bv) :=  G_{t + 1}^{(1)}({\varphi}_{t + 1}(\bx^{\leq t + 1}; \bv); \bv)$. The composition of uniformly Lipschitz functions is still uniformly Lipschitz. As a consequence, we can conclude that $f_t, g_{t + 1}$ are uniformly Lipschitz functions. Based on the choice of $\{f_s\}_{0 \leq s \leq t}$ and $\{g_s\}_{1 \leq s \leq t + 1}$, we can compute the coefficients for the Onsager correction terms $\{\xi_{t,s}\}_{1 \leq s \leq t}$ and $\{\eta_{t+1,s}\}_{1 \leq s \leq t + 1}$, which are uniformly bounded as sequences in $n$. 
	
	Then we define $\ba^{t + 1}$, $\bb^{t + 1}$ via the AMP iteration \eqref{eq:AMP3}, which gives
	\begin{align*}
		& \bb^{t + 1} = \bv^{t + 1} -  F_t^{(2)}(\varphi_t(\bb^{\leq t}; \bv); \bv) - \sum_{s = 1}^t \xi_{t,s} G_{s}^{(1)}(\varphi_t(\bb^{\leq s}; \bv); \bv), \\
		& \ba^{t + 1} = \bu^{t + 1} - G_{t + 1}^{(2)} (\bar{\varphi}_t(\ba^{\leq t}; \by, \bu); \by, \bu) - \sum_{s = 1}^{t + 1} \eta_{t + 1, s} F_{s - 1}^{(1)}(\bar{\varphi}_{s - 1}(\ba^{\leq s - 1}; \by, \bu); \by, \bu). 
	\end{align*}
	Solving for $\bu^{t + 1}$ and $\bv^{t + 1}$ leads to the definition of $\varphi_{t + 1}$ and $\bar{\varphi}_{t + 1}$. Furthermore, by setting $b_{ts} = \xi_{t,s}$ and $\bar{b}_{t + 1,s} = \eta_{t + 1,s}$, we have
	\begin{align*}
		& \varphi_{t +1}(\bmu_{\leq t +1} \btheta + \bg_{\leq t + 1}; \bv) \\
		= & (\varphi_{t}(\bmu_{t} \btheta + \bg_{\leq t}; \bv), \mu_{t + 1} \btheta + \bg_{t + 1} + F_t^{(2)}(\varphi_t(\bmu_{\leq t} \btheta + \bg_{\leq t}; \bv); \bv) + \sum_{s = 1}^t \xi_{t,s} G_{s}^{(1)}(\varphi_t(\bmu_{\leq s} \btheta + \bg_{\leq s}; \bv); \bv)) \\
		= & (\by^{\leq t}, \by^{t + 1}), \\
		& \bar{\varphi}_{t + 1}(\bar{\bg}_{\leq t + 1}; \by_{\ast}, \bu) \\
		= & (\bar{\varphi}_{t}(\bar{\bg}_{\leq t}; \by_{\ast}, \bu), \bar{\bg}_{t + 1} + G_{t + 1}^{(2)} (\bar{\varphi}_t(\bar{\bg}_{\leq t}; \by_{\ast}, \bu); \by_{\ast}, \bu) + \sum_{s = 1}^{t + 1} \eta_{t + 1, s} F_{s - 1}^{(1)}(\bar{\varphi}_{s - 1}(\bar\bg_w{\leq s - 1}; \by_{\ast}, \bu); \by_{\ast}, \bu)) \\
		= & (\bar{\by}^{\leq t}, \bar{\by}^{ t + 1}),
	\end{align*}
	thus completes the proof of the lemma by induction. 
\end{proof}
As an immediate consequence of Lemma \ref{lemma:GFOM-AMP3}, Corollary \ref{cor:GFOM-AMP} holds true under Setting \ref{setting:3} as well. 

\subsection{Orthogonalization}

By linear algebra, $\{\bb^t\}_{t \geq 1}$ derived via AMP iteration \eqref{eq:AMP3} can be further reduced to a set of vectors that are approximately orthogonal after subtracting the component along $\btheta$, which leads to the following lemma:
\begin{lemma}\label{lemma:ortho3}
	Let $\{\ba^t\}_{t \geq 1}$, $\{\bb^t\}_{t \geq 1}$ be sequences produced by the AMP iteration \eqref{eq:AMP3} under Setting \ref{setting:3}. Then there exist functions $\{\phi_t: \RR^{d(t + 1)} \rightarrow \RR^{dt}\}_{t \geq 1}$ which are uniformly Lipschitz, such that the following holds:
	\begin{enumerate}
		\item[(i)] For all $t \in \NN_{>0}$, there exist $n$-independent constants $\{c_{ts}\}_{0 \leq s \leq t}$ such that $c_{tt} \neq 0$ and $\bq^{t + 1} = \sum_{s = 0}^t c_{ts} \bb^{s + 1}$. We write $\bq^{\leq t} = \phi_t(\bb^{\leq t})$, and $\phi_t$ as a sequence in $n$ is uniformly Lipschitz. 
		\item[(ii)] For all $t \in \NN_{>0}$, there exist $(x_0, \cdots, x_{t - 1}) \in \{0,1\}^t$ and $(\alpha_1, \cdots, \alpha_t) \in \RR^t$, such that for any $\{\psi_n: \RR^{n(t + 2)} \rightarrow \RR^n\}$ uniformly pseudo-Lipschitz of order 2, 
		\begin{align*}
			\psi_n(\bq^{\leq t}; \btheta, \bv) = \E[\psi_n(\bq^{\leq t}; \btheta, \bv)] + o_P(1), 
		\end{align*}
		where $\bq^i = x_{i - 1}(\alpha_i \btheta + \bz_i)$, with $\{\bz_i\}_{i \geq 1} \iidsim \normal(\mathbf{0}, \id_d)$ independent of $(\btheta, \bv)$.
	\end{enumerate}
\end{lemma}
\begin{proof}
	Recall that $\by_{\ast} = h(\bar\bg_0, \bw)$. Given the state evolution \eqref{eq:SE3} of the AMP iteration, we define
	\begin{align*}
		\bh_t := f_t(\bar{\bg}_{\leq t}; \by_{\ast}, \bu), \qquad \mathcal{S}_t := \mbox{span}(\bh_k: 0 \leq k \leq t).
	\end{align*}
	Note that by state evolution, $\lim_{n,d \rightarrow \infty}\E\langle \bh_t, \bh_s\rangle / n = \Sigma_{s +1, t + 1} $. By linear algebra, for all $t \in \NN$, there exist deterministic constants $\{c_{ts}\}_{0 \leq s \leq t}$ and $x_t \in \{0,1\}$, such that $c_{tt} \neq 0$ and 
	\begin{align*}
		\sum_{i = 0}^t\sum_{j = 0}^s c_{ti} c_{sj} \Sigma_{i+1,j+1} = \mathbbm{1}_{s = t} x_t. 
	\end{align*}
	We define $\br_t := \sum_{s = 0}^t c_{ts} \bh_s$, then $\lim_{n \rightarrow \infty}\E\langle \br_t, \br_s \rangle / n = \mathbbm{1}_{s = t} x_t$ for all $s,t \in \NN$. Next, we prove the lemma by induction. For the base case $t = 1$, we let $\bq^1 = c_{00}\bb^1$, thus, claim $(i)$ follows. As for claim $(ii)$, we consider two cases. In the first case, $x_0 = 0$, then $\E\langle \bh_0, \bh_0 \rangle / n \rightarrow 0$. By state evolution \eqref{eq:SE3},
	\begin{align*}
		\mu_1 \overset{(a)}{=} & \lim_{n,d \rightarrow \infty}\frac{1}{n} \sum_{i = 1}^n \frac{\delta \E[\bar{g}_{0,i} f_{0,i}(h(\bar{\bg}_0, \bw), \bu)]}{\E[\Theta^2]}, \\
		\overset{(b)}{\leq } &  \limsup_{n,d \rightarrow \infty}\frac{1}{\sqrt{n}} \frac{\delta^{1/2}}{\E[\Theta^2]^{1/2}} \E[\|f_0(\by_{\ast}, \bu)\|_2^2]^{1/2} \rightarrow 0,
	\end{align*}
	where $(a)$ holds by Stein's lemma, and $(b)$ holds by Cauchy-Schwartz inequality. Thus, claim $(ii)$ holds with $\bq^1 \equiv \textbf{0}$. In the second case, $x_0 = 1$, whence $c_{00} = \Sigma_{11}^{-1/2}$, and claim $(ii)$ holds by the state evolution \eqref{eq:SE3}. Moreover,  
	\begin{align}\label{eq:38}
		\alpha_{1} =  \lim_{n,d \rightarrow \infty}\frac{1}{\sqrt n}\sum_{i = 1}^n \frac{\E[\partial_{\bar{g}_{0,i}}f_{0,i}(h(\bar{\bg}_0, \bw), \bu)]}{\E[\|f_0(h(\bar{\bg}_0, \bw), \bu)\|_2^2]^{1/2}}.
	\end{align}
	Suppose the lemma holds for the first $t$ iterations, then we prove it holds for the $(t + 1)$-th iteration as well. We let $\bq^{t + 1} = \sum_{s = 0}^t c_{ts} \bb^{s + 1}$, and the definition of $\phi_{t + 1}$ together with claim $(i)$ follows immediately. As for claim $(ii)$, first notice that the following mapping is uniformly Lipschitz of order 2:
	\begin{align*}
		(\bx_1, \cdots, \bx_{t + 1}, \btheta, \bv) \mapsto \psi_n(\phi_{t + 1}(\bx_1, \cdots, \bx_{t + 1}); \btheta, \bv).
	\end{align*}
	Again we consider two cases. In the first case, $x_t = 0$, thus by state evolution \eqref{eq:SE3}, $(ii)$ holds with $\bq^{t + 1} = \mathbf{0}$. In the second case, $x_t = 1$, then again by state evolution recursion, we can set $\bq^{t + 1} = \alpha_{t + 1} \btheta + \bz_{t + 1}$, with 
	\begin{align}\label{eq:39}
		\alpha_{t + 1} = \lim_{n,d \rightarrow \infty} \frac{\sqrt{n}\E[\langle\bar{\bg}_0^{\perp, t}, \Pi^{\perp}_{\mathcal{S}_{t - 1}}(\bh_t) \rangle]}{\E[\|\Pi^{\perp}_{\mathcal{S}_{t - 1}}(\bh_t)\|_2^2]^{1/2}\E[\|\bar{\bg}_0^{\perp, t}\|_2^2]},
	\end{align}
	where $\bar{\bg}_0^{\perp, t} := \Pi^{\perp}_{\bar{\mathcal{G}}_t}(\bar\bg_0)$ with $\bar{\mathcal{G}}_t := \mbox{span}(\bar{\bg}_i: 1 \leq i \leq t)$. Therefore, we complete the proof of the lemma by induction. 
	\end{proof}
	
	\subsection{Optimality analysis}
	As before, we restrict to the case with $x_t = 1$ for all $t \in \NN$. Given $(\bv, \balpha_{\leq t} \btheta + \bg_{\leq t})$, a sufficient statistics of $\btheta$ is $(\bv, \|\balpha_{\leq t}\|_2 \btheta + \bg)$ with $\bg \sim \normal(\mathbf{0}, \id_d)$ independent of $\btheta$. Therefore, by Lemma \ref{lemma:GFOM-AMP3} and \ref{lemma:ortho3}, in order to derive the minimum estimation error achieved by any GFOM with $t$ iterations, it suffices to study the maximum value of $\|\balpha_{\leq t}\|_2$, which leads to the following lemma:
	\begin{lemma}
		For all $t \in \NN_{>0}$ and all AMP iterations \eqref{eq:AMP3}, we have $\|\balpha_{\leq t}\|_2^2 \leq \beta_t^2$. 
	\end{lemma}
	\begin{proof}
	Recall that $\bar{\bg}_0^{\perp, t} := \Pi^{\perp}_{\bar{\mathcal{G}}_t}(\bar\bg_0)$ with $\bar{\mathcal{G}}_t := \mbox{span}(\bar{\bg}_i: 1 \leq i \leq t)$. We define:
	\begin{align*}
		\omega_t^2 := \lim_{n,d \rightarrow \infty}\frac{1}{n}\E[\|\bar{\bg}_0^{\perp, t}\|_2^2], \qquad \zeta_t^2 := \frac{1}{\delta} \E[\Theta^2] - \omega_t^2. 
	\end{align*} 
	The above limit exists by the assumption of the AMP algorithm. Here, we will prove a stronger result. To be precise, we will establish that the following two claims hold for all $t \in \NN^+$: (1) $\omega_{t - 1} \geq \sigma_{t}$; (2) $\|\balpha_{\leq t}\|_2^2 \leq \beta_t^2$. We prove the claims via induction. By definition, $\omega_0 = \sigma_1$.
		Furthermore, by Eq.~\eqref{eq:38},
		\begin{align*}
		\alpha_{1}^2 = & \lim_{n,d \rightarrow \infty}\left\{\frac{1}{\sqrt n}\sum_{i = 1}^n \frac{\E[\partial_{\bar{g}_{0,i}}f_{0,i}(h(\bar{\bg}_0, \bw), \bu)]}{\E[\|f_0(h(\bar{\bg}_0, \bw), \bu)\|_2^2]^{1/2}}\right\}^2 \\
		\overset{(a)}{=} & \lim_{n,d \rightarrow \infty} \frac{\delta^2\E[\langle f_0(h(\bar{\bg}_0, \bw), \bu),  \bar\bg_0 \rangle]^2}{{n}\E[\|f_0(h(\bar{\bg}_0, \bw), \bu)\|_2^2]\E[\Theta^2]^2} \\
		= & \lim_{n,d \rightarrow \infty} \frac{\delta^2\E[\langle f_0(h(\bar{\bg}_0, \bw), \bu),  \E[\bar\bg_0 \mid h(\bar{\bg}_0, \bw), \bu] \rangle]^2}{n\E[\|f_0(h(\bar{\bg}_0, \bw), \bu)\|_2^2]\E[\Theta^2]^2} \\
		\overset{(b)}{\leq} & \lim_{n,d \rightarrow \infty} \frac{\delta^2\E[ \|  \E[\bar\bg_0 \mid h(\bar{\bg}_0, \bw), \bu]\|_2^2 ]}{n\E[\Theta^2]^2} = \beta_1^2,
	\end{align*} 
	where $(a)$ is by Stein's lemma, and $(b)$ is by Cauchy-Schwartz inequality. Then we assume the lemma holds for the first $t$ iterations, and we prove by induction that it also holds for iteration $(t + 1)$. For $t \in \NN_{>0}$, we let
	\begin{align*}
		\bk_t := g_t(\bmu_{\leq t} \btheta + \bg_{\leq t}; \bv), \qquad \mathcal{S}_t' := \mbox{span}(\bk_i: 1 \leq i \leq t).
	\end{align*}
	By the state evolution of the AMP algorithm, $\omega_{t}^2 =  \lim_{n,d \rightarrow \infty}\E[\|\Pi^{\perp}_{\mathcal{S}_{t}'}(\btheta)\|_2^2] / n$. Thus, we have
	\begin{align*}
		\omega_{t}^2 \overset{(d)}{=} & \frac{1}{\delta} \E[\Theta^2] - \lim_{n,d \rightarrow \infty}\frac{1}{n} \E[\|\Pi_{\mathcal{S}_{t}'}(\btheta)\|_2^2] \\
		\overset{(e)}{\geq} & \frac{1}{\delta} \E[\Theta^2] -\lim_{n,d \rightarrow \infty} \frac{1}{n} \E[\|\E[\btheta \mid \balpha_{\leq t} \btheta + \bz_{\leq t}, \bv]\|_2^2] \\
		\overset{(f)}{=} & \frac{1}{\delta} \E[\Theta^2] - \lim_{n,d \rightarrow \infty} \frac{1}{n} \E[\|\E[\btheta \mid \|\balpha_{\leq t} \|_2\btheta + \bz, \bv]\|_2^2] \\
		\overset{(g)}{\geq} & \frac{1}{\delta} \E[\Theta^2] - \frac{1}{\delta} \E[\E[\Theta \mid \beta_{t} \Theta + G, V]^2] = \sigma_{t + 1}^2,
	\end{align*}
	where $(d)$ is by Pythagora's theorem, $(e)$ is by Jensen's inequality, $(f)$
	 is by property of sufficient statistics, and $(g)$ is by induction hypothesis. Thus, we have completed the proof of claim (1). 

Then we prove claim (2). By Eq.~\eqref{eq:39}, 
	\begin{align*}
		\alpha_{t + 1}^2 = & \lim_{n,d \rightarrow \infty} \frac{n\E[\langle\E[\bar{\bg}_0^{\perp, t}\mid \bar{\bg}_{\leq t}, \bu, h(\bar{\bg}_0, \bw) ], \Pi_{\mathcal{S}_{t - 1}}^{\perp}(\bh_t) \rangle]^2}{\E[\|\Pi_{\mathcal{S}_{t - 1}}^{\perp}(\bh_t)\|_2^2]\E[\|\bar{\bg}_0^{\perp, t}\|_2^2]^2} \\
		\overset{(a)}{\leq} & \lim_{n,d \rightarrow \infty} \frac{\E[\|\E[\bar{\bg}_0^{\perp, t}\mid \bar{\bg}_{\leq t}, \bu, h(\bar{\bg}_0, \bw) ]\|_2^2]}{n\omega_t^4} - \lim_{n,d \rightarrow \infty} \sum_{s = 0}^{t - 1} \frac{\E[\langle \Pi^{\perp}_{\mathcal{S}_{s - 1}}(\bh_s), \E[\bar{\bg}_0^{\perp, t} \mid \bar\bg_{\leq s}, \bu, h(\bar{\bg}_0, \bw)  ]  \rangle]^2}{n\omega_t^4\E[\|\Pi^{\perp}_{\mathcal{S}_{s - 1}}(\bh_s)\|_2^2]} \\
		\overset{(b)}{=} & \lim_{n,d \rightarrow \infty} \frac{1}{\omega_t^2} \E[\E[Z_0 \mid h(\omega_tZ_0 + \zeta_t Z_1, W), U, Z_1]^2] - \lim_{n,d \rightarrow \infty} \sum_{s = 0}^{t - 1} \frac{\E[\langle \Pi^{\perp}_{\mathcal{S}_{s - 1}}(\bh_s), \E[\bar{\bg}_0^{\perp, s} \mid \bar\bg_{\leq s}, \bu, h(\bar{\bg}_0, \bw)  ] \rangle]^2}{n \omega_s^4 \E[\| \Pi^{\perp}_{\mathcal{S}_{s - 1}}(\bh_s) \|_2^2]} \\ 
		= & \lim_{n,d \rightarrow \infty} \frac{1}{\omega_t^2} \E[\E[Z_0 \mid h(\omega_tZ_0 + \zeta_t Z_1, W), U, Z_1]^2] - \lim_{n,d \rightarrow \infty}\sum_{s = 0}^{t - 1} \frac{n\E[\langle \bar{\bg}_0^{\perp, s} ,\Pi^{\perp}_{\mathcal{S}_{s - 1}}(\bh_s)\rangle ]^2}{\E[\| \Pi^{\perp}_{\mathcal{S}_{s - 1}}(\bh_s) \|_2^2]\E[\|\bar{\bg}_0^{\perp, s}\|_2^2]^2} \\
		\overset{(c)}{\leq} &  \frac{1}{\sigma_{t + 1}^2}\E[\E[Z_0 \mid h(\sigma_{t + 1}Z_0 + \tilde{\sigma}_{t + 1}Z_1, W), U, Z_1]^2] - \sum_{s = 1}^t \alpha_s^2,
	\end{align*}
	where $(a)$ is by Pythagora's theorem, $(b)$ is by Lemma \ref{lemma:conditional-distribution}, and $(c)$ is by induction hypothesis and Lemma \ref{lemma:monotone}. The last inequality above gives $\sum_{s = 1}^{t + 1} \alpha_s^2 \leq \beta_{s + 1}^2$. 
	 Thus, we have completed the proof of the lemma by induction. 
	\end{proof}
	
\section{Reduction to matrices with sub-Gaussian entries}
\label{app:SubGaussian}

In this section, we show that in order to prove Theorem \ref{thm:main} under 
Setting \ref{setting:Wigner}.$(a)$ (or to prove Theorem \ref{thm:GLM} under Setting 
\ref{setting:4}.$(a)$), it suffices to consider cases where the matrix $\bW$(or $\bX$) has 
sub-Gaussian entries. Here, we prove this claim for Theorem \ref{thm:main} under 
Setting \ref{setting:Wigner}.$(a)$. Proof of the claim for Theorem \ref{thm:GLM} under 
Setting \ref{setting:4}.$(a)$ follows by the same argument, with notational adaptations. 

By assumption, $\E[W_{ij}^4 ] \leq C / n^2$ and $\E[W_{ij}] = 0$. Thus, we claim that for all 
$\epsilon > 0$ and $i,j \in [n]$, there exists decomposition $W_{ij} = W_{ij}^{(1)} + W_{ij}^{(2)}$, 
such that $\E[W_{ij}^{(1)}] = \E[W_{ij}^{(2)}] = 0$, ${\rm ess}\, \sup_n\sqrt{n}|W_{ij}^{(1)}| < \infty$, 
$\sup_nn^2\E[(W_{ij}^{(2)})^4] < \infty$ and $n\Var[W_{ij}^{(2)}] \leq \epsilon$. 
Furthermore, $(W_{ij}^{(1)})_{i< j \leq n}$ are independent and identically distributed 
random variables, and the same property holds for $(W_{ij}^{(2)})_{i<j \leq n}$. 
To prove this claim, we let $\xi_{\epsilon} > 0$ such that $C / \xi_{\epsilon}^2 < \epsilon$. We define
\begin{align*}
	& W_{ij}^{(1)} := W_{ij}\mathbbm{1}_{\sqrt{n}|W_{ij}| \leq \xi_{\epsilon}} - \E[W_{ij}\mathbbm{1}_{\sqrt{n}|W_{ij}| \leq \xi_{\epsilon}}], \\
	& W_{ij}^{(2)} := W_{ij}\mathbbm{1}_{\sqrt{n}|W_{ij}| > \xi_{\epsilon}} - \E[W_{ij}\mathbbm{1}_{\sqrt{n}|W_{ij}| > \xi_{\epsilon}}]. 
\end{align*}
Then $\sqrt{n}|W_{ij}^{(1)}| \leq 2 \xi_{\epsilon}$, $\E[W_{ij}^{(1)}] = \E[W_{ij}^{(2)}] = 0$, $\sup_n n^2\E[(W_{ij}^{(1)})^4] < \infty$ and $\sup_n n^2\E[(W_{ij}^{(2)})^4] < \infty$. Furthermore, $n\Var[W_{ij}^{(2)}] \leq n \E[W_{ij}^2 \mathbbm{1}_{\sqrt{n}|W_{ij}| > \xi_{\epsilon}}] \leq C / \xi_{\epsilon}^2 < \epsilon$, thus completes the proof of the claim.

With the above decomposition, we let $\bW^{(1)} = (W_{ij}^{(1)})_{i,j \leq n}$ and 
$\bW^{(2)} = (W_{ij}^{(2)})_{i,j \leq n}$ be $n \times n$ matrices. 
 By the Bai-Yin law \cite{vershynin2018high}, we have 
$\|\bW^{(2)}\|_{\mbox{\tiny\rm op}} \leq 2\sqrt{\epsilon} + o_P(1)$. If we replace 
$\bW$ with $\bW^{(1)}$ in model definition \eqref{model:spike}, and denote the iterates 
obtained by GFOM \eqref{eq:GFOM} by $\{\tilde{\bu}^t\}_{t \geq 1}$, then we can prove by 
induction that for all $t \in \NN_{>0}$, with probability $1 - o_n(1)$, 
\begin{align*}
	\frac{1}{\sqrt{n}}\|\bu^t - \tilde{\bu}^t\|_2 \leq F(\epsilon, t).
\end{align*}
Here, $F(\epsilon, t) \rightarrow 0$ as $\epsilon \rightarrow 0^+$. The proof is 
via simple application of the Lipschitz assumption and the upper bound of the 
spectral norm of $\bW^{(2)}$ we have just derived. 
Since $\epsilon$ is arbitrary, we conclude that if Theorem \ref{thm:main} holds for 
sub-Gaussian distributions, then it also holds for distributions with bounded fourth moments.

\end{appendices}


\begin{thebibliography}{MLKZ20}

\bibitem[BLM15]{bayati2015universality}
Mohsen Bayati, Marc Lelarge, and Andrea Montanari.
\newblock Universality in polytope phase transitions and message passing
  algorithms.
\newblock {\em The Annals of Applied Probability}, 25(2):753--822, 2015.

\bibitem[BM11]{bayati2011dynamics}
Mohsen Bayati and Andrea Montanari.
\newblock The dynamics of message passing on dense graphs, with applications to
  compressed sensing.
\newblock {\em IEEE Transactions on Information Theory}, 57(2):764--785, 2011.

\bibitem[BMN20]{berthier2020state}
Raphael Berthier, Andrea Montanari, and Phan-Minh Nguyen.
\newblock State evolution for approximate message passing with non-separable
  functions.
\newblock {\em Information and Inference: A Journal of the IMA}, 9(1):33--79,
  2020.

\bibitem[CC17]{chen2017solving}
Yuxin Chen and Emmanuel~J Cand{\`e}s.
\newblock Solving random quadratic systems of equations is nearly as easy as
  solving linear systems.
\newblock {\em Communications on pure and applied mathematics}, 70(5):822--883,
  2017.

\bibitem[CL21]{chen2021universality}
Wei-Kuo Chen and Wai-Kit Lam.
\newblock {Universality of approximate message passing algorithms}.
\newblock {\em Electronic Journal of Probability}, 26(none):1 -- 44, 2021.

\bibitem[CLM16]{cai2016optimal}
T~Tony Cai, Xiaodong Li, and Zongming Ma.
\newblock Optimal rates of convergence for noisy sparse phase retrieval via
  thresholded wirtinger flow.
\newblock {\em The Annals of Statistics}, 44(5):2221--2251, 2016.

\bibitem[CLS15]{candes2015phase}
Emmanuel~J Candes, Xiaodong Li, and Mahdi Soltanolkotabi.
\newblock Phase retrieval via wirtinger flow: Theory and algorithms.
\newblock {\em IEEE Transactions on Information Theory}, 61(4):1985--2007,
  2015.

\bibitem[CMW20]{celentano2020estimation}
Michael Celentano, Andrea Montanari, and Yuchen Wu.
\newblock The estimation error of general first order methods.
\newblock In {\em Conference on Learning Theory}, pages 1078--1141. PMLR, 2020.

\bibitem[DAM17]{deshpande2017asymptotic}
Yash Deshpande, Emmanuel Abbe, and Andrea Montanari.
\newblock Asymptotic mutual information for the balanced binary stochastic
  block model.
\newblock {\em Information and Inference: A Journal of the IMA}, 6(2):125--170,
  2017.

\bibitem[DM14]{deshpande2014information}
Yash Deshpande and Andrea Montanari.
\newblock Information-theoretically optimal sparse pca.
\newblock In {\em Information Theory (ISIT), 2014 IEEE International Symposium
  on}, pages 2197--2201. IEEE, 2014.

\bibitem[DR19]{duchi2019solving}
John~C Duchi and Feng Ruan.
\newblock Solving (most) of a set of quadratic equalities: Composite
  optimization for robust phase retrieval.
\newblock {\em Information and Inference: A Journal of the IMA}, 8(3):471--529,
  2019.

\bibitem[FS20]{fannjiang2020numerics}
Albert Fannjiang and Thomas Strohmer.
\newblock The numerics of phase retrieval.
\newblock {\em Acta Numerica}, 29:125--228, 2020.

\bibitem[JM13]{javanmard2013state}
Adel Javanmard and Andrea Montanari.
\newblock State evolution for general approximate message passing algorithms,
  with applications to spatial coupling.
\newblock {\em Information and Inference: A Journal of the IMA}, 2(2):115--144,
  2013.

\bibitem[LM19]{lelarge2019fundamental}
Marc Lelarge and L{\'e}o Miolane.
\newblock Fundamental limits of symmetric low-rank matrix estimation.
\newblock {\em Probability Theory and Related Fields}, 173(3):859--929, 2019.

\bibitem[MLKZ20]{maillard2020phase}
Antoine Maillard, Bruno Loureiro, Florent Krzakala, and Lenka Zdeborov{\'a}.
\newblock Phase retrieval in high dimensions: Statistical and computational
  phase transitions.
\newblock {\em arXiv:2006.05228}, 2020.

\bibitem[MM18]{mondelli2018fundamental}
Marco Mondelli and Andrea Montanari.
\newblock Fundamental limits of weak recovery with applications to phase
  retrieval.
\newblock In {\em Conference On Learning Theory}, pages 1445--1450. PMLR, 2018.

\bibitem[Mon19]{montanari2019optimization}
Andrea Montanari.
\newblock {Optimization of the Sherrington-Kirkpatrick Hamiltonian}.
\newblock In {\em IEEE Symposium on the Foundations of Computer Science, FOCS},
  November 2019.

\bibitem[MRY18]{montanari2018adapting}
Andrea Montanari, Feng Ruan, and Jun Yan.
\newblock Adapting to unknown noise distribution in matrix denoising.
\newblock {\em arXiv:1810.02954}, 2018.

\bibitem[MV21a]{mondelli2021approximate}
Marco Mondelli and Ramji Venkataramanan.
\newblock Approximate message passing with spectral initialization for
  generalized linear models.
\newblock In {\em International Conference on Artificial Intelligence and
  Statistics}, pages 397--405. PMLR, 2021.

\bibitem[MV21b]{montanari2021estimation}
Andrea Montanari and Ramji Venkataramanan.
\newblock Estimation of low-rank matrices via approximate message passing.
\newblock {\em The Annals of Statistics}, 49(1):321--345, 2021.

\bibitem[MXM19]{ma2019optimization}
Junjie Ma, Ji~Xu, and Arian Maleki.
\newblock {Optimization-Based AMP for Phase Retrieval: The Impact of
  Initialization and $\ell_2$ Regularization}.
\newblock {\em IEEE Transactions on Information Theory}, 65(6):3600--3629,
  2019.

\bibitem[Nes03]{nesterov2003introductory}
Yurii Nesterov.
\newblock {\em Introductory lectures on convex optimization: A basic course},
  volume~87.
\newblock Springer, 2003.

\bibitem[SR14]{schniter2014compressive}
Philip Schniter and Sundeep Rangan.
\newblock Compressive phase retrieval via generalized approximate message
  passing.
\newblock {\em IEEE Transactions on Signal Processing}, 63(4):1043--1055, 2014.

\bibitem[Ver18]{vershynin2018high}
Roman Vershynin.
\newblock {\em High-dimensional probability: An introduction with applications
  in data science}, volume~47.
\newblock Cambridge university press, 2018.

\bibitem[Wal18]{waldspurger2018phase}
Irene Waldspurger.
\newblock Phase retrieval with random gaussian sensing vectors by alternating
  projections.
\newblock {\em IEEE Transactions on Information Theory}, 64(5):3301--3312,
  2018.

\bibitem[WGE17]{wang2017solving}
Gang Wang, Georgios~B Giannakis, and Yonina~C Eldar.
\newblock Solving systems of random quadratic equations via truncated amplitude
  flow.
\newblock {\em IEEE Transactions on Information Theory}, 64(2):773--794, 2017.

\end{thebibliography}
\end{document}